\documentclass[11pt]{article}

\usepackage[T1]{fontenc}
\usepackage[utf8]{inputenc}
\usepackage{amsmath,amssymb,tikz}
\usepackage{pict2e}
\usepackage{xcolor}
\usepackage[protrusion=true,expansion=false]{microtype}
\usetikzlibrary{decorations.markings}
\usepackage[unicode=true,colorlinks=true,linkcolor=blue,citecolor=blue,urlcolor=blue]{hyperref}
\hypersetup{pdfstartview=FitH, linktoc=all}
\usepackage{comment}

\advance\textheight by \topskip

\numberwithin{equation}{section}

\def\beq{\begin{equation}}
\def\eeq{\end{equation}}

\def\bit{\begin{itemize}}
\def\eit{\end{itemize}}

\makeatletter
\def\eqalign#1{\null\vcenter{\def\\{\cr}\openup\jot\m@th
\ialign{\strut$\displaystyle{##}$\hfil&$\displaystyle{{}##}$\hfil
\crcr#1\crcr}}\,}
\makeatother

\newcommand{\ii}{\textrm{i}}
\newcommand{\dd}{\textrm{d}}
\newcommand{\Ai}{\textrm{Ai}}

\def\bigO{{\cal O}}
\newenvironment{proof}%
{\rm \trivlist \item[\hskip \labelsep{\bf Proof. }]}%
{\hspace*{\fill}$\Box$\endtrivlist}

\begin{document}
\tikzset{middlearrow/.style={
        decoration={markings,
            mark= at position 0.6 with {\arrow{#1}} ,
        },
        postaction={decorate}
    }
}
 \def\ds{\displaystyle}
    \def\tr{{\rm tr \,}}
    \def\Re{{\rm Re \,}}
    \def\Im{{\rm Im \,}}
    \def\Ai{{\rm Ai \,}}
    \def\I{{\rm I \,}}
    \def\II{{\rm II \,}}
    \def\III{{\rm III \,}}
    \def\IV{{\rm IV \,}}
    \def\bigO{{\cal O}}
    \def\Supp{{\rm Supp}}
    \def\Res{{\rm Res}}
 \newtheorem{theorem}{Theorem}[section]
    \newtheorem{lemma}[theorem]{Lemma}
    \newtheorem{corollary}[theorem]{Corollary}
    \newtheorem{proposition}[theorem]{Proposition}
    \newtheorem{conjecture}{Conjecture}
    \newtheorem{Definition}[theorem]{Definition}
    \newenvironment{definition}{\begin{Definition}\rm}{\end{Definition}}
    \newtheorem{Remark}[theorem]{Remark}
    \newenvironment{remark}{\begin{Remark}\rm}{\end{Remark}}
    \newtheorem{Example}[theorem]{Example}
    \newenvironment{example}{\begin{Example}\rm}{\end{Example}}
    \newtheorem{Assumptions}[theorem]{Assumptions}
    \newenvironment{assumptions}{\begin{Assumptions}\rm}{\end{Assumptions}}
 \renewenvironment{proof}%
    {\rm \trivlist \item[\hskip \labelsep{\bf Proof. }]}%
    {\hspace*{\fill}$\Box$\endtrivlist}
    \newenvironment{varproof}%
    {\rm \trivlist \item[\hskip \labelsep{\bf Proof}]}%
    {\hspace*{\fill}$\Box$\endtrivlist}

\def\Xint#1{\mathchoice
{\XXint\displaystyle\textstyle{#1}}%
{\XXint\textstyle\scriptstyle{#1}}%
{\XXint\scriptstyle\scriptscriptstyle{#1}}%
{\XXint\scriptscriptstyle\scriptscriptstyle{#1}}%
\!\int}
\def\XXint#1#2#3{{\setbox0=\hbox{$#1{#2#3}{\int}$ }
\vcenter{\hbox{$#2#3$ }}\kern-.59\wd0}}
\def\ddashint{\Xint=}
\def\dashint{\Xint-}
\providecommand{\PVint}{\dashint}
\providecommand{\calO}{\mathcal{O}}
\newcommand{\unusedbib}[1]{\textcolor{gray}{#1}}

\title{Asymptotics of Hankel determinants for potentials with singular edge points}
\author{Dan Dai\footnotemark[1], \ Jia-Hao Du\footnotemark[2] \ and Chenhao Lu\footnotemark[3]}

\renewcommand{\thefootnote}{\fnsymbol{footnote}}
\footnotetext[1]{Department of Mathematics, City University of Hong Kong, Tat Chee
Avenue, Kowloon, Hong Kong. E-mail: \texttt{dandai@cityu.edu.hk}}

\footnotetext[2]{Department of Mathematics, City University of Hong Kong, Tat Chee
Avenue, Kowloon, Hong Kong. E-mail: \texttt{jiahaodu@cityu.edu.hk}}

\footnotetext[3]{Department of Mathematics, City University of Hong Kong, Tat Chee
Avenue, Kowloon, Hong Kong. E-mail: \texttt{chenhaolu3-c@my.cityu.edu.hk}}
\date{}

\maketitle

\begin{abstract}
We establish the large-$n$ asymptotics of Hankel determinants for unitary random matrix ensembles possessing a higher-order edge singularity. We focus on ensembles where the equilibrium measure is supported on a single interval and the limiting eigenvalue density vanishes to order $2k+\frac{1}{2}$ for $k \in \mathbb{N}$. Notably, we explicitly evaluate the constant term in the asymptotic expansion, which involves a regularized integral of the Hamiltonian associated with the Painlev\'e I ($P_{\rm I}^{2k}$) hierarchy. As a by-product, we also prove the universality of the eigenvalue correlation kernel near this singular edge and derive a limiting kernel expressed through functions related to a special solution of the $P_{\rm I}^{2k}$ equation. Our method relies on the Deift-Zhou nonlinear steepest descent analysis for the Riemann-Hilbert problem of orthogonal polynomials.

\end{abstract}



\section{Introduction and main results}
\subsection{Unitary random matrix ensembles}
    \label{subsection: unitary ensembles}

Let $\mathcal{H}_n$ denote the space of $n\times n$ Hermitian matrices. For $n\in\mathbb{N}$ and $\mathbf{t}=(t_1,\ldots,t_{2k})\in\mathbb{R}^{2k}$, we consider the unitary random matrix ensemble defined by 
\begin{equation}\label{random matrix model}
\frac{1}{\mathcal {Z}_{n,\mathbf{t}}}e^{-n\,\tr V_{\mathbf{t}}(M)}dM,
\end{equation}
where $dM$ is the Lebesgue measure on $\mathcal{H}_n$. Here, $\mathcal{Z}_{n,\mathbf{t}}$ is the normalization constant, also commonly referred to as the partition function. The confining potential $V_{\mathbf{t}}$ is assumed to be a real-analytic function depending on the parameters $t_j \in \mathbb{R}$ for $j=1,\dots,2k$, satisfying the asymptotic condition
\begin{equation} \label{conditionV}
    \lim_{x\to\pm\infty} \frac{V_{\mathbf{t}}(x)}{\log(x^2+1)} = +\infty,
        \qquad \mbox{uniformly for $t_j\in[-\delta_0, \delta_0]$ for some $\delta_0>0$.}
\end{equation}
It is well-known, see e.g. \cite{Mehta}, that the eigenvalue correlation kernel of the ensemble \eqref{random matrix model} is given by the following orthogonal polynomial kernel
\begin{equation} \label{kernel}
    K_n^{(\mathbf{t})}(x,y)=
        e^{-\frac{n}{2}V_{\mathbf{t}}(x)} e^{-\frac{n}{2}V_{\mathbf{t}}(y)}
        \sum_{k=0}^{n-1} p_k^{(\mathbf{t})}(x) p_k^{(\mathbf{t})}(y),
\end{equation}
where
\begin{equation}\label{def of OP}
    p_k^{(\mathbf{t})}(x)=\kappa_k^{(\mathbf{t})} x^k + \cdots,
    \qquad\qquad \mbox{$\kappa_k^{(\mathbf{t})}>0$,}
\end{equation}
is the orthonormal polynomial with respect to the varying weights $e^{-nV_{\mathbf{t}}}$ on $\mathbb R$.

As $n \to \infty$, the macroscopic behavior of the eigenvalues is governed by potential theory. Specifically, the limiting mean eigenvalue distribution coincides with the equilibrium measure in the external field $V_{\mathbf{t}}$  (cf. \cite{DKMVZ2}). This is defined as the unique probability measure on $\mathbb{R}$ that minimizes the logarithmic energy \cite{SaTo}
\begin{equation}\label{eq:energyf}
I_{V_{\mathbf{t}}}(\mu)=\iint \log |x-y|^{-1} \,d\mu(x)d\mu(y) + \int V_{\mathbf{t}}(x)\,d\mu(x)
\end{equation}
among all probability measures $\mu$ on $\mathbb R$. It is well known (see, for instance, \cite{SaTo}) that the equilibrium measure $\mu_{V_{\mathbf{t}}}$ is characterized by the following Euler-Lagrange variational conditions: there exists a constant $\ell_{\mathbf{t}}\in\mathbb{R}$ such that
\begin{align}
    \label{variationalcondition:must-equality1}
    & 2\int \log |x-u|d\mu_{V_{\mathbf{t}}}(u)-V_{\mathbf{t}}(x)=\ell_{\mathbf{t}},
        &\mbox{for $x\in \mathbb S_{\mathbf{t}}$,}
    \\[1ex]
    \label{variationalcondition:must-inequality1}
    & 2\int \log |x-u|d\mu_{V_{\mathbf{t}}}(u)-V_{\mathbf{t}}(x)\leq \ell_{\mathbf{t}},
        &\mbox{for $x\in \mathbb R\setminus \mathbb S_{\mathbf{t}}$,}
\end{align}
where $\mathbb S_{\mathbf{t}}$ denotes the support of $\mu_{V_{\mathbf{t}}}$. Let $\rho_{\mathbf{t}}(x)$ denote the density of $\mu_{V_{\mathbf{t}}}$. In the context of the unitary ensemble, this density can be recovered from the scaled limit of the one-point correlation function:
\begin{equation}\label{psit-kernel}
    \rho_{\mathbf{t}}(x)=\lim_{n\to\infty}\frac{1}{n}K_n^{(\mathbf{t})}(x,x).
\end{equation}
Furthermore, because the confining potential $V_{\mathbf{t}}$ is real-analytic, the density takes a specific algebraic form. There exists a real-analytic function $Q_{\mathbf{t}}$ such that (cf. \cite{DeiKriMcL})
\begin{equation}\label{definition: qst}
    \rho_{\mathbf{t}}(x)=\frac{1}{\pi}\sqrt{Q_{\mathbf{t}}^-(x)},
\end{equation}
where $Q_{\mathbf{t}}^-$ denotes the negative part of $Q_{\mathbf{t}}$ (i.e., $Q_{\mathbf{t}}=Q_{\mathbf{t}}^+-Q_{\mathbf{t}}^-$, with $Q_{\mathbf{t}}^\pm\geq 0$ and $Q_{\mathbf{t}}^+Q_{\mathbf{t}}^-=0$). The growth condition \eqref{conditionV} ensures that $Q_{\mathbf{t}}(x)\to +\infty$ as $x\to\pm\infty$; consequently, the support $\mathbb S_{\mathbf{t}}$ consists of a finite union of bounded intervals.

The external field $V_{\mathbf{t}}$ is called regular if the strict inequality holds in \eqref{variationalcondition:must-inequality1}, the density $\rho_{\mathbf{t}}$ does not vanish in the interior of the support $\mathbb S_{\mathbf{t}}$, and $Q_{\mathbf{t}}$ has a simple zero at each endpoint of $\mathbb S_{\mathbf{t}}$. If any of these conditions fail, $V_{\mathbf{t}}$ is called singular. Following \cite{DeiKriMcL,KM}, singular points $x^*$ are classified into three types:
\begin{itemize}
    \item[(i)] $x^*\in\mathbb{R}\setminus \mathbb S_{\mathbf{t}}$ is a type $\I$
        singular point if equality in (\ref{variationalcondition:must-inequality1})
        holds. 
    \item[(ii)] $x^*\in\mathbb S_{\mathbf{t}}$ is a type II singular point if it is an
        interior point of $\mathbb S_{\mathbf{t}}$ where $Q_{\mathbf{t}}$
         has a zero of multiplicity $4k$.
    \item[(iii)] $x^*$ is a type III singular point if it is an endpoint
        of the support $\mathbb S_{\mathbf{t}}$ where $Q_{\mathbf{t}}$ has a zero of multiplicity $4k+1$.
\end{itemize}

In this paper, we consider external fields $V_{\mathbf{t}}$ such that, in the unperturbed critical case $t_1=...=t_{2k}=0$, the potential $V$ has a type III
singular edge point $x^*$ for an arbitrary integer $k\ge 1$. Consequently, the associated density $\rho(x)$ satisfies
\begin{equation}
    \rho(x)\sim c|x-x^*|^{\frac{4k+1}{2}},\qquad\mbox{as $x\to x^*$.}
\end{equation}
for some constant $c > 0$. Furthermore, we assume the potential $V_{\mathbf{t}}$ takes the specific form
\begin{equation}
    V_{\mathbf{t}}=V+\sum_{j=1}^{2k}t_jV_j,
\end{equation}
where $V_j$, $j=1,\ldots,2k$, are real-analytic functions satisfying additional conditions that we specify in Section \ref{subsection: statement of results} below.

\subsection{Asymptotics of Hankel determinants}
\label{subsection: hankel-determinants}
Next, we turn to the primary focus of this paper: the large-$n$ asymptotics of the Hankel determinants associated with the perturbed potential $V_{\mathbf{t}}$. Recall the partition function $\mathcal{Z}_{n,\mathbf{t}}$, which normalizes the unitary random matrix ensemble introduced in \eqref{random matrix model}. Then, we introduce the corresponding weight function
\begin{equation}\label{weight_general}
 w_{\mathbf{t}}(x)=e^{-nV_{\mathbf{t}}(x)},
\end{equation}
The Hankel determinant associated with this weight is defined as
\begin{equation}\label{Hankel determinants 1}
 D_n(V_{\mathbf{t}})=\det\left(\int_{\mathbb R}x^{j+k-2}w_{\mathbf{t}}(x)\,dx\right)_{j,k=1}^{n}.
\end{equation}
By Heine's identity \cite{Szego}, this determinant admits the exact multiple-integral representation
\begin{equation}\label{Hankel determinants}
 D_n(V_{\mathbf{t}})=\frac{1}{n!}\int_{\mathbb R^n}\prod_{1\leq j<\ell\leq n}(x_\ell-x_j)^2\prod_{j=1}^{n}w_{\mathbf{t}}(x_j)\,dx_j.
\end{equation}
The above formula indicates that $D_n(V_{\mathbf{t}})$ is proportional to the partition function $\mathcal{Z}_{n,\mathbf{t}}$ up to an $n$-dependent constant. Consequently, evaluating the large-$n$ asymptotics of the Hankel determinant $D_n(V_{\mathbf{t}})$ is equivalent to determining the asymptotic behavior of the partition function for these singular ensembles.

A substantial literature \cite{BWW,BI,BG1,Charlier,CFWW,CGML,EML,Mehta,XDZ,ZXZ,ZhaoCD} describes the large-$n$ asymptotics of $D_n(V)$ for various potentials. We briefly recall the results most closely related to our work. The simplest example arises in the context of the Gaussian unitary ensemble. It is a classical result (see, e.g., \cite[equation (3.3.10)]{Mehta}) that the asymptotics of the Hankel determinant for a quadratic potential are given by
\begin{multline}\label{asymGUE}
\log D_n\left(2\sigma x^2\right) 
=-n^2\left(\frac{3}{4}-\frac{1}{2}\log \frac{1}{4\sigma}\right)+n\log (2\pi)-\frac{1}{12} \log n+\zeta'(-1)+\mathcal O(n^{-1}),
\end{multline}
uniformly for $\sigma$ in compact subsets of $(0,+\infty)$ as $n\to \infty$, where $\zeta'(-1)$ is the derivative of the Riemann-zeta function at $-1$.

More generally, it was proven in \cite{DeiKriMcL} that when $V$ is real analytic and satisfies \eqref{conditionV}, the equilibrium measure takes the form
 $d\mu_V(x)=\frac{\psi_V(x)}{i\pi}(R_+^{\frac{1}{2}}(x)) dx$, where $R(z)=\prod_{j=1}^{m}(z-a_j)(z-b_j)$. Here, the branches are chosen such that $R^{\frac{1}{2}}(z)$ is analytic in $\mathbb{C} \setminus\bigcup_{j=1}^m[a_j,b_j]$ and $R^{\frac{1}{2}}(z)\sim z^m$ as $z\to\infty$.  For the one-cut case ($m=1$), this problem was studied in \cite{BWW,BI,BG1,Charlier,EML}. By \cite[Proposition 5.5]{BWW} or \cite[Theorem 1.1]{Charlier}, we have
\begin{multline} \label{eq:one-cut asy}
\log D_n(V)=-n^2 I_V(\mu_V)+n\log 2\pi-\frac{1}{12} \log n\\ +\zeta'(-1)-\frac{1}{24}\log \left(\frac{\tilde \psi(a_1)\tilde \psi(b_1)|b_1-a_1|^3}{2^6}\right)+\mathcal O(n^{-1}),
\end{multline}
as $n\to \infty$, where $I_V$ is defined in \eqref{eq:energyf} and
\begin{equation} \label{tildepsi}
\tilde  \psi(q):=\lim_{\lambda\to q} {\pi} \left| \frac{\psi_{V}(\lambda)}{\left(\lambda-q\right)^{1/2}}\right|,
\end{equation} 
for $q\in \{a_1,b_1\}$. In particular, for the Gaussian potential, $I_V({\mu_V})$ can be calculated explicitly, recovering the leading term shown in \eqref{asymGUE}.

For the two-cut case ($m=2$), Claeys, Grava and McLaughlin \cite[equation (1.9)]{CGML} obtained the following asymptotic expansion:
\begin{multline}\label{formCGML}
\log D_n(V)=-n^2 I_V(\mu_V)+n\log (2\pi)-\frac{1}{6}\log n+ \log \theta(n\Omega)\\+2 \zeta'(-1) -\frac{1}{2}\log \frac{K(\mathrm k)}{\pi} -\frac{1}{24}\sum_{q\in\{a_j,b_j\}_{j=1}^2} \log \tilde \psi(q)\\+\frac{1}{8}\log (b_2-b_1)(a_2-a_1)
 -\frac{1}{8}\sum_{l,j=1}^2\log|b_j-a_l|+\mathcal O(n^{-1}),
\end{multline}
as $n\to \infty$. Here, $\Omega=\int_{a_2}^{b_2}d\mu_V$, the elliptic modulus is $\mathrm k =\sqrt{\frac{(a_2-b_1)(b_2-a_1)}{(b_2-b_1)(a_2-a_1)}}$, and $K(\mathrm k)=\int_0^1 \frac{dx}{\sqrt{(1-x^2)(1-\mathrm k^2x^2)}}$ is the complete elliptic integral of the first kind. The function $\tilde \psi$ is given by \eqref{tildepsi} for $q\in \{a_j,b_j\}_{j=1}^2$, and $\theta(\cdot |\tau)$ is the Riemann theta function.

For a general number of cuts $m\ge2$,  Charlier et al. \cite{CFWW}  studied the case where $V$ is a regular multi-cut potential, deriving the corresponding large-$n$ asymptotics of the Hankel determinant:
\begin{multline}
\log D_n(V)=-n^2I_V(\mu_V)+n\log(2\pi)-\frac{m}{12}\log n+
\log\frac{\theta(n \mathbf{\Omega})}{\theta(0)}\\
\hspace{-2.5cm} +\frac{m}{4}\log2+m\zeta'(-1)
-\frac1{24}\sum_{q\in\{a_j,b_j\}_{j=1}^{m}}\log\widetilde\psi(q)
\\
+\frac18\sum_{1\leq l<j\leq m}\log\{(b_j-b_l)(a_j-a_l)\}
-\frac18\sum_{l,j=1}^{m}\log|b_j-a_l|+\bigO(n^{-1}), 
\end{multline}
as $n\to \infty$, where $\mathbf{\Omega}$ is a vector with components $\Omega_j=\int_{a_{j+1}}^{b_m}d\mu_V$.

When the potential $V$ possesses the singularities discussed at the end of Section \ref{subsection: unitary ensembles}, the only known result is the work of Bleher and Its \cite{BI} concerning type II singularities. They studied the singular potential $V(x)=\frac{1}{4}x^4-x^2$ and its corresponding deformation $V_t(x)=\frac{1}{4t^2}x^4+(1-\frac{2}{t})x^2$, obtaining the following double-scaling asymptotics for the Hankel determinant:
\begin{equation}
   \log(D_n(V_t))=\log(D_n(x^2)) - n^2F_n^{\rm reg}(t) - F_n^{\rm sing}(t)+\bigO(n^{-\frac{1}{3}+\epsilon}),
\end{equation}
for every $0<\epsilon<\frac{1}{12}$ as $n\to\infty$, in the double-scaling regime where the parameter ${s}=(t-1)2^{\frac{2}{3}}n^{\frac{2}{3}}$ is bounded ($|s|<C$). 
Here
\[
F_n^{\mathrm{reg}}(t)
=
F(t)+n^{-2}F^{(2)}(t)
\]
is the truncation of the one-cut regular expansion
\cite[(8.20), (9.65)]{BI}, with
\[
F(t)
=
\int_t^\infty
\frac{t-\tau}{\tau^2}
\left[
\frac{2\tau^2}{9}
\left(
2-\tau+\sqrt{(2-\tau)^2+3}
\right)^2
-\frac12
\right]d\tau,
\]
and
\[
F^{(2)}(t)
=
\frac{1}{12}
\int_t^\infty
(\tau-t)
\frac{
\left(
2-\tau+\sqrt{(2-\tau)^2+3}
\right)
\left(
2(2-\tau)+5\sqrt{(2-\tau)^2+3}
\right)
}{
\left((2-\tau)^2+3\right)^2
}
\,d\tau.
\]
Furthermore, the singular contribution is given by
\begin{equation}
    F_n^{\mathrm{sing}}(t) =-\log F_{\rm TW}\left((t-1)2^\frac{2}{3}n^\frac{2}{3}\right)=-\log F_{\rm TW}\left(s
    \right).
\end{equation}
In this scaling limit, $F_n^{\mathrm{sing}}(t)$ represents an $\mathcal{O}(1)$ contribution as $n \to \infty$.  It is interesting to note that this $\bigO(1)$-term involves the cumulative distribution function $F_{\rm TW}(s)$ of the Tracy-Widom distribution \cite{TW94}. Specifically, $F_{\rm TW}(s)$ is defined as $F_{\rm TW}(s)=\exp{\left(\int_s^{\infty}(s-x)u^2(x)dx\right)}$, where $u(x)$ is the Hastings-McLeod solution \cite{HastingsMcLeod} to the Painlev\'e II equation 
\begin{equation}
    u''(x)=xu(x)+2u^3(x),
\end{equation}
characterized by the boundary condition $u(x) \sim \Ai(x)$ as $x \to +\infty$. Recently, the asymptotics of partition functions for 2D random matrix models associated with planar orthogonal polynomials have attracted considerable attention \cite{ACC2026,ByunKangSeoYang2025,ByunSeoYang2025,ByunYangYoo2026,DMMS}.  In particular, in the critical regime of the complex Ginibre ensemble, the Tracy-Widom distribution $F_{\rm TW}(s)$ also appears in the constant term of the asymptotic expansion of the partition function \cite{ByunSeoYang2025}.

While the analysis of type II interior singularities was successfully carried out in \cite{BI}, the asymptotic behavior of Hankel determinants for potentials with type III edge singularities has remained an open problem. In this paper, we fill this gap by rigorously deriving the large-$n$ asymptotics for random matrix ensembles exhibiting a type III higher-order edge singularity.

\subsection{Universality in random matrix theory}
Beyond the asymptotics of Hankel determinants, our analysis naturally yields results concerning the universality of the eigenvalue correlation kernel. At a regular edge point $x^*$, it is a well-known result \cite{Deift,DKMVZ2} that the local eigenvalue correlations are governed by soft-edge universality. More precisely, there exists a constant $c>0$ such that
\begin{equation} \label{eq:Airy-kernel}
    \lim_{n\to\infty} \frac{1}{cn^{\frac{2}{3}}}
    K_n\left(x^*+\frac{u}{cn^{\frac{2}{3}}},x^*+\frac{v}{cn^{\frac{2}{3}}}\right)=
    \frac{\Ai(u)\Ai'(v)-\Ai(v)\Ai'(u)}{u-v},
\end{equation}
where $\Ai$ is the Airy function. For a comprehensive survey of universality phenomena across unitary, orthogonal, and symplectic ensembles, we refer the reader to Kuijlaars \cite{Kuijlaars_survey}.

Near spectral singularities,  non-standard universality classes emerge. For the type III singularities considered in the present paper, the density vanishes at the edge point $x^*$ to order $2k+\frac{1}{2}$, meaning
\[
    \rho(x)\sim (x-x^*)^{2k+\frac{1}{2}},\qquad \mbox{as $x\to x^*$}.
\]
It was conjectured in the physics literature \cite{BB,BMP} that, instead of the Airy kernel \eqref{eq:Airy-kernel}, the limiting eigenvalue correlation kernel in this regime takes the form
 \begin{equation}\label{eq:limiting kernel}
        \lim_{n \to \infty}
        \frac{1}{cn^{\frac{2}{4k+3}}}K_n^{(\mathbf t)}\left(x^*+\frac{u}{cn^{\frac{2}{4k+3}}}, x^*+\frac{v}{cn^{\frac{2}{4k+3}}}\right)
        =K^{(2k)}(u,v;s,\boldsymbol\tau),
    \end{equation}
for a certain constant $c>0$, uniformly for $u$ and $v$ in compact subsets of $\mathbb{R}$. Here, the limiting kernel is built out of functions associated with a special solution to the $2k$-th member of the Painlev\'e I hierarchy. More precisely, we have 
\begin{equation}\label{def:limKer}
 K^{(2k)}(u,v;s,\boldsymbol\tau)
 :=\frac{\Psi_1^{(2k)}(u;s,\boldsymbol\tau)\Psi_{2}^{(2k)}(v;s,\boldsymbol\tau)-
 \Psi_1^{(2k)}(v;s,\boldsymbol\tau)\Psi_{2}^{(2k)}(u;s,\boldsymbol\tau)}{-2\pi i (u-v)},
\end{equation}
where the functions $\Psi_1^{(2k)}(\zeta;s,\boldsymbol\tau)$ and $\Psi_2^{(2k)}(\zeta;s,\boldsymbol\tau)$, with $\boldsymbol\tau=(\tau_1,\dots,\tau_{2k-1})$, arise from the following Lax pair (cf. \cite{Claeys}):
\begin{equation}\label{Lax-pair}
\frac{\partial \Psi}{\partial \zeta}(\zeta;s,\boldsymbol\tau)=A(\zeta;s,\boldsymbol\tau)\Psi(\zeta;s,\boldsymbol\tau),\qquad \frac{\partial \Psi}{\partial s}(\zeta;s,\boldsymbol\tau)=L(\zeta;s,\boldsymbol\tau)\Psi(\zeta;s,\boldsymbol\tau).
\end{equation}
In this Lax pair, $A$ and $L$ are polynomials in $\zeta$ of degrees $2k+1$ and $1$, respectively.

For $k=1$, the universality limit \eqref{eq:limiting kernel} was rigorously established by Claeys and Vanlessen \cite{ClaeysVan}. In this paper, as a by-product of our steepest descent analysis, we prove \eqref{eq:limiting kernel} for general $k \in \mathbb{N}$ under the multi-scaling limit where $n\to\infty$ and $t_j\to 0$ for $j=1,\ldots,2k$.

\subsection{$\Psi$-functions associated with a special solution of the $P_{\rm I}^{2k}$ equation}
    \label{subsection: PI2 equation}
To state our main results, we first introduce the Painlev\'{e}~I hierarchy, denoted by $\mathrm{P_{I}^m}$ (cf. \cite{Gordoa,Kud97,Mugan,Shim04}). The $m$-th member of this hierarchy is a nonlinear ordinary differential equation of order $2m$, defined by   
\begin{equation}\label{def:PIm}
	s+\mathcal{L}_m(q)+\sum_{j=1}^{m-1}\tau_j\mathcal{L}_{j-1}(q)=0,\qquad \tau_1,\ldots,\tau_{m-1}\in\mathbb{R},
\end{equation}
where the operators $\mathcal{L}_k$ are generated by the Lenard-Magri recursion relation:
\begin{align}\label{LOdef}
	\begin{cases}
		\frac{d}{ds}\mathcal{L}_{k+1}(q)=\bigg(\frac14 \frac{d^3}{ds^3}+2q\frac{d}{ds}+q_s \bigg)\mathcal{L}_{k}(q), \quad k=0,\ldots,m-1,\\
		\mathcal{L}_{0}(q)=4q,\quad \mathcal{L}_{j}(0)=0, \quad j=1,\ldots,m.
	\end{cases}
\end{align}
If $m=1$, equation \eqref{def:PIm} reduces to the classical Painlev\'{e} I equation $q_{ss}=-6q^2-s$.

The second member of the hierarchy ($m=2$) takes the form
\begin{equation}\label{eq-PI2-this-paper}
q_{ssss}+4s+40q^3+10q_{s}^2+20q q_{ss}+16\tau_{1}q=0.
\footnote{After the rescalings
 $U=60^{\frac{2}{7}}q$, $X=60^{-\frac{1}{7}}s$, and $T=-4\times 60^{-\frac{3}{7}}\tau_1$,
 this equation reduces to
 \begin{equation*}\label{eq-PI2-previous-paper}
  \frac{1}{240}U_{XXXX}+\frac{1}{24}(U_{X}^2+2UU_{XX})+
  \frac{1}{6}U^3+X-TU=0,
 \end{equation*}
 which is the $\mathrm{P_{I}^2}$ equation studied in the literature \cite{Claeys,ClaeysVan2,Grava-Kapaev-Klein-2015}.}
\end{equation}

The relevance to our work is the even member of the Painlev\'{e} I hierarchy. It was established in \cite{Claeys} that for each $\mathrm{P_{I}^{2k}}$ equation, there exists a unique real and pole-free solution $q(s)=q(s,\tau_1,\dots,\tau_{2k-1})$ satisfying the asymptotic boundary condition
\begin{equation}\label{eq:asyq}
	q(s)=\mp\frac{1}{2}\alpha_k^{-\frac{1}{2k+1}}|s|^{\frac{1}{2k+1}}
 +\mathcal{O}\left(|s|^{-\frac{1}{2k+1}}\right),
 \qquad s\to\pm\infty.
\end{equation}
Moreover, the corresponding Hamiltonian $h(s)=h(s,\tau_1,\dots,\tau_{2k-1})$ of the $\mathrm{P_{I}^{2k}}$ equation, related to 
$q$ through the identity $\partial_sh=-q$, admits the following asymptotic expansion (cf. \cite{DLXYZ}):
\begin{equation}\label{eq:asy of Ham}
    h(s)=\frac{2k+1}{4k+4}\alpha_k^{-\frac{1}{2k+1}}|s|^\frac{2k+2}{2k+1}+\frac{ks}{12(2k+1)(s^2+1)}+\bigO(|s|^{-\frac{8k+5}{4k+2}}),\qquad s\to\pm\infty,
\end{equation}
where
\begin{equation}\label{eq:def alpha k}
    \alpha_k=\frac{2\Gamma(2k+\frac{3}{2})}{\Gamma(2k+2)\Gamma(\frac{3}{2})}.
\end{equation}
For $\mathrm{P_{I}^{2}}$, this distinguished solution is the well-known tritronqu\'{e}e solution studied in \cite{Grava-Kapaev-Klein-2015}. It is worth noting that this same hierarchy governs the universal critical behavior of Hamiltonian PDEs \cite{Claeys,Dub06,Dub08,Dub09}, including the Korteweg-de Vries hierarchy \cite{CG09,CG12}.

The distinguished $\mathrm{P_{I}^{2k}}$ solution and its corresponding Hamiltonian can be characterized via the following Riemann-Hilbert (RH) problem \cite{Claeys,ClaeysItsK}:
\subsubsection*{RH problem for $\Psi$:}
Fix
\[
 \vartheta=\frac{(4k+2)\pi}{4k+3},
 \qquad
 \Gamma_1=\mathbb R_+,
 \quad \Gamma_2=e^{i\vartheta}\mathbb R_+,
 \quad \Gamma_3=\mathbb R_-,
 \quad \Gamma_4=e^{-i\vartheta}\mathbb R_+,
 \qquad \Gamma=\bigcup_{j=1}^4\Gamma_j.
\]
The ray $\Gamma_1$ is oriented from $0$ to $+\infty$, while $\Gamma_2$, $\Gamma_3$, and $\Gamma_4$ are oriented from infinity towards $0$.  The fractional powers in the normalization at infinity are taken with the principal branch cut along $\Gamma_3=\mathbb R_-$.

\begin{itemize}
    \item[(a)] The function $\Psi: \mathbb{C} \setminus \Gamma \mapsto \mathbb{C}^{2 \times 2}$ is analytic and remains bounded as $\zeta\to0$.
    \item[(b)] $\Psi$ satisfies the following jump relations on
    $\Gamma$,
    \begin{align}
        \label{RHP Psi: b1}
        &\Psi_+(\zeta)=\Psi_-(\zeta)
        \begin{pmatrix}
            0 & 1 \\
            -1 & 0
        \end{pmatrix},& \mbox{for $\zeta\in\Gamma_3$,} \\[1ex]
        \label{RHP Psi: b2}
        &\Psi_+(\zeta)=\Psi_-(\zeta)
        \begin{pmatrix}
            1 & 1 \\
            0 & 1
        \end{pmatrix},& \mbox{for $\zeta\in\Gamma_1$,} \\[1ex]
        \label{RHP Psi: b3}
        &\Psi_+(\zeta)=\Psi_-(\zeta)
        \begin{pmatrix}
            1 & 0 \\
            1 & 1
        \end{pmatrix},& \mbox{for $\zeta\in\Gamma_2\cup \Gamma_4$.}
    \end{align}
    \item[(c)] $\Psi$ has the following behavior at infinity,
    \begin{equation}\label{RHP Psi: c}
        \Psi(\zeta)=\zeta^{-\frac{1}{4}\sigma_3}N\left(I-\frac{h}{\zeta^{\frac{1}{2}}}\sigma_3
        +\frac{1}{2 \zeta}\begin{pmatrix}h^2 & iq\\-iq &
        h^2\end{pmatrix} +\bigO(\zeta^{-\frac{3}{2}})\right)
        e^{-\theta(\zeta;s,\tau_1,...,\tau_{2k-1})\sigma_3},
    \end{equation}
    where 
\begin{equation}\label{definition: N}
    N=\frac{1}{\sqrt{2}}\begin{pmatrix}
        1 & 1\\
        -1 & 1
    \end{pmatrix}e^{-\frac{1}{4}\pi i\sigma_3},
\end{equation}    \begin{equation}\label{definition: theta}
\theta(\zeta;s,{\tau_1},...,\tau_{2k-1})=\frac{4}{4k+3}\zeta^{\frac{4k+3}{2}}+\sum_{j=1}^{2k-1}\frac{4}{2j+1}\tau_j\zeta^{\frac{2j+1}{2}}+s\zeta^{\frac{1}{2}},
\end{equation}
and $q=q(s,\tau_1,...,\tau_{2k-1})$ is the special solution of the $P_{\rm I}^{2k}$ equation (\ref{def:PIm}), where $\frac{\partial h}{\partial s}=-q$.
\end{itemize}
The functions $\Psi_1$ and $\Psi_2$ appearing in \eqref{def:limKer} are the
analytic extensions of $\Psi_{11}$ and $\Psi_{21}$ from the sector between
$\Gamma_1$ and $\Gamma_2$ to the entire complex plane.

\subsection{Statement of results}
    \label{subsection: statement of results}
   
In this paper, we work under the following assumptions.
\begin{assumptions}\label{assumptions}
    \
    \begin{itemize}
    \item[(i)] We consider external fields $V_{\mathbf{t}}$ of the form
        \begin{equation}\label{Vst}
            V_{\mathbf{t}}(x)=V(x)+\sum_{j=1}^{2k}t_jV_j(x),
        \end{equation}
        where $V$ and $V_j,j=1,...,2k$  are even, real-analytic functions. Furthermore, we assume 
        there exists a $\delta_0>0$ such that the condition \eqref{conditionV} holds.

    \item[(ii)] The equilibrium measure of $V$ is supported on a single interval $[-A,A]$, and is of the form  $d\mu_{V}(x)=\rho(x)\,dx$, where
    \begin{equation} \label{eq:V-E-measure}
        \rho(x)=\psi(x)(A^2-x^2)^{\frac{4k+1}{2}},
        \qquad \psi>0\quad\hbox{on }[-A,A].
    \end{equation}
    In addition, the Euler--Lagrange inequality \eqref{variationalcondition:must-inequality1} for $V$ is strict on $\mathbb R\setminus[-A,A]$, such that $V$ has no additional type I singular points. 
    \item[(iii)] For each $j=1,\ldots,2k$, the deformation direction $V_j$ is chosen so that
    \begin{equation}
        h_j^{(m)}(\pm A)=0\quad (m=0,\ldots,j-2),
        \qquad h_j^{(j-1)}(\pm A)\neq0,
    \end{equation}
    where the first family of conditions is empty when $j=1$. The functions $h_j$ are defined by
    \begin{equation}\label{definition: hj}
    h_j(z)=\frac{1}{2\pi i}\oint_\gamma
    R(\xi)V_j'(\xi)\frac{d\xi}{\xi-z},\qquad \mbox{for $z\in\operatorname{Int}(\gamma)$ and
    $j=1,...,2k$}
\end{equation}
with 
\begin{equation}\label{definition: R}
    R(z)=\bigl((z-A)(z+A)\bigr)^{\frac{1}{2}},
        \qquad\mbox{for $z\in\mathbb{C}\setminus[-A,A]$,}
\end{equation}
where the principal branch of the square root is chosen so that $R$ is
analytic in $\mathbb{C}\setminus[-A,A]$ and $R(z)\sim z$ as $z\to\infty$.
In the above definition,  $\gamma$ is a positively oriented contour in $\mathcal V$ with
$[-A,A]\subset\operatorname{Int}(\gamma)$. Throughout the rest of this paper,  $\mathcal V$ denotes a neighborhood of the real line on which $V,V_1,...,V_{2k}$ and $\psi$ are analytic. 

    \end{itemize}
\end{assumptions}

Our main result establishes the double-scaling asymptotic expansion of the Hankel determinant for a potential exhibiting a type III edge singularity. A notable contribution of our analysis is the explicit evaluation of the constant $\mathcal{O}(1)$ term, which we show to be a regularized integral of the Hamiltonian from the Painlev\'e I hierarchy. To present these asymptotics in a clear and explicit form, we state our main theorem under the restriction that the higher-order deformation parameters vanish (i.e., $t_j=0$ for $j=2,\dots,2k$). Although the general case would introduce additional $t_j$-dependent terms, the essential structure of the constant term (namely, the emergence from the Painlev\'e I Hamiltonian) remains unchanged. 

\begin{theorem}\label{Main theorem}
 Let $V_t(x)=V(x)+t\left(x^2-V(x)\right)$ satisfy Assumptions \ref{assumptions} and $A^2\ne2$. Assume that $V_t$
 is one-cut regular for $t\in(0,1]$, and denote by $\mu_{V_t}$ its equilibrium measure, with density $\rho_t$ and support $[-b_t,b_t]$. For every fixed $B>0$, let $t=sn^{-\frac{4k+2}{4k+3}}$, then, as $n\to\infty$,

\begin{align}\label{asymp thm}
\log D_n(V_t)={B_1}n^2
 +B_2n^{\frac{4k+4}{4k+3}}+B_3n+B_4n^{\frac{2}{4k+3}}+B_5\log n +B_6+o(1),
\end{align}
uniformly for $s\in[0,B]$. Here, the coefficients $B_k$ are given by
\begin{align}
    &B_1=-\frac12\log2-\frac34-\int_0^1\int_{-b_t}^{b_t}(V(u)-u^2)\rho_t(u)\,du\,dt, \\
 &B_2=s\int_{-A}^{A}(V(u)-u^2)\rho(u)\,du ,\\
 & B_3=\log(2\pi),\\
 & B_4=\frac{s^2}2\Bigg(\frac1\pi\int_{-A}^{A}(V(x)-x^2)\sqrt{A^2-x^2}\dd x
  -\int_{-A}^{A}(V(x)-x^2)\psi(x)(A^2-x^2)^{2k+\frac12}\dd x \notag
 \\
  &\qquad
  +\frac{2-A^2}{2\pi}\int_{-A}^{A}\frac{V(x)-x^2}{\sqrt{A^2-x^2}}\dd x\Bigg),\\
  &B_5=-\frac{1}{4(4k+3)},\\
  &B_6=-\frac{(2k+1)^2}{(4k+3)(2k+2)}\alpha_k^{-\frac{1}{2k+1}}f_1(A)|f_1(A)|^{\frac{2k+2}{2k+1}}s^{\frac{4k+3}{2k+1}}+\zeta'(-1)
 \notag \\
 &\qquad+\frac{k}{12(2k+1)}\log\!\left(\frac{f_1(A)^2}{1+f_1(A)^2s^2}\right)-C_{\rm reg}+2f_1(A)\int_s^{\infty}\widehat h(f_1(A)y,0)\,dy.\label{eq:B6}
\end{align}
In the above formulas, the functions $\rho(x)$ and $\psi(x)$ are given in \eqref{eq:V-E-measure},   
\begin{equation*}
    C_{\rm reg}=\frac{1}{24}
\log\left\{
\frac{A^{4}}{16}
\left[
2\pi\psi(A)\frac{(4k+1)!!}{(2k)!}
\right]^{\frac{2}{2k+1}}
\left[
(2k+1)\lvert 2-A^{2}\rvert
\right]^{\frac{4k}{2k+1}}
\right\},
\end{equation*}
\[
 f_1(A)=\frac{A^2-2}{\sqrt{2AC_A}},\qquad
 C_A=\left(\frac{\pi}{2}(2A)^{\frac{4k+1}{2}}\psi(A)\right)^{\frac{2}{4k+3}},
\]
and
\begin{equation}\label{eq:def-hat-h-main}
 \widehat h(y,0)=h(y,0)-\frac{2k+1}{4k+4}\alpha_k^{-\frac{1}{2k+1}}|y|^{\frac{2k+2}{2k+1}
 }-\frac{ky}{12(2k+1)(y^2+1)},
\end{equation}
with $\alpha_k=\frac{2\Gamma(2k+\frac{3}{2})}{\Gamma(2k+2)\Gamma(\frac{3}{2})}$. The function $h$ is the Hamiltonian associated with the distinguished real pole-free solution of $\mathrm P_{\mathrm I}^{2k}$ equation defined in \eqref{def:PIm}.  
\end{theorem}
\begin{remark}
One may revert to the original parameter $t$ via the relation $t=sn^{-\frac{4k+2}{4k+3}}$, and combine the three $t$-dependent terms of orders $n^2$, $n^{\frac{4k+4}{4k+3}}$, $n^{\frac{2}{4k+3}}$ and the first term in $B_6$. Then, the asymptotic expansion \eqref{asymp thm} can be written in the following form
\begin{align}
    \label{eq:asy Dn(Vt)}
        &\log D_n(V_t)=-n^2I_{V_t}(\mu_{V_t})+n\log(2\pi)-\frac{1}{4(4k+3)}\log n +\zeta'(-1)
  \\
 &\quad +\frac{k}{12(2k+1)}\log\!\left(\frac{f_1(A)^2}{1+f_1(A)^2s^2}\right)-C_{\rm reg}+2f_1(A)\int_s^{\infty}\widehat h(f_1(A)y,0)\,dy+o(1), \notag
  \end{align}
    where $I_{V_t}(\mu_{V_t})$ is defined in \eqref{eq:energyf}. By setting the parameters to the special values $k=0$ and $t=0$, we recover the one-cut regular case (cf. the equilibrium measure in \eqref{eq:V-E-measure}). In this situation, the Painlev\'e system \eqref{def:PIm} reduces to the algebraic equation $s+4q(s)=0$, yielding $q(s)=-\frac{1}{4}s$ and $h(s)=\frac{1}{8}s^2$. Substituting $\alpha_0=2$ into the definition of $\widehat{h}$ given in \eqref{eq:def-hat-h-main}, we find that $\widehat{h}(s)=0$. Consequently, the expansion above simplifies to
    \begin{align}
        \log D_n(V)=&-n^2I_{V}(\mu_{V})+n\log(2\pi)-\frac{1}{12}\log n \nonumber \\
        &+\zeta'(-1)-\frac{1}{24}\log\left(\frac{A^4}{16}(2\pi\psi(A))^2\right)+
 o(1),
  \end{align}
  which is consistent with the asymptotics of the one-cut regular case in \eqref{eq:one-cut asy}.
 \end{remark}  
\begin{remark}\label{transition}
Although Theorem \ref{Main theorem} is rigorously established in the region $s\in[0,B]$ for any finite $B>0$, we can formally recover the transition to the one-cut regular asymptotics by considering the limit $s \to +\infty$. To see this, we set $s_n=\xi n^\delta$ with $\xi>0$ and $0<\delta<\frac{4k+2}{4k+3}$. Note that this choice gives us $t_n=\xi n^{\delta-\frac{4k+2}{4k+3}}$, which is a small and positive quantity. In this regime, the final integral in \eqref{eq:asy Dn(Vt)} tends to zero, reducing the asymptotic expansion to
 \begin{align}
         \log D_n(V_{t_n})=&-n^2I_{V_{t_n}}(\mu_{V_{t_n}})+n\log(2\pi)-\frac1{4(4k+3)}\log n+\zeta'(-1)\notag \\ &-\frac{k}{6(2k+1)}\delta\log n 
          -\frac{k}{6(2k+1)}\log\xi -C_{\mathrm{reg}}+o(1).
    \end{align}
 On the other hand, from \eqref{eq:Creg} and \eqref{eq:I3-Kt-identity}, we have, as $n\to\infty$,
 \begin{align}
        -\left[\frac1{4(4k+3)}+\frac{k}{6(2k+1)}\delta\right]\log n 
        -\frac{k}{6(2k+1)}\log\xi-C_{\mathrm{reg}}  \notag\\
        =-\frac{1}{12}\log n -\frac{1}{24}\log\left(\frac{b_{t_n}^4h_t(b_{t_n})h_t(-b_{t_n})}{16}\right)+o(1).
     \end{align}
     Combining the above two formulas, we have
     \begin{align}
         \log D_n(V_{t_n})=&-n^2I_{V_{t_n}}(\mu_{V_{t_n}})+n\log(2\pi)-\frac1{12}\log n+\zeta'(-1)\notag \\&-\frac{1}{24}\log\left(\frac{b_{t_n}^4h_t(b_{t_n})h_t(-b_{t_n})}{16}\right) +o(1).
    \end{align}
    which again recovers the asymptotics of the one-cut regular case given in \eqref{eq:one-cut asy}. Therefore, the regularized integral of the Hamiltonian $h$ associated with the $\mathrm{P_{I}^{2k}}$ equation acts as a crossover function, which describes the transition from the one-cut singular regime to the one-cut regular regime when $s$ varies.
 \end{remark}

 \begin{remark}\label{V even} The assumption that the potential $V$ is an even function is imposed only for technical convenience. One could extend this analysis to a general asymmetric potential whose equilibrium measure is supported on an interval $[a,b]$, where $a$ is a regular endpoint and $b$ is a singular endpoint of order $2k+\frac{1}{2}$. In such a setting, the asymptotic expansion \eqref{asymp thm} still contains terms of $n^2$, $n^{\frac{4k+4}{4k+3}}$, $n$, $n^{\frac{2}{4k+3}}$, $\log n$ and constant term, but with coefficients modified accordingly. Notably, the appearance of the $\mathrm{P_{I}^{2k}}$ Hamiltonian in the $\mathcal{O}(1)$ constant term is a universal feature that persists for any potential possessing a type III edge singularity.
\end{remark}
\begin{remark}
   In Theorem \ref{Main theorem}, we impose the technical assumption that $A^2\ne2$. If $A^2=2$, from \eqref{eq:def f1} and \eqref{eq:rho1}, we have 
    \begin{equation}
        f_1(z)=\frac{2\sqrt{2A}}{3\sqrt{C_A}}(z-A)+\bigO((z-A)^2).
    \end{equation}
    Then, from Remark \ref{remark h} and the RH analysis in Section \ref{subsection: parametrix near b}, the nontrivial term in the $\theta(\zeta)$ (defined in \eqref{definition: theta}) is $\frac{4}{4k+3}\zeta^\frac{4k+3}{2}+\frac{4}{3}\tau_1\zeta^\frac{3}{2}$ rather than $\frac{4}{4k+3}\zeta^\frac{4k+3}{2}+x\zeta^\frac{1}{2}$. Accordingly, the appropriate double-scaling parameter is $t=\tau n^{-\frac{4k}{4k+3}}$, and the integral in $B_6$ is replaced by $\int_{\tau}^{\infty}{h}(0,u,...,0)\,du$. However, the RH analysis in this paper is still applicable. 
\end{remark}
\begin{remark}
To illustrate, a concrete example of a potential satisfying Assumption~\ref{assumptions} is given by
   $$V(x)=\sum_{j=0}^{2k}(-1)^j\binom{2k+\frac{1}{2}}{j}\frac{A^{2j}}{4k+2-2j}x^{4k+2-2j}.$$ 
   A direct computation gives the density of the corresponding equilibrium measure
   \(\rho(x)=\frac{1}{2\pi}(A^2-x^2)^\frac{4k+1}{2}\), where
   \(A^{4k+2}=\frac{2^{2k+2}(2k+1)!}{(4k+1)!!}\). Consider the deformed potential
\[
V_t(x)=(1-t)V(x)+tx^2,
\]
one finds that the density of the corresponding equilibrium measure takes the form
\begin{equation} \label{eq:rho-t-example}
\rho_t(x)=\frac{1}{2\pi}\left[(1-t)\sum_{j=0}^{2k}\binom{2k+\frac{1}{2}}{j}(b_t^2-A^2)^j(x^2-b_t^2)^{2k-j}+2t\right]\sqrt{b_t^2-x^2},
\end{equation}
Here, \(b_t\) is determined by the following equation:
$$
(1-t)\gamma_k(b_t^2-A^2)^{2k+1}+t(b_t^2-2)=0 \quad \textrm{with }\gamma_k=\binom{2k+\frac{1}{2}}{2k+1}=\frac{(4k+1)!!}{2^{2k+1}(2k+1)!}.
$$
Because the quantity inside the square brackets in \eqref{eq:rho-t-example} is strictly positive for all $x \in [-b_t, b_t]$, it follows that $V_t$ remains a regular one-cut potential for all $t \in (0,1]$.

\end{remark}
\begin{remark}\label{remark FH singularities}
    We can further consider the asymptotics of the Hankel determinant when Fisher-Hartwig singularities are introduced, like the case considered in \cite{Charlier,CFWW,CG,CFLW}. More precisely, following the same framework, let
    \[
 -A<a_1<\cdots<a_m<A,
 \qquad
 \min\!\left\{
   \min_{j\ne \ell}|a_j-a_\ell|,
   \min_j(A-a_j),
   \min_j(A+a_j)
 \right\}\geq \delta
\]
for some fixed \(\delta>0\), and define
\[
 \omega_{\alpha_j}(x)=|x-a_j|^{\alpha_j},
 \qquad
 \omega_{\beta_j}(x)=
 \begin{cases}
  e^{\ii\pi\beta_j},&x<a_j,\\
  e^{-\ii\pi\beta_j},&x>a_j,
 \end{cases}
\]
with the parameters $
 \Re\alpha_j>-1,
 |\Re\beta_j|<\frac14,
  j=1,\ldots,m.$
Define
\begin{equation}\label{eq:fh-determinant}
 D_n^{\mathrm{FH}}(V_t;\boldsymbol\alpha,\boldsymbol\beta)
 =\det\!\left(
  \int_{\mathbb R}x^{p+q}e^{-nV_t(x)}
  \prod_{j=1}^m
  \omega_{\alpha_j}(x)\omega_{\beta_j}(x)\,dx
 \right)_{p,q=0}^{n-1}.
\end{equation}
Then, under the same conditions shown in Theorem \ref{Main theorem}, we have as $n\to\infty$,
\begin{equation}\label{eq:fh-regular-asymptotics}
 \log\frac{D_n^{\mathrm{FH}}(V_t;\boldsymbol\alpha,\boldsymbol\beta)}{D_n(V_t)}
 =n\Phi_{\mathrm{FH}}[\mu_{V_t}]
 +\sum_{j=1}^m
 \left(\frac{\alpha_j^2}{4}-\beta_j^2\right)\log n
 + B_{\mathrm{FH}}+o(1),
\end{equation}
where
\begin{equation}\label{eq:fh-phase-equilibrium}
 \Phi_{\mathrm{FH}}[\mu_{V_t}]
 =\sum_{j=1}^m\left\{
  \frac{\alpha_j}{2}\bigl(V_t(a_j)+\ell_t\bigr)
  +i\pi\beta_j
   \left(1-2\int_{a_j}^{b_t}\rho_t(x)\,dx\right)
 \right\},
\end{equation}
{\small
\begin{align}
    B_{\mathrm{FH}}
    ={}&
    -\frac12
    \sum_{1\leq j<\ell\leq m}
    \alpha_j\alpha_\ell
    \log\left(
        \frac{2|a_j-a_\ell|}{A}
    \right)+
    2 
    \sum_{1\leq j<\ell\leq m}
    \beta_j\beta_\ell
    \log \left( \frac{
        A^2-{a_ja_\ell}
        -
        \sqrt{A^2-{a_j^2}}
        \sqrt{A^2-{a_\ell^2}}
    }{
        {|a_j-a_\ell|}{A}
    }\right)  \notag \\
    &+
    \sum_{j=1}^{m}
    \left[
        i\mathfrak A\beta_j
        \arcsin\left(\frac{a_j}{A}\right)
        -
        \frac{i\pi}{2}
        \mathfrak A_j\beta_j
    \right]
        +
    \sum_{j=1}^{m}
    \log
    \frac{
        G\left(
            1+\frac{\alpha_j}{2}+\beta_j
        \right)
        G\left(
            1+\frac{\alpha_j}{2}-\beta_j
        \right)
    }{
        G(1+\alpha_j)
    }\notag \\
    &+\sum_{j=1}^{m}
    \left(
        \frac{\alpha_j^2}{4}-\beta_j^2
    \right)
    \log\left[
        \frac{\pi}{2}
        A^{2}
        \psi(a_j)
        \left(
            A^2-{a_j^2}
        \right)^{2k}
    \right]+
    \sum_{j=1}^{m}
    \left(
        \frac{\alpha_j^2}{4}-3\beta_j^2
    \right)
    \log\left[
        2
        \sqrt{
            1-\frac{a_j^2}{A^2}
        }
    \right].
    \label{eq:critical-FH-constant}
\end{align}}
Here, $\mathfrak{A}:=\sum_{j=1}^m\alpha_j$, $\mathfrak A_j
    :=
    \sum_{\ell<j}\alpha_\ell
    -
    \sum_{\ell>j}\alpha_\ell,j=1,\ldots,m$, and $G$ denotes the Barnes $G$-function. It is worth noting that the regularized integral of the Painlev\'e I Hamiltonian $h$, implicitly contained within the expression for $D_n(V_t)$ on the left-hand side of \eqref{eq:fh-regular-asymptotics}, still persists.
\end{remark}

Our second main result establishes the universality of the eigenvalue correlation kernel near the singular edge. As mentioned before, the case of $k=1$ was studied in \cite{ClaeysVan}. Here, we extend this to an arbitrary integer $k \ge 1$ and establish the universality result for the potential $V_{\mathbf{t}}$ with a general parameter set $(t_1, \dots, t_{2k})$.

\begin{theorem}\label{theorem: universality}
Let $V_{\mathbf t}=V+\sum_{j=1}^{2k}t_jV_j$ be a potential satisfying Assumptions~\ref{assumptions}.  Consider the multi-scaling limit where $n\to\infty$ and $t_j\to 0$ ($j=1,\ldots,2k$) such that each of the following limit exists:
\begin{equation}\label{lim stn-1}
 \widehat t_j:=\lim_{n\to\infty}c_j n^{\frac{4k-2j+4}{4k+3}}t_j\in\mathbb R,
 \qquad j=1,\ldots,2k,
\end{equation}
where the constants $c_1,\ldots,c_{2k}$ are defined in \eqref{eq:def c1} and~\eqref{eq:def cj}.  Set $c=\left(\frac{\pi}{2}(2A)^{\frac{4k+1}{2}}\psi(A)\right)^{\frac{2}{4k+3}}$, $s=\widehat t_1$, and $\boldsymbol\tau=(\widehat t_2,\ldots,\widehat t_{2k})$.  Then, the limiting kernel relation 
\eqref{eq:limiting kernel} holds at the critical edge $x^*=A$. 
\end{theorem}
\begin{remark}
    Since $q(s,\tau_1,\dots,\tau_{2k-1})$ is pole-free for all $s, \tau_j \in \mathbb{R}$ with $j=1,\dots,2k-1$ (cf. \cite{Claeys}), an argument similar to the $k=1$ case in \cite{ClaeysVan} shows that the kernel $K^{(2k)}(u,v;s,\boldsymbol\tau)$ is real for real $u,v,s$ and $\boldsymbol{\tau}$. Moreover, this kernel admits the integral representation
    \begin{equation}
        K^{(2k)}(u,v;s,\boldsymbol\tau)=\frac{1}{2\pi i}\int_{-\infty}^s\Psi_1(u;\sigma,\boldsymbol{\tau})\Psi_1(v;\sigma,\boldsymbol{\tau})d\sigma.
    \end{equation}
\end{remark}

\subsection{Organization of the Paper}
In Section \ref{section: equilibrium measures}, we derive the relevant equilibrium measures, which play a crucial role in the subsequent RH analysis. In Sections \ref{Section:OP and Y} and \ref{Section: steepest descent t>0}, we perform the powerful Deift-Zhou steepest descent analysis of the RH problems. With the required asymptotic results, we finally prove our main theorems in Sections \ref{sec:proof 1.3} and \ref{section: universality}. Finally, in the Appendix, we provide a brief discussion on the properties of Cauchy operators, the $P_{\rm I}^{2k}$ parametrix, and the Airy model RH problem.

\section{Equilibrium measures}\label{section: equilibrium measures}
In this section, we consider the external fields $V_{\mathbf{t}}=V+\sum_{j=1}^{2k}t_jV_j$ that satisfy the
Assumptions \ref{assumptions} proposed in the beginning of Section
\ref{subsection: statement of results}. 
Similar to \cite{ClaeysVan}, we seek signed equilibrium measures $\nu_{\mathbf{t}}$ in the following form,
\begin{equation}\label{definition: nust}
    \nu_{\mathbf{t}}=\nu+t_1\nu_1+...+t_{2k}\nu_{2k}.
\end{equation}
From Assumption \ref{assumptions} (ii), we know that $\nu$ can be written as,
\begin{equation}\label{definition: h0}
    d\nu(x) = \rho(x)\chi_{[-A,A]}dx=\psi(x)(A^2-x^2)^\frac{4k+1}{2}\chi_{[-A,A]}dx,
\end{equation}
where $\psi(x)$ is positive on $[-A,A]$, and $\chi_{[-A,A]}$ denotes the indicator function of $[-A,A]$. Note that we use the symbols $\rho(x)$ and $\rho_j(x)$ to denote the densities on $[-A,A]$, which is the upper boundary values of the corresponding analytic functions $\varrho(z)$ or $\varrho_j(z)$. To be specific, we have
\begin{equation}\label{definition: psi0}
    \varrho(z)=\frac{1}{2\pi
        i}R(z) h(z),\qquad\mbox{for $z\in\mathcal V\setminus[-A,A]$,}
\end{equation}
with $h$ analytic in the neighborhood $\mathcal V$ of the real line and $R$ given in \eqref{definition: R}. Then $\nu$ satisfies the following condition: there exists
$\ell\in\mathbb{R}$ such that
\begin{align}
    \label{variationalcondition:nu0-equality}
    & 2\int\log|x-u|d\nu(u)-V(x)=\ell, && \mbox{for $x\in [-A,A]$,} \\
    \label{variationalcondition:nu0-inequality}
    & 2\int\log|x-u|d\nu(u)-V(x)< \ell, && \mbox{for $x\in \mathbb R\setminus [-A,A]$.}
\end{align}

In order to construct the remaining measures $\nu_j$, $j=1,...,2k$, note that the fractional residue theorem gives
\begin{equation}\label{hj real}
    h_j(x)=-\frac{1}{\pi i}\PVint_{-A}^A
    R_+(u)V_j'(u)\frac{du}{u-x},\qquad \mbox{for $x\in[-A,A]$,}
\end{equation}
where $h_j$ is defined in \eqref{definition: hj} and the integral takes the Cauchy principal value. Hence, $h_j$ is real on $[-A,A]$.

\begin{lemma}
    Define signed measures
    $\nu_j$, $j=1,...,2k$, that are supported on $[-A,A]$ as
    \begin{equation}
        d\nu_j(x)=\rho_{j}(x)\chi_{[-A,A]}dx,\qquad j=1,...,2k,
    \end{equation}
    where the corresponding $\varrho_{j}(z)$ is defined by
\begin{equation}\label{definition: psij}
        \varrho_j(z)=\frac{1}{2\pi
        i}\frac{h_j(z)}{R(z)},\qquad\mbox{for $z\in\mathcal V\setminus[-A,A]$.}
    \end{equation}
    Here, $h_j$-s are given by {\rm(\ref{definition: hj})}; see also {\rm (\ref{hj real})} for its expression on $[-A,A]$, and
    $R$ is given by {\rm (\ref{definition: R})}. Then, 
    \begin{equation}\label{zero mass nuj}
        \nu_j([-A,A])=\int_{-A}^A\rho_{j}(u)du=0,
    \end{equation}
    and there exist constants $\ell_j\in\mathbb R$ such that
    \begin{equation}\label{variational equality: nuj}
        2\int\log|x-u|d\nu_j(u)-V_j(x)=\ell_j, \qquad\mbox{for $x\in[-A,A]$.}
    \end{equation}
\end{lemma}

\begin{proof}
    Our proof is similar to the work in \cite{ClaeysVan}. Define, for $j=1,...,2k$, the auxiliary functions
    \begin{equation}\label{definition: F}
        F_j(z)=\frac{1}{2\pi iR(z)}\int_{-A}^A
        R_+(u)V_j'(u)\frac{du}{u-z},\qquad\mbox{for $z\in\mathbb
        C\setminus[-A,A]$}.
    \end{equation}
   Then, 
    \begin{align}
        \label{SP1}
        & F_{j,+}(x)-F_{j,-}(x)=-2\pi i\rho_{j}(x), &\mbox{for $x\in[-A,A]$,}
        \\[1ex]
        \label{SP2}
        &F_{j,+}(x)+F_{j,-}(x)=V_j'(x), &\mbox{for $x\in[-A,A]$.}
    \end{align}
Since $F_j$ is analytic in $\mathbb C\setminus[-A,A]$, and $F_j(z)=\bigO(z^{-2})$ as $z\to\infty$, its Cauchy representation and \eqref{SP1} give that
    \[
        F_j(z)=-\int_{-A}^A\frac{\rho_{j}(u)}{u-z}du=
        z^{-1}\int_{-A}^A\rho_{j}(u)du+\bigO(z^{-2}),\qquad\mbox{as $z\to\infty$}.
    \]
Comparing with the fact that $F_j(z)=\bigO(z^{-2})$ as $z\to\infty$, we obtain $\int_{-A}^A\rho_{j}(u)du=0$.
From \eqref{SP2}, we obtain
    \begin{equation}
        2\,\PVint_{-A}^A\frac{\rho_{j}(u)}{x-u}\,du=V_j'(x).
    \end{equation}
    Therefore,
    \begin{equation}
        \frac{d}{dx}\left(2\int\log|x-u|\,d\nu_j(u)-V_j(x)\right)
        =2\,\PVint_{-A}^A\frac{\rho_{j}(u)}{x-u}\,du-V_j'(x)=0.
    \end{equation}
This proves \eqref{variational equality: nuj} and completes the proof of the lemma.
\end{proof}

\begin{corollary}
\label{corollary: equilibrium measures}
    Let $\nu_{\mathbf{t}}=\nu+t_1\nu_1+...+t_{2k}\nu_{2k}$.  Then the signed measure
    $d\nu_{\mathbf{t}}(x)=\rho_{\mathbf{t}}(x)\chi_{[-A,A]}dx$ is the fixed-support modified equilibrium measure, where
    \begin{equation}\label{eq:definition: psist}
        \varrho_{\mathbf{t}}=\varrho+t_1\varrho_1+...+t_{2k}\varrho_{2k},\qquad\mbox{on $\mathcal
        V\setminus[-A,A]$,}
    \end{equation}
    with $\varrho$ given by {\rm (\ref{definition: psi0})} and
    $\varrho_{j}$, $j=1,2,...,2k$ given by {\rm (\ref{definition: psij})}.
    Thus, we have $\nu_{\mathbf{t}}([-A,A])=1$.  
    Further, there exist constants $\ell_{\mathbf{t}}\in\mathbb{R}$ such
    that for any $\delta>0$ there are $\varepsilon,\kappa>0$
    sufficiently small such that for
    $t_j\in[-\varepsilon,\varepsilon]$, we have 
    \begin{align}
    \label{nust: variational equality}
    & 2\int\log|x-u|d\nu_{\mathbf{t}}(u)-V_{\mathbf{t}}(x)=\ell_{\mathbf{t}},
        &&\mbox{for $x\in[-A,A]$.}
    \\
    \label{nust: variational inequality}
    & 2\int\log|x-u|d\nu_{\mathbf{t}}(u)-V_{\mathbf{t}}(x)<\ell_{\mathbf{t}}-\kappa,
        &&\mbox{for $x\in\mathbb R\setminus [-A-\delta, A+\delta]$.}
\end{align}
\end{corollary}

\begin{proof}
     From (\ref{zero mass nuj}) and  the fact that $\nu([-A,A])=1$, it is clear that $\nu_{\mathbf{t}}([-A,A])=1$. Next, with
$\ell_{\mathbf{t}}=\ell+t_1\ell_1+...+t_{2k}\ell_{2k}$, we have
    \begin{equation}\label{proof: corrolary equilibrium measures: eq1}
        2\int \log|x-u|d\nu_{\mathbf{t}}(u)-V_{\mathbf{t}}(x)-\ell_{\mathbf{t}}= I(x)+t_1
        I_1(x)+...+t_{2k}I_{2k}(x)
    \end{equation}
    where
     \[
        I(x)= 2\int
        \log|x-u|d\nu(u)-V(x)-\ell
    \]
    \[
        I_j(x)= 2\int
        \log|x-u|d\nu_j(u)-V_j(x)-\ell_j,\qquad j=1,...2k.
    \]
    Then, condition (\ref{nust: variational equality}) follows from
    (\ref{variationalcondition:nu0-equality}) and (\ref{variational equality: nuj}). Similar to \cite{ClaeysVan}, on can check that there exists $\kappa>0$ such that
    \begin{equation}\label{proof: corrolary equilibrium measures: eq2}
        I(x)<-\frac{3}{2}\kappa \quad \mbox{and} \quad \left|\sum_{j=1}^{2k}t_jI_j(x)\right|<\frac{1}{2}\kappa \qquad\mbox{on
        $\mathbb{R}\setminus[-A-\delta,A+\delta]$,}
    \end{equation}
     with $t_j\in[-\epsilon,\epsilon]$, $j=1,...,2k$. This completes the proof of the lemma.
\end{proof}

\begin{remark}\label{remark: nust positive}
Since $\nu$ has a strictly positive density on
$(-A,A)$, we have that for any $\delta>0$, $\nu_{\mathbf{t}}$ is positive on $(-A+\delta,A-\delta)$ for sufficiently small $t_j$ and $j=1,...,2k$.
\end{remark}

\section{Riemann-Hilbert problem and steepest descent analysis as $t\to 0$}\label{Section:OP and Y}

Consider the matrix valued function $Y(z) = Y^{(n)}(z;V_{\mathbf{t}})$, defined by
\begin{equation}
Y(z) = \begin{pmatrix}\label{Y definition}
(\kappa_n^{(\mathbf{t})})^{-1}p_n^{(\mathbf{t})} (z) & \frac{(\kappa_n^{(\mathbf{t})})^{-1}}{2\pi i} \int_{\mathbb{R}} \frac{p_n^{(\mathbf{t})}(x)w_{\mathbf{t}}(x)}{x-z}dx \\
-2\pi i \kappa_{n-1}^{(\mathbf{t})} p^{(\mathbf{t})}_{n-1}(z) & -\kappa_{n-1}^{(\mathbf{t})} \int_{\mathbb{R}} \frac{p^{(\mathbf{t})}_{n-1}(x)w_{\mathbf{t}}(x)}{x-z}dx
\end{pmatrix},
\end{equation}
where $p_n^{(\mathbf{t})}$ is given in \eqref{def of OP}. It is known \cite{FokasItsKitaev} that $Y$ can be characterized as the following RH problem. 
\subsubsection*{RH problem for $Y$}
\begin{itemize}
\item[(a)] $Y : \mathbb{C}\setminus \mathbb{R} \to \mathbb{C}^{2\times 2}$ is analytic.
\item[(b)] $Y$ satisfies the following jump conditions
\begin{equation}\label{jump relations of Y}
Y_{+}(x) = Y_{-}(x) \begin{pmatrix}
1 & w_{\mathbf{t}}(x) \\ 0 & 1
\end{pmatrix}, \hspace{0.5cm} \mbox{ for } x \in \mathbb{R},
\end{equation}
where $w_{\mathbf{t}}(x)$ is defined in \eqref{weight_general}.
\item[(c)] As $z \to \infty$, we have $Y(z) = \left(I + \bigO(z^{-1})\right) z^{n\sigma_{3}}$, where $\sigma_{3} = \begin{pmatrix}
1 & 0 \\ 0 & -1
\end{pmatrix}$.

\end{itemize}

In the next subsections, we will use the Deift-Zhou steepest descent method \cite{DKMVZ1,DKMVZ2,DeiftZhou,DeiftZhou1992} to analyze the RH problem $Y$ introduced above.

\subsection{Normalization of the RH problem at infinity: $Y\mapsto T$}
    \label{subsection: sd-normalization}

In order to normalize the RH problem for $Y$ at infinity, the signed fixed-support modified equilibrium measures $\nu_{\mathbf{t}}$, introduced in Section
\ref{section: equilibrium measures}, play a key role. Consider 
\begin{equation}\label{definition: gst}
    g(z)=\int_{-A}^A\log(z-u)d\nu_{\mathbf{t}}(u),
        \qquad\mbox{for $z\in\mathbb C\setminus(-\infty,A]$,}
\end{equation}
where $\log(\cdot)$ takes the principal branch.
From \eqref{definition: gst} and condition \eqref{nust: variational equality}, it follows that
\begin{equation}\label{property g: 1}
    g_{+}(x)+g_{-}(x)-V_{\mathbf{t}}(x)-\ell_{\mathbf{t}}=0,\qquad\mbox{for
    $x\in[-A,A]$.}
\end{equation}
One can also show that
\begin{equation}\label{property g: 2}
    g_{+}(x)-g_{-}(x)=2\pi i\int_x^A d\nu_{\mathbf{t}}(u),\qquad \mbox{for $x\in\mathbb
    R$,}
\end{equation}
so that since $\nu_{\mathbf{t}}$ is supported on $[-A,A]$ and has mass of one,
\begin{equation}\label{property g: 3}
    g_{+}(x)-g_{-}(x)=
    \begin{cases}
        2\pi i, & \mbox{for $x<-A$,} \\
        0, & \mbox{for $x>A$.}
    \end{cases}
\end{equation}

Now we are ready to perform the first transformation $Y\mapsto T$.
Define the matrix valued function $T$ as
\begin{equation}\label{TinY}
    T(z)=e^{-\frac{1}{2} n\ell_{\mathbf{t}}\sigma_3} Y(z) e^{-ng(z)\sigma_3}
    e^{\frac{1}{2}n\ell_{\mathbf{t}}\sigma_3}, \qquad \mbox{for
    $z\in\mathbb{C}\setminus\mathbb R$,}
\end{equation}
where $\ell_{\mathbf{t}}$ is the constant that appears in the variational
conditions (\ref{nust: variational equality}) and (\ref{nust: variational inequality}).
It is straightforward to check that $T$ is a
solution to the following RH problem.

\subsubsection*{RH problem for $T$:}
\begin{itemize}
    \item[(a)] $T:\mathbb{C}\setminus \mathbb{R}\to\mathbb{C}^{2\times 2}$ is analytic.
    \item[(b)] $T_+(x)=T_-(x)v_T(x)$ for $x\in\mathbb{R}$, with
        \begin{equation} \label{RHP T: b}
            v_T(x) =
            \begin{cases}
                \begin{pmatrix}
                    e^{-n(g_{+}(x)-g_{-}(x))} & 1 \\
                    0 & e^{n(g_{+}(x)-g_{-}(x))}
                \end{pmatrix}, & \mbox{on $(-A,A)$,} \\[3ex]
                \begin{pmatrix}
                    1 & e^{n(g_{+}(x)+g_{-}(x)-V_{\mathbf{t}}(x)- \ell_{\mathbf{t}})} \\
                    0 & 1
                \end{pmatrix}, & \mbox{on $\mathbb R\setminus (-A,A)$.}
            \end{cases}
        \end{equation}
    \item[(c)] $T(z)=I+\bigO(\frac{1}{z})$,\qquad as $z\to\infty$.
\end{itemize}

\subsection{Opening of the lens: $T\mapsto S$}
    \label{subsection: sd-lens}
We next introduce a function $\phi_{\mathbf{t}}$ as follows,
\begin{equation}\label{definition: phist}
    \phi_{\mathbf{t}}(z) =-\pi i\int_z^A\varrho_{\mathbf{t}}(\xi)d\xi,\qquad\mbox{for
    $z\in\mathcal V\setminus(-\infty,A]$,}
\end{equation}
where $\varrho_{\mathbf{t}}$ is defined by (\ref{eq:definition: psist}), and the path of integration does not cross $(-\infty,A]$. Then $\phi_{\mathbf{t}}$ satisfies
\begin{equation}
    -2\phi_{\mathbf{t},+}(x)=2\phi_{\mathbf{t},-}(x)=2\pi i\int_x^{A} d\nu_{\mathbf{t}}(u)=g_{+}(x)-g_{-}(x),
    \quad\mbox{for $x\in(-A,A)$,}
   \label{equation-phig-bulk}
\end{equation}
\begin{equation}\label{property phi: 2}
    2g(x)-V_{\mathbf{t}}(x)-\ell_{\mathbf{t}}=-2\phi_{\mathbf{t}}(x),\qquad \mbox{on $\mathcal{V}\setminus(-\infty,-A]$.}
\end{equation}
Using \eqref{property g: 3}, this yields
\begin{equation}
    g_{+}(x)+g_{-}(x)-V_{\mathbf{t}}(x)-\ell_{\mathbf{t}}
   =-2\phi_{\mathbf{t},-}(x)+2\pi
    i, \quad \mbox{on $(-\infty,-A)$}.\label{property phi: 3}
\end{equation}

Inserting \eqref{equation-phig-bulk}, \eqref{property phi: 2}, and \eqref{property phi: 3} into \eqref{RHP T: b}, the jump matrix for
$T$ can be written in terms of $\phi_{\mathbf{t}}$ as
\begin{equation}
    v_T(x)=
    \begin{cases}
        \begin{pmatrix}
            e^{2n\phi_{\mathbf{t},+}(x)} & 1 \\
                    0 & e^{2n\phi_{\mathbf{t},-}(x)}
                \end{pmatrix}, & \mbox{on $(-A,A)$,} \\[3ex]
                \begin{pmatrix}
                    1 & e^{-2n\phi_{\mathbf{t},-}(x)} \\
                    0 & 1
                \end{pmatrix}, & \mbox{on $\mathbb R\setminus (-A,A)$.}
    \end{cases}
\end{equation}
It is straightforward to check that $v_T$ has the following factorization on the
interval $(-A,A)$,
\begin{equation}\label{factorization}
    v_T(x)
    =
    \begin{pmatrix}
        1 & 0 \\
        e^{2n \phi_{\mathbf{t},-}(x)} & 1
    \end{pmatrix}
    \begin{pmatrix}
        0 & 1 \\
        -1 & 0
    \end{pmatrix}
    \begin{pmatrix}
        1 & 0 \\
        e^{2n\phi_{\mathbf{t},+}(x)} & 1
    \end{pmatrix},\qquad \mbox{on $(-A,A)$.}
\end{equation}
This motivates us to introduce the following transformation 
\begin{equation}\label{SinT}
    S(z)=
    \begin{cases}
        T(z), & \mbox{for $z$ outside the lens.} \\[1ex]
        T(z)
            \begin{pmatrix}
                1 & 0\\
                -e^{2n\phi_{\mathbf{t}}(z)} & 1
            \end{pmatrix}, & \mbox{for $z$ in the upper part of the lens,}\\[3ex]
        T(z)
            \begin{pmatrix}
                1 & 0 \\
                e^{2n\phi_{\mathbf{t}}(z)}& 1
            \end{pmatrix}, & \mbox{for $z$ in the lower part of the lens,}
    \end{cases}
\end{equation}
see Figure \ref{figure: opening of the lens} for an illustration. Then $S$ is the unique solution to the following RH problem.
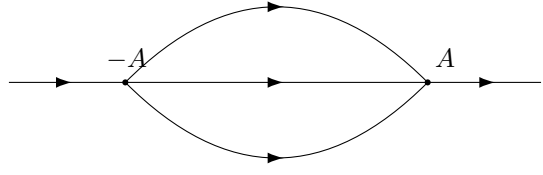
\begin{figure}[t]
\begin{center}
    \setlength{\unitlength}{1mm}
    \begin{picture}(137.5,26)(-2.5,11.5)

        \put(45,25){\thicklines\circle*{.8}} \put(42.5,27){\small $-A$}
        \put(85,25){\thicklines\circle*{.8}} \put(86,27){\small $A$}

        \put(94,25){\thicklines\vector(1,0){.0001}}
        \put(29.6,25){\line(1,0){70.4}} \put(38,25){\thicklines\vector(1,0){.0001}}
        \put(66,25){\thicklines\vector(1,0){.0001}}

        \qbezier(45,25)(65,45)(85,25) \put(66,35){\thicklines\vector(1,0){.0001}}
        \qbezier(45,25)(65,5)(85,25) \put(66,15){\thicklines\vector(1,0){.0001}}

    \end{picture}
    \caption{The contours $\Sigma_S$ of the RH problem for $S$}
    \label{figure: opening of the lens}
\end{center}
\end{figure}

\subsubsection*{RH problem for $S$:}
\begin{itemize}
    \item[(a)] $S:\mathbb{C}\setminus\Sigma_S\to\mathbb{C}^{2\times 2}$ is analytic, where the contours $\Sigma_S$ are shown in Figure \ref{figure: opening of the lens}.
    \item[(b)] $S_+(z)=S_-(z)v_S(z)$ for $z\in\Sigma_S$, with
        \begin{equation}\label{RHP S: b}
            v_S(z)=
            \begin{cases}
                \begin{pmatrix}
                    0 & 1 \\
                    -1 & 0 \\
                \end{pmatrix}, & \mbox{on $(-A,A)$,}
                \\[3ex]
                \begin{pmatrix}
                    1 & 0 \\
                    e^{2n \phi_{\mathbf{t}}(z)} & 1 \\
                \end{pmatrix}, & \mbox{on $\Sigma_S\cap\mathbb C_\pm$,}
                \\[3ex]
                \begin{pmatrix}
                    1 & e^{-2n\phi_{\mathbf{t},-}(x)} \\
                    0 & 1
                \end{pmatrix}, & \mbox{on $\mathbb R\setminus(-A,A)$.}
            \end{cases}
        \end{equation}
    \item[(c)] $S(z)=I+\bigO(\frac{1}{z})$, \qquad as $z\to\infty$.
    \end{itemize}

\subsection{Global parametrix for $\hat{P}^{(\infty)}$}\label{sec:hat Pinfty}
    \label{subsection: sd-outside}
From the Cauchy-Riemann conditions, we have 
\begin{equation}\label{Re phist<0}
    \Re\phi_{\mathbf{t}}(z)<0,\quad\mbox{for $|\Im z|\neq 0$ small and $-A+\delta<\Re z
    <A-\delta$}.
\end{equation}
This, combining with \eqref{nust: variational inequality}, gives that all jumps bounded away from $[-A,A]$ are exponentially close to the identity as $n\to\infty$. Then we seek the following RH problem for the global
parametrix $\hat{P}^{(\infty)}$.

\subsubsection*{RH problem for $\hat{P}^{(\infty)}$:}
\begin{itemize}
    \item[(a)] $\hat{P}^{(\infty)}:\mathbb{C}\setminus [-A,A]\to\mathbb{C}^{2\times 2}$ is
        analytic.
    \item[(b)] $\hat{P}^{(\infty)}_+(x)=\hat{P}^{(\infty)}_-(x)
        \begin{pmatrix}
            0 & 1 \\
            -1 & 0
        \end{pmatrix}$,\qquad for $x\in(-A,A)$.
    \item[(c)] $\hat{P}^{(\infty)}(z)=I+\bigO(\frac{1}{z})$,\qquad as $z\to\infty$.
\end{itemize}

It is well known, see for example \cite{Deift,DKMVZ1}, that the solution
$\hat{P}^{(\infty)}$ is given by
\begin{equation}\label{definition: Pinfty}
    \hat{P}^{(\infty)}(z)=M^{-1}a(z)^{-\sigma_3}M ,\qquad\mbox{for $z\in\mathbb C\setminus[-A,A]$,}
\end{equation}
with $a(z)=\left(\frac{z-A}{z+A}\right)^\frac{1}{4}$, $M=\frac{I+i\sigma_1}{\sqrt{2}}$.

\subsection{Local parametrices at the critical endpoints $\pm A$}
    \label{subsection: parametrix near b}

In this subsection, we perform a local analysis near the critical endpoint $A$. Let
$U_{\delta,A}=\{z\in\mathbb{C}:|z-A|<\delta\}$ be a small disk with center $A$ and radius $\delta>0$ sufficiently small such that $U_{\delta,A}$ lies in $\mathcal V$ and such that $U_{\delta,-A}$ and $U_{\delta,A}$ do not intersect. We then seek $P^{(A)}$ satisfying the following RH problem.

\subsubsection*{RH problem for $P^{(A)}$:}

\begin{itemize}
    \item[(a)] $P^{(A)}:U_{\delta,A}\setminus\Sigma_S\to\mathbb{C}^{2\times 2}$ is
        analytic.
    \item[(b)] $P^{(A)}_+(z)=P^{(A)}_-(z)v_S(z)$ for $z\in U_{\delta,A}\cap
        \Sigma_S$, where $v_S$ is given by (\ref{RHP S: b}).
    \item[(c)] $P^{(A)}$ satisfies the matching condition
        \begin{equation}
            P^{(A)}(z)(\hat{P}^{(\infty)})^{-1}(z)=I+\bigO(n^{-\frac{1}{4k+3}}),
        \end{equation}
        as $n\to\infty$ and $t_j\to 0$, $j=1,...,2k$ such that \eqref{lim stn-1} holds uniformly for
        $z\in\partial U_{\delta,A}\setminus\Sigma_S$.
\end{itemize}

We then construct $P^{(A)}$ using the model RH problem $\Psi$ of the $P_{\rm I}^{2k}$ equation, which was introduced in Section \ref{subsection: PI2 equation}. Let
\begin{equation}\label{definition: Psingular}
    P^{(A)}(z)=E^{(A)}(z)\Psi\left(n^{\frac{2}{4k+3}}f_A(z);n^{\frac{4k+2}{4k+3}}t_1f_1(z),n^{\frac{4k}{4k+3}}t_2f_2(z),...,n^{\frac{4}{4k+3}}t_{2k}f_{2k}(z)\right)
        e^{n\phi_{\mathbf{t}}(z)\sigma_3},
\end{equation}
where $E^{(A)}$ is defined as
\begin{equation}\label{definition: Eb}
    E^{(A)}(z)=\hat{P}^{(\infty)} (z) N^{-1} \left(n^{\frac{2}{4k+3}}f_A(z) \right)^{\frac{\sigma_3}{4}},
\end{equation}
with $N$ given by \eqref{definition: N} and $\hat{P}^{(\infty)}$ given by \eqref{definition: Pinfty}. It is straightforward to see that $E^{(A)}$ is invertible and analytic in $U_{\delta,A}$. In addition, $f_A$ and $f_j$, $j=1,...,2k$, are scalar analytic functions on $U_{\delta,A}$ that are real on $(A-\delta,A+\delta)$.

Let
\begin{equation}\label{definition: fb}
    f_A(z)=\left[\frac{4k+3}{4}\left(-\pi i \int_z^A\varrho(\xi)d\xi\right)\right]^{\frac{2}{4k+3}}.
\end{equation}
It follows from \eqref{eq:V-E-measure} that
\begin{equation}
f_A(z) = C_A(z-A)+\bigO(z-A)^2
\end{equation}
as $z\to A$, where
\begin{equation}\label{eq:def CA}
     C_A=\left(\frac{\pi}{2}(2A)^\frac{4k+1}{2}\psi(A)\right)^\frac{2}{4k+3}.
\end{equation}
Since $h^{(2k)}(A)>0$, we have $C_A>0$. Hence, we have defined an analytic function $f_A$ with $f_A(A)=0$ and
$f_A'(A)=C_A>0$, which is real on $(A-\delta,A+\delta)$. It is a conformal mapping on $U_{\delta,A}$ provided $\delta>0$ is sufficiently small.

Next, define
\begin{equation}\label{eq:def f1}
    f_1(z)=\left(-\pi i \int_z^A
    \varrho_{1}(\xi)d\xi\right)f_A(z)^{-\frac{1}{2}},
\end{equation}
and $f_j$ as
\begin{equation}\label{definition: fj}
    f_j(z)=\left[\frac{2j-1}{4}\left(-\pi i \int_z^A\varrho_{j}(\xi)d\xi\right)\right]f_A(z)^{-\frac{(2j-1)}{2}},~~j=2,...,2k.
\end{equation}
Since $f_A$ is a conformal mapping in $U_{\delta,A}$, it is clear from \eqref{definition: psij} and \eqref{definition: R} that $f_1$ is analytic in $U_{\delta,A}$. Furthermore, $f_1$ is real on $(A-\delta,A+\delta)$, and one can verify
\begin{equation}\label{eq:def c1}
    f_1(A)=\frac{h_1(A)}{C_A^{\frac{1}{2}}(2A)^{\frac{1}{2}}} =: c_1,
\end{equation}
where $h_1$ is defined in \eqref{definition: hj}. Similarly, $f_j$, $j=2,...,2k$ are all real and analytic in $U_{\delta,A}$. This motivates us to define 
\begin{equation}\label{eq:def cj}
    c_j:=f_j(A), \qquad j=2,...,2k.
\end{equation}
Thus, $P^{(A)}$ defined by \eqref{definition: Psingular} satisfies conditions (a) and (b) of the RH problem for $P^{(A)}$.
\begin{remark}\label{remark h}
   From Assumptions \ref{assumptions} (iii), one can see that as \(z\to A\), it gives
\[
 h_j(z)=
 \frac{h_j^{(j-1)}(A)}{(j-1)!}(z-A)^{j-1}
 +O\bigl((z-A)^j\bigr).
\]
Together with \eqref{definition: psij}, this yields
\[
 -\pi\mathrm{i}\int_z^A\varrho_j(\xi)\,d\xi
 =d_j(z-A)^{\frac{2j-1}{2}}
 \bigl(1+O(z-A)\bigr),
 \qquad d_j\neq0.
\]
Consequently, the direction \(V_j\) generates precisely the term of order
\((z-A)^{\frac{2j-1}{2}}\) in the local phase. These are exactly the \(2k\) deformation monomials in the phase of the model RH problem for \(P_{\mathrm I}^{2k}\), which ensures that the functions \(f_j\) in \eqref{eq:def f1}, \eqref{definition: fj} are analytic near \(A\), satisfy \(f_j(A)\neq0\), and further provide the natural double-scaling regimes
$t_j=O\!\left(n^{-\frac{4k-2j+4}{4k+3}}\right),
j=1,\ldots,2k.$ Thus, using \eqref{definition: theta}, \eqref{eq:definition: psist}, \eqref{definition: phist}, \eqref{definition: fb}, \eqref{eq:def f1} and \eqref{definition: fj}, we have
    \begin{equation}\label{important relation fbf1f2}
        \theta\left(n^{\frac{2}{4k+3}}f_A(z);n^{\frac{4k+2}{4k+3}}t_1f_1(z),n^{\frac{4k}{4k+3}}t_2f_2(z),...,n^{\frac{4}{4k+3}}t_{2k}f_{2k}(z)\right)=n\phi_{\mathbf{t}}(z),
    \end{equation}
    for $z\in U_{\delta,A}\setminus(A-\delta,A]$. From this and \eqref{RHP Psi: c}, it is clear that our choice of $f_A$ and $f_j$ will do the job.

\end{remark}

For condition (c), we will make use of the following proposition.
\begin{proposition}
    Let $n\to\infty$ and $t_j\to 0$, $j=1,...,2k$ such that {\rm (\ref{lim stn-1})}
    holds. Then,
    \begin{multline}\label{asymptotic expansion: Psingular}
        P^{(A)}(z)=
        E^{(A)}(z)\left(n^{\frac{2}{4k+3}}f_A(z)\right)^{-\frac{\sigma_3}{4}}N
        \\[1ex]
        \times\,\left[
            I-hf_A(z)^{-\frac{1}{2}}\sigma_3 n^{-\frac{1}{4k+3}}+\frac{1}{2}
            \begin{pmatrix}
                h^2 & iq \\
                -iq & h^2
            \end{pmatrix}f_A(z)^{-1}n^{-\frac{2}{4k+3}}+\bigO(n^{-\frac{3}{4k+3}})
        \right].
    \end{multline}
    where we have used the notation
    \[
        h=h(n^{\frac{4k+2}{4k+3}}t_1f_1(z),...,n^{\frac{4}{4k+3}}t_{2k}f_{2k}(z)),\mbox{ and }
        q=q(n^{\frac{4k+2}{4k+3}}t_1f_1(z),...,n^{\frac{4}{4k+3}}t_{2k}f_{2k}(z)).
    \]
    for brevity.
\end{proposition}

\begin{proof}
    This follows directly from \eqref{RHP Psi: c}, \eqref{definition: Psingular}, \eqref{definition: Eb}, and \eqref{important relation fbf1f2}.
\end{proof}

By symmetry, one can define
\begin{equation}\label{local at -A}
    P^{(-A)}(z)=\sigma_3P^{(A)}(-z)\sigma_3
\end{equation}
on $U_{\delta, -A}$.

\subsection{Final transformation: $S\mapsto R$}
    \label{subsection: sd-R}

\begin{figure}[t]
\begin{center}

\tikzset{every picture/.style={line width=0.75pt}} 

\begin{tikzpicture}[x=0.75pt,y=0.75pt,yscale=-1,xscale=1]

\draw   (100,142) .. controls (100,128.19) and (111.19,117) .. (125,117) .. controls (138.81,117) and (150,128.19) .. (150,142) .. controls (150,155.81) and (138.81,167) .. (125,167) .. controls (111.19,167) and (100,155.81) .. (100,142) -- cycle ;
\draw   (124.65,141.65) .. controls (124.65,141.84) and (124.8,142) .. (125,142) .. controls (125.2,142) and (125.35,141.84) .. (125.35,141.65) .. controls (125.35,141.45) and (125.2,141.29) .. (125,141.29) .. controls (124.8,141.29) and (124.65,141.45) .. (124.65,141.65) -- cycle ;
\draw    (102,131) -- (106.33,125) ;
\draw [shift={(107.09,123.95)}, rotate = 125.84] [fill={rgb, 255:red, 0; green, 0; blue, 0 }  ][line width=0.08]  [draw opacity=0] (8.93,-4.29) -- (0,0) -- (8.93,4.29) -- cycle    ;
\draw   (265,142) .. controls (265,128.19) and (276.19,117) .. (290,117) .. controls (303.81,117) and (315,128.19) .. (315,142) .. controls (315,155.81) and (303.81,167) .. (290,167) .. controls (276.19,167) and (265,155.81) .. (265,142) -- cycle ;
\draw   (289.65,141.65) .. controls (289.65,141.84) and (289.8,142) .. (290,142) .. controls (290.2,142) and (290.35,141.84) .. (290.35,141.65) .. controls (290.35,141.45) and (290.2,141.29) .. (290,141.29) .. controls (289.8,141.29) and (289.65,141.45) .. (289.65,141.65) -- cycle ;
\draw    (267,131) -- (271.33,125) ;
\draw [shift={(272.09,123.95)}, rotate = 125.84] [fill={rgb, 255:red, 0; green, 0; blue, 0 }  ][line width=0.08]  [draw opacity=0] (8.93,-4.29) -- (0,0) -- (8.93,4.29) -- cycle    ;
\draw    (138,121) .. controls (178,91) and (234.33,90) .. (271.33,125) ;
\draw [shift={(210.76,98.82)}, rotate = 182.09] [fill={rgb, 255:red, 0; green, 0; blue, 0 }  ][line width=0.08]  [draw opacity=0] (8.93,-4.29) -- (0,0) -- (8.93,4.29) -- cycle    ;
\draw    (138,164) .. controls (180.33,191) and (228.33,192) .. (271.33,160) ;
\draw [shift={(210.54,183.84)}, rotate = 177.37] [fill={rgb, 255:red, 0; green, 0; blue, 0 }  ][line width=0.08]  [draw opacity=0] (8.93,-4.29) -- (0,0) -- (8.93,4.29) -- cycle    ;
\draw    (37.33,142) -- (100,142) ;
\draw [shift={(73.67,142)}, rotate = 180] [fill={rgb, 255:red, 0; green, 0; blue, 0 }  ][line width=0.08]  [draw opacity=0] (8.93,-4.29) -- (0,0) -- (8.93,4.29) -- cycle    ;
\draw    (315,142) -- (377.33,142) ;
\draw [shift={(351.17,142)}, rotate = 180] [fill={rgb, 255:red, 0; green, 0; blue, 0 }  ][line width=0.08]  [draw opacity=0] (8.93,-4.29) -- (0,0) -- (8.93,4.29) -- cycle    ;

\draw (110,140) node [anchor=north west][inner sep=0.75pt]   [align=left] {$\displaystyle -A$};
\draw (289.65,141.65) node [anchor=north west][inner sep=0.75pt]   [align=left] {$\displaystyle A$};

\end{tikzpicture}
    \caption{The contours $\Sigma_R$ of the RH problem for $R$.}
    \label{figure: system contours R}
\end{center}
\end{figure}
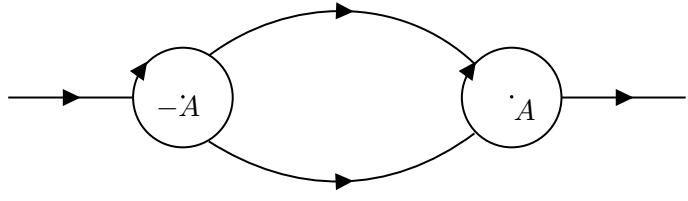

In the final transformation, we define
\begin{equation}\label{R-S-P}
    R(z)=
    \begin{cases}
        S(z)\left(P^{(-A)}\right)^{-1}(z), &\mbox{for $z\in U_{\delta,-A}\setminus\Sigma_S$,} \\[1ex]
        S(z)\left(P^{(A)}\right)^{-1}(z), &\mbox{for $z\in U_{\delta,A}\setminus\Sigma_S$,} \\[1ex]
        S(z)\left(\hat{P}^{(\infty)}\right)^{-1}(z), &\mbox{for $z\in
            \mathbb{C}\setminus (\Sigma_S\cup U_{\delta,-A}\cup U_{\delta,A})$.}
    \end{cases}
\end{equation}
Then $R$ satisfies the following RH problem.

\subsubsection*{RH problem for $R$:}

\begin{itemize}
    \item[(a)] $R:\mathbb{C}\setminus\Sigma_R\to \mathbb{C}^{2\times
    2}$ is analytic, where the contours $\Sigma_R$ are shown in Figure \ref{figure: system contours R}. 
    \item[(b)] $R_+(z)=R_-(z)v_R(z)$ for $z\in\Sigma_R$, with
    \begin{equation}\label{RHP R: b-critical}
        v_R(z)=
        \begin{cases}
            P^{(-A)}(z) \left(\hat{P}^{(\infty)}(z) \right)^{-1}, & \mbox{on $\partial
            U_{\delta,-A}$,}\\[1ex]
            P^{(A)}(z) \left(\hat{P}^{(\infty)} (z)\right)^{-1}, & \mbox{on $\partial
            U_{\delta,A}$,}\\[1ex]
            \hat{P}^{(\infty)}(z) v_S(z) \left(\hat{P}^{(\infty)} (z) \right)^{-1}, &
            \mbox{on the rest of $\Sigma_R$.}
        \end{cases}
    \end{equation}
    \item[(c)] $R(z)=I+\bigO(\frac{1}{z})$,\qquad as $z\to\infty$.
\end{itemize}

As $n\to\infty$ and $t_j\to 0$, $j=1,...,2k$ such that (\ref{lim stn-1}) holds, we
have,
\begin{equation}
    v_R(z)=
    \begin{cases}
        I+\bigO(n^{-\frac{1}{4k+3}}), &\mbox{on $\partial U_{\delta,-A}\cup \partial U_{\delta,A}$,} \\[1ex]
        I+\bigO(e^{-\gamma n}),&\mbox{on the rest of $\Sigma_R$,}
    \end{cases}
\end{equation}
with $\gamma>0$ some fixed constant. This, together with \eqref{asymptotic expansion: Psingular}, \eqref{definition: Eb} and \eqref{local at -A}, gives 
\begin{equation}\label{expansion: vR}
    v_R(z)=I+\frac{\Delta_1(z)}{n^\frac{1}{4k+3}}+\bigO(n^{-\frac{2}{4k+3}}),
\end{equation}
as $n\to\infty$ and $t_j\to 0$ such that (\ref{lim stn-1}) hold uniformly for $z\in\Sigma_R$. From  \eqref{asymptotic expansion: Psingular}, \eqref{definition: Eb} and \eqref{RHP R: b-critical}, we obtain
\eqref{expansion: vR} with
\begin{align}
    \label{definition: Delta1 near a}
    & \Delta_1(z)=-h f_A(z)^{-\frac{1}{2}}\hat{P}^{(\infty)}(z)\sigma_3
        \hat{P}^{(\infty)}(z)^{-1},
        && \mbox{for $z\in\partial U_{\delta,A}$,}
    \\[2ex]
    \label{definition: Delta1 near b}
    & \Delta_1(z)=-\sigma_3
        h f_A(-z)^{-\frac{1}{2}}\hat{P}^{(\infty)}(-z)\sigma_3
       \hat{P}^{(\infty)}(-z)^{-1}\sigma_3, && \mbox{for $z\in\partial U_{\delta,-A}$.}
\end{align}

By a standard argument of the small norm RH problem as in \cite{DKMVZ1,DKMVZ2}, one can ensure the unique solvability of the RH problem for $R$ as $n\to\infty$
and $t_j\to 0$. Moreover, we obtain from \eqref{expansion: vR} that $R$ satisfies
\begin{equation}\label{expansion: R}
    R(z)=I+\frac{R^{(1)}(z)}{n^{\frac{1}{4k+3}}}+\bigO(n^{-\frac{2}{4k+3}}),
\end{equation}
as $n\to\infty$ and $t_j\to 0$ such that \eqref{lim stn-1} holds, and it is valid uniformly for $z\in\mathbb{C}\setminus(\partial
U_{\delta,-A}\cup\partial U_{\delta,A})$. Furthermore, $R^{(1)}$ has the following explicit expression
\begin{equation}\label{definition: R(1)}
    R^{(1)}(z)=
    \begin{cases}
        {\displaystyle \frac{A^{(1)}}{z-A}-\frac{\sigma_3A^{(1)}\sigma_3}{z+A}}, &\mbox{for $z\in\mathbb C\setminus(\overline{U}_{\delta,-A}\cup
        \overline{U}_{\delta,A})$,}
        \\[3ex]
        {\displaystyle \frac{A^{(1)}}{z-A}-\frac{\sigma_3A^{(1)}\sigma_3}{z+A}-\Delta_1(z)},
        &\mbox{for $z\in U_{\delta,-A}\cup U_{\delta,A}$,}
    \end{cases}
\end{equation}
A straightforward calculation shows that
\begin{equation}\label{eq:A1}
    A^{(1)}=-\frac{\sqrt{2A}}{2\sqrt{C_A}}h\left(n^\frac{4k+2}{4k+3}t_1f_1(A),...,n^\frac{4}{4k+3}t_{2k}f_{2k}(A)\right)\begin{pmatrix}
        1 & -i\\
        -i & -1
    \end{pmatrix}.
\end{equation}

\section{Steepest descent analysis when $t\in(0,1]$}\label{Section: steepest descent t>0}

We now restrict the multi-parameter deformation to $t_1=t$, $V_1(x)=x^2-V(x)$, and $t_j\equiv0$ for $2\le j\le2k$.  The differential identity \eqref{eq:differential-identity} will then yield the asymptotics of $D_n(V_t)$.  By the assumed one-cut regularity of $V_t$ for $t\in(0,1]$, we have
\[
 d\mu_{V_t}(x)=\rho_t(x)\chi_{[-b_t,b_t]}(x)\,dx,
\]
 where the corresponding $\varrho_{t}(z)$ is defined by
\begin{equation}
        \varrho_t(z)=\frac{1}{2\pi
        i}\frac{h_j(z)}{R_t(z)},\qquad\mbox{for $z\in\mathbb C\setminus[-b_t,b_t]$,}
    \end{equation}
    with $R_t(z)=\sqrt{(z-b_t)(z+b_t)}$ and $R_t(z)\sim z$ as $z\to\infty$.
On $[-b_t,b_t]$, we have
\begin{equation}\label{eq:regular-equilibrium-density}
\rho_{t}(x) =\frac{1}{2\pi}h_t(x)\sqrt{b_t^2-x^2},~x\in[-b_t,b_t],
\end{equation}
with
\begin{equation}\label{eq:def ht}
h_t(x)=\frac{1}{\pi} \dashint_{-b_t}^{b_t}\frac{V_t'(s)}{(s-x)\sqrt{b_t^2-s^2}}\,ds,~h_t(x)>0, ~x\in[-b_t,b_t].
\end{equation}
The endpoint $b_t$ satisfies the following equation
\begin{equation}
  (1-t)M(b_t)+\pi t b_t^2=2\pi ,
  \label{eq:endpoint-equation}
\end{equation}
where
\begin{equation}
  M(b_t):=\int_{-b_t}^{b_t} \frac{xV'(x)}{\sqrt{b_t^2-x^2}}\dd x .
  \label{eq:M-def}
\end{equation}

We note that the equilibrium measure in \eqref{eq:regular-equilibrium-density} differs from the fixed-support modified measure introduced in \eqref{eq:definition: psist}. For the potential under consideration, $V_t(x) =(1-t)V(x)+tx^2$, the modified measure is given by
\[
 d\nu_{\mathbf{t}}(x)
 =\hat{\rho}_{{t}}(x)\chi_{[-A,A]}(x)\,dx,
 \qquad
 \hat{\rho}_{{t}}(x)=\rho(x)+t\rho_1(x).
\]
Here \(\rho_1\) is defined in \eqref{definition: psij} and is the fixed-support measure corresponding to
\(V_1=x^2-V\). Hence \(\nu_{\mathbf{t}}\) is a signed measure on the fixed interval \([-A,A]\), whereas \(\mu_{V_t}\) is the positive
equilibrium measure on the moving interval \([-b_t,b_t]\). These two measures coincide at \(t=0\) and differ for \(t>0\). Specifically, for every compact set \(K\subset(-A,A)\), we have
\begin{equation}\label{eq:density-bulk-comparison}
 \rho_{t}(x)
 =\hat\rho_t(x)
 +O\!\left(t^{1+\frac1{2k+1}}\right),
 \qquad x\in K,
\end{equation}
uniformly as \(t\to0\). Therefore, in the case of $t\in(0,1]$, the RH analysis is slightly different from that in Section \ref{Section:OP and Y}. To be specific, the global parametrix in the two RH analysis have the same form, with $A$ being replaced by $b_t$. However, we use the standard Airy model RH problem to construct the local parametrices in this section, whereas in Section \ref{Section:OP and Y} we used the model RH problem of the $P_{\rm I}^{2k}$ equation.

\subsection{First transformation: $Y\mapsto T$}
For $z\in\mathbb C\setminus(-\infty,b_t]$, set
\begin{equation}\label{eq:regular-g}
  g_t(z)=\int_{-b_t}^{b_t}\log(z-s)\rho_t(s)\,\dd s,
  \qquad
  \varphi_t(z)=-\pi i\int_z^{b_t}\varrho_t(\xi)\,\dd\xi.
\end{equation}
The Euler--Lagrange conditions imply
\begin{equation}
    g_{t,+}(x)+g_{t,-}(x)-V_t(x)-\ell_t=0,~
 -2\varphi_{t,+}(x)=2\varphi_{t,-}(x)=g_{t,+}(x)-g_{t,-}(x), \label{eq:g++g-}
\end{equation}
on $(-b_t,b_t)$.
We use the same normalization as in the previous section
\begin{equation}\label{eq:regular-T}
 T(z)=e^{-\frac{n\ell_t}{2}\sigma_3}Y(z)e^{-ng_t(z)\sigma_3}
 e^{\frac{n\ell_t}{2}\sigma_3}.
\end{equation}
Then $T(z)=I+\mathcal O(z^{-1})$ at infinity and its jump matrix is
\begin{equation}\label{eq:regular-T-jump}
 J_T(x)=
 \begin{cases}
 \begin{pmatrix}e^{2n\varphi_{t,+}(x)}&1\\0&e^{2n\varphi_{t,-}(x)}\end{pmatrix},
      &x\in(-b_t,b_t),\\[3ex]
 \begin{pmatrix}1&e^{-2n\varphi_{t,-}(x)}\\0&1\end{pmatrix},
      &x\in\mathbb R\setminus (-b_t,b_t).
 \end{cases}
\end{equation}

\subsection{Second transformation: $T\mapsto S$}
On $(-b_t,b_t)$, the first matrix in \eqref{eq:regular-T-jump} factors as
\begin{equation}\label{eq:regular-factorization}
\begin{pmatrix}e^{2n\varphi_{t,+}(x)}&1\\0&e^{2n\varphi_{t,-}(x)}\end{pmatrix}
=
\begin{pmatrix}1&0\\e^{2n\varphi_{t,-}(x)}&1\end{pmatrix}
\begin{pmatrix}0&1\\-1&0\end{pmatrix}
\begin{pmatrix}1&0\\e^{2n\varphi_{t,+}(x)}&1\end{pmatrix}.
\end{equation}
\begin{figure}[t]
    \begin{center}
    \setlength{\unitlength}{1mm}
    \begin{picture}(137.5,26)(-2.5,11.5)

        \put(45,25){\thicklines\circle*{.8}} \put(42.5,27){\small $-b_t$}
        \put(85,25){\thicklines\circle*{.8}} \put(86,27){\small $b_t$}

        \put(94,25){\thicklines\vector(1,0){.0001}}
        \put(29.6,25){\line(1,0){70.4}} \put(38,25){\thicklines\vector(1,0){.0001}}
        \put(66,25){\thicklines\vector(1,0){.0001}}

        \qbezier(45,25)(65,45)(85,25) \put(66,35){\thicklines\vector(1,0){.0001}}
        \qbezier(45,25)(65,5)(85,25) \put(66,15){\thicklines\vector(1,0){.0001}}

    \end{picture}
    \caption{The contours $\Sigma_S$ of the RH problem for $S$. \label{Fig:S}}
\end{center}
\end{figure}
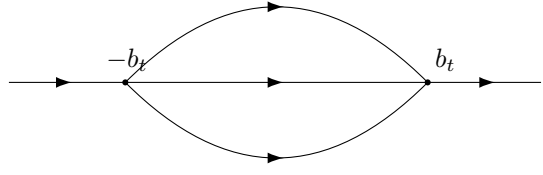
We open the lenses $\gamma_+$ and $\gamma_-$ around $[-b_t,b_t]$ as in
Figure~\ref{Fig:S}, and define
\begin{equation}\label{eq:regular-S}
S(z)=T(z)
\begin{cases}
\begin{pmatrix}1&0\\-e^{2n\varphi_t(z)}&1\end{pmatrix},
   &z\in\gamma_+,\\[3ex]
\begin{pmatrix}1&0\\ e^{2n\varphi_t(z)}&1\end{pmatrix},
   &z\in\gamma_-,\\[3ex]
I,&\text{ elsewhere}.
\end{cases}
\end{equation}
Then the jump matrices for $S$ are given by
\begin{equation}\label{eq:regular-S-jumps}
J_S(z)=
\begin{cases}
\begin{pmatrix}0&1\\-1&0\end{pmatrix},&z\in(-b_t,b_t),\\[3ex]
\begin{pmatrix}1&e^{-2n\varphi_{t,-}(z)}\\0&1\end{pmatrix},&z\in\mathbb R\setminus (-b_t,b_t),\\[3ex]
\begin{pmatrix}1&0\\e^{2n\varphi_t(z)}&1\end{pmatrix},&z\in\gamma_+\cup\gamma_-.
\end{cases}
\end{equation}

\subsection{Global parametrix}
As in the construction of $\hat{P}^{(\infty)}$ in Section~\ref{sec:hat Pinfty},
the global parametrix is explicit.  Define
$a_t(z)=\Big( \frac{z-b_t}{z+b_t}\Big)^{\frac{1}{4}}$, analytic on
$\mathbb{C}\setminus [-b_t,b_t]$ and normalized by $a_t(z)\sim1$ as
$z\to\infty$.  Then
\begin{equation}\label{Pinf Region 1}
P^{(\infty)}(z) =M^{-1}a_t(z)^{-\sigma_3}M  ,\qquad\mbox{for $z\in\mathbb C\setminus[-b_t,b_t]$,}
\end{equation}
with $M=\frac{I+i\sigma_1}{\sqrt{2}}$.

\subsection{Local parametrices near $\pm b_t$}
Let
$U_{\delta,b_t}=\{z\in\mathbb{C}:|z-b_t|<\delta\}$ be a small disk with
center $b_t$ and radius $\delta>0$ sufficiently small, and $t\ge t_0$ for a fixed $t_0$. Inside $U_{\delta,b_t}$, the local parametrix $P^{(b_t)}$ satisfies the following RH problem.
\subsubsection*{RH problem for $P^{(b_t)}$}
\begin{itemize}
\item[(a)] $P^{(b_t)} : U_{\delta,b_t}\setminus \Sigma_S \to \mathbb{C}^{2\times 2}$ is analytic, where the contours $\Sigma_S=\mathbb{R} \cup \gamma_{+} \cup \gamma_{-}$ are shown in Figure \ref{Fig:S}.
\item[(b)] $P^{(b_t)}$ has the following jumps:
\begin{align}
& P_{+}^{(b_t)}(z) = P_{-}^{(b_t)}(z)J_S(z),~z\in U_{\delta,b_t}\cap\Sigma_S.
\end{align}
\item[(c)] As $n \to \infty$, we have $P^{(b_t)}(z) = (I + \bigO(n^{-1}))P^{(\infty)}(z)$ uniformly for $z \in \partial U_{\delta,b_t}$.
\end{itemize}
We denote
\begin{equation}\label{eq:def ft}
f_{t}(z) = \left( \frac{3\pi i}{2} \int_{b_t}^{z} \varrho_t(s)ds \right)^{\frac{2}{3}}.
\end{equation}
This is a conformal map from $U_{\delta,b_t}$ to a neighborhood of $0$ and
\begin{equation}\label{eq:asymptotics of f near b}
f_{t}(z) = C_t(z-b_t)+\frac{2}{5}C_t\left(\frac{h_t'(b_t)}{h_t(b_t)}+\frac{1}{4b_t}\right)(z-b_t)^2+ \bigO((z-b_t)^{3}) , \quad \mbox{as } z \to b_t.
\end{equation}
with $C_t=(\frac{1}{2}h_t(b_t)\sqrt{2b_t})^{\frac{2}{3}}$. We choose the lenses such that $f_{t}(\gamma_{+}\cap U_{\delta,b_t}) \subset e^{\frac{2\pi i}{3}}\mathbb{R}^{+}$ and $f_{t}(\gamma_{-}\cap U_{\delta,b_t}) \subset e^{-\frac{2\pi i}{3}}\mathbb{R}^{+}$. The solution of the above RH problem is given by the standard Airy parametrix
\begin{equation}
P^{(b_t)}(z) = E_{b_t}(z) \Phi_{\mathrm{Ai}}(n^{\frac{2}{3}}f_{t}(z))e^{n\varphi_t(z)\sigma_{3}},
\end{equation}
where $E_{b_t}$ is analytic in $U_{\delta,b_t}$ and is given by
\begin{equation}
E_{b_t}(z) = P^{(\infty)}(z)M^{-1} (n^{\frac{2}{3}}f_{t}(z))^{\frac{\sigma_{3}}{4}},
\end{equation}
with $M=\frac{I+i\sigma_1}{\sqrt{2}}$. One shows with a direct calculation together with \eqref{Asymptotics Airy}, that as $n \to \infty$, uniformly for $z \in \partial U_{\delta,b_t}$, $P^{(b_t)}(z)P^{(\infty)}(z)^{-1}$ admits the following expansion
\begin{equation}\label{asy D_b}
P^{(b_t)}(z)P^{(\infty)}(z)^{-1} = I + \frac{P^{(\infty)}(z)\Phi_{\mathrm{Ai,1}}P^{(\infty)}(z)^{-1}}{ (f_{t}(z))^{\frac{3}{2}}}\frac{1}{n}+ \bigO ( n^{-2} ).
\end{equation}
By the symmetry, we have
\begin{equation}\label{eq:asy D_-b}
    P^{(-b_t)}(z)=\sigma_3P^{(b_t)}(-z)\sigma_3
\end{equation}
on $U_{\delta, -b_t}$

\subsection{Small-norm RH problem}\label{Subsection: small-norm RH problem}
In the final transformation,
set
\begin{equation}\label{def of R}
    R(z)=
    \begin{cases}
        S(z)\left(P^{(-b_t)}\right)^{-1}(z), &\mbox{for $z\in U_{\delta,-b_t}\setminus\Sigma_S$,} \\[1ex]
        S(z)\left(P^{(b_t)}\right)^{-1}(z), &\mbox{for $z\in U_{\delta,b_t}\setminus\Sigma_S$,} \\[1ex]
        S(z)\left({P}^{(\infty)}\right)^{-1}(z), &\mbox{for $z\in
            \mathbb{C}\setminus (\Sigma_S\cup U_{\delta,-b_t}\cup U_{\delta,b_t})$.}
    \end{cases}
\end{equation}
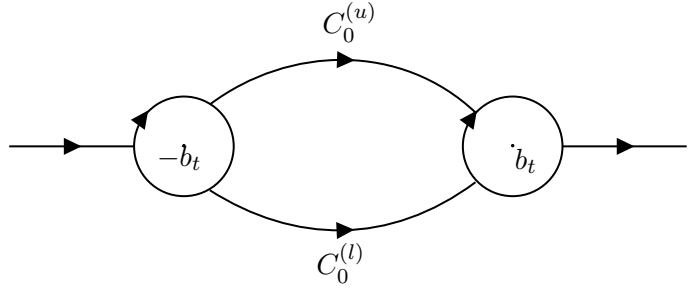
\begin{figure}[t]
\begin{center}

\tikzset{every picture/.style={line width=0.75pt}} 

\begin{tikzpicture}[x=0.75pt,y=0.75pt,yscale=-1,xscale=1]

\draw   (100,142) .. controls (100,128.19) and (111.19,117) .. (125,117) .. controls (138.81,117) and (150,128.19) .. (150,142) .. controls (150,155.81) and (138.81,167) .. (125,167) .. controls (111.19,167) and (100,155.81) .. (100,142) -- cycle ;
\draw   (124.65,141.65) .. controls (124.65,141.84) and (124.8,142) .. (125,142) .. controls (125.2,142) and (125.35,141.84) .. (125.35,141.65) .. controls (125.35,141.45) and (125.2,141.29) .. (125,141.29) .. controls (124.8,141.29) and (124.65,141.45) .. (124.65,141.65) -- cycle ;
\draw    (102,131) -- (106.33,125) ;
\draw [shift={(107.09,123.95)}, rotate = 125.84] [fill={rgb, 255:red, 0; green, 0; blue, 0 }  ][line width=0.08]  [draw opacity=0] (8.93,-4.29) -- (0,0) -- (8.93,4.29) -- cycle    ;
\draw   (265,142) .. controls (265,128.19) and (276.19,117) .. (290,117) .. controls (303.81,117) and (315,128.19) .. (315,142) .. controls (315,155.81) and (303.81,167) .. (290,167) .. controls (276.19,167) and (265,155.81) .. (265,142) -- cycle ;
\draw   (289.65,141.65) .. controls (289.65,141.84) and (289.8,142) .. (290,142) .. controls (290.2,142) and (290.35,141.84) .. (290.35,141.65) .. controls (290.35,141.45) and (290.2,141.29) .. (290,141.29) .. controls (289.8,141.29) and (289.65,141.45) .. (289.65,141.65) -- cycle ;
\draw    (267,131) -- (271.33,125) ;
\draw [shift={(272.09,123.95)}, rotate = 125.84] [fill={rgb, 255:red, 0; green, 0; blue, 0 }  ][line width=0.08]  [draw opacity=0] (8.93,-4.29) -- (0,0) -- (8.93,4.29) -- cycle    ;
\draw    (138,121) .. controls (178,91) and (234.33,90) .. (271.33,125) ;
\draw [shift={(210.76,98.82)}, rotate = 182.09] [fill={rgb, 255:red, 0; green, 0; blue, 0 }  ][line width=0.08]  [draw opacity=0] (8.93,-4.29) -- (0,0) -- (8.93,4.29) -- cycle    ;
\draw    (138,164) .. controls (180.33,191) and (228.33,192) .. (271.33,160) ;
\draw [shift={(210.54,183.84)}, rotate = 177.37] [fill={rgb, 255:red, 0; green, 0; blue, 0 }  ][line width=0.08]  [draw opacity=0] (8.93,-4.29) -- (0,0) -- (8.93,4.29) -- cycle    ;
\draw    (37.33,142) -- (100,142) ;
\draw [shift={(73.67,142)}, rotate = 180] [fill={rgb, 255:red, 0; green, 0; blue, 0 }  ][line width=0.08]  [draw opacity=0] (8.93,-4.29) -- (0,0) -- (8.93,4.29) -- cycle    ;
\draw    (315,142) -- (377.33,142) ;
\draw [shift={(351.17,142)}, rotate = 180] [fill={rgb, 255:red, 0; green, 0; blue, 0 }  ][line width=0.08]  [draw opacity=0] (8.93,-4.29) -- (0,0) -- (8.93,4.29) -- cycle    ;

\draw (110,140) node [anchor=north west][inner sep=0.75pt]   [align=left] {$\displaystyle -b_{t}$};
\draw (289.65,141.65) node [anchor=north west][inner sep=0.75pt]   [align=left] {$\displaystyle b_{t}$};
\draw (193,68) node [anchor=north west][inner sep=0.75pt]   [align=left] {$\displaystyle C_{0}^{( u)}$};
\draw (190,189) node [anchor=north west][inner sep=0.75pt]   [align=left] {$\displaystyle C_{0}^{( l)}$};

\end{tikzpicture}

    \caption{The contours $\Sigma_R$ of the RH problem for $R$.}
    \label{figure: contours R}
\end{center}
\end{figure}
Then, $R_+(z) = R_-(z)J_R(z)$ for $z \in \Sigma_R$ (see Figure \ref{figure: contours R}),  with the jump matrices given below:
\begin{equation}\label{RHP R: b-regular}
 J_R(z)=
 \begin{cases}
 P^{(-b_t)}(z)P^{(\infty)}(z)^{-1},&z\in\partial U_{\delta,-b_t},\\[1ex]
 P^{(b_t)}(z)P^{(\infty)}(z)^{-1},&z\in\partial U_{\delta,b_t},\\[1ex]
 P^{(\infty)}(z)J_S(z)P^{(\infty)}(z)^{-1},&\text{ elsewhere.}
 \end{cases}
\end{equation}
By the standard results for small-norm RH problems (see \cite{DKMVZ1,DKMVZ2} for instance), we obtain 
\begin{equation}\label{eq:uniform-R-prop}
     R(z)=I+\frac{R^{(1)}(z)}{n}+\mathcal O(n^{-2}),
\end{equation}
uniformly for $z\in\mathbb{C}\setminus(\partial
U_{\delta,-b_t}\cup\partial U_{\delta,b_t})$ and $t\ge t_0$.
We now compute the coefficient $R^{(1)}(z)=R^{(1)}(z;t)$ in \eqref{eq:uniform-R-prop}.  Note that the expression for $J_{R}^{(1)}(z)$ with $z\in \partial U_{\delta, b_t} \cup \partial U_{\delta, -b_t}$ can be analytically continued on $\overline{U_{\delta, b_t}} \cup \overline{U_{\delta, -b_t}}$, except at $-b_t$ and $b_t$, where the expressions in \eqref{asy D_b} and \eqref{eq:asy D_-b} admit poles. These poles are of order $2$ at $-b_t$ and $b_t$. Therefore, for $z\in\mathbb C\setminus\overline{U_{\delta, b_t}} \cup \overline{U_{\delta, -b_t}}$, $R^{(1)}(z)$ is given by
\begin{equation}\label{R^{(1)}}
\begin{array}{r c l}
\displaystyle R^{(1)}(z)  =  \displaystyle\frac{A_{-1}}{z-b_t} + \frac{A_{-2}}{(z-b_t)^{2}} 
  \displaystyle -\frac{\sigma_3A_{-1}\sigma_3}{z+b_t} + \frac{\sigma_3A_{-2}\sigma_3}{(z+b_t)^{2}} .  \\
\end{array}
\end{equation}
where
\begin{equation}
    A_{-1}=\mbox{Res}\big(J_{R}^{(1)}(s),s=b_t\big),~A_{-2}=\mbox{Res}\big((s-b_t)J_{R}^{(1)}(s),s=b_t\big).
\end{equation}
Recalling \eqref{Pinf Region 1},
\eqref{eq:asymptotics of f near b}, and \eqref{asy D_b}, a residue calculation gives us
\begin{equation}\label{eq:A-2}
A_{-2} = \frac{5}{48h_t(b_t)} \begin{pmatrix}
-1 & i \\ i & 1
\end{pmatrix} 
\end{equation}
and
\begin{equation}\label{eq:A-1}
A_{-1} = \frac{1}{h_t(b_t)}\begin{pmatrix}
\frac{1}{16}\left(\frac{h_t'(b_t)}{h_t(b_t)}+\frac{1}{b_t}\right) & i \left(-\frac{1}{16}\frac{h_t'(b_t)}{h_t(b_t)}+\frac{1}{12b_t}\right) \\
i \left(-\frac{1}{16}\frac{h_t'(b_t)}{h_t(b_t)}+\frac{1}{12b_t}\right)  & -\frac{1}{16}\left(\frac{h_t'(b_t)}{h_t(b_t)}+\frac{1}{b_t}\right)
\end{pmatrix}.
\end{equation}

Note that the asymptotic expansion \eqref{eq:uniform-R-prop} holds when $t \geq t_0 >0$. We now turn to its behavior as $t \to 0$. First, from \eqref{eq:H-expansion} and \eqref{eq:bt-expansion}, we have 
 \begin{equation}\label{eq:ht and ht'}
     h_t(b_t)\sim \bigO(t^{\frac{2k}{2k+1}}),\quad \frac{h_t'(b_t)}{h_t(b_t)}\sim \bigO(t^{-\frac{1}{2k+1}}),\quad C_t\sim \bigO(t^\frac{4k}{6k+3}), \quad \textrm{as } t\to 0.
 \end{equation} 
 Consequently, the expansion of $f_t(z)$ in \eqref{eq:asymptotics of f near b} becomes 
 \begin{equation}\label{eq:uniform ft}
     f_t(z)=C_t(z-b_t)\left[1+\bigO(|z-b_t|t^{-\frac{1}{2k+1}})\right],\quad \quad \textrm{as } z\to b_t.
 \end{equation}
 On the circle $\partial U_{\delta,b_t}$ for a fixed $\delta$, the term $\delta t^{-\frac{1}{2k+1}}\to\infty$, as $t\to0$.  Therefore, we need to consider a shrinking neighborhood for $U_{\delta_t,b_t}=\{z\in\mathbb{C}:|z-b_t|<\delta_t=\delta t^{\frac{1}{2k+1}}\}$. In this shrinking region, 
 the estimate in \eqref{eq:uniform-R-prop} is modified as detailed in the following proposition.

\begin{proposition}
\label{prop:uniform-R}
Fix
\begin{equation}\label{eq:overlap-exponent-range}
 \frac{3(2k+1)}{6k+5}<a<\frac{4k+2}{4k+3},
\end{equation}
  then the RH problem for $R$ defined in \eqref{def of R} is uniquely solvable and has the following expansion
\begin{equation}
 R(z)=I+\frac{R^{(1)}(z;t)}{n}+\mathcal O(n^{-2}t^{-\frac{8k+5}{4k+2}}),
 \qquad
 R'(z)=\frac{(R^{(1)})'(z;t)}{n}+\mathcal O(n^{-2}t^{-\frac{8k+5}{4k+2}}),
\end{equation}
uniformly for $n^{-a}\le t\le1$ and $z\in K$, where $\operatorname{dist}(K,\Sigma_R)\ge\eta>0$ and the contours $\Sigma_R$ are shown in Figure \ref{figure: contours R}.
\end{proposition}

\begin{proof}
From \eqref{eq:regular-g}, we have 
\begin{equation}
    \mbox{Re}\,\varphi_t(z) 
=-\frac{|\sigma|}{2}
h_t(s)\sqrt{b_t^2-s^2}\times
\left[
1+O\!\left(
\frac{\sigma^2}{(b_t-|s|)^2}
\right)
\right],~z = s + i\sigma \in  C_{0}^{(u)} \cup C_{0}^{(l)},
\end{equation}where $C_{0}^{(u,l)}$ are the parts of the curves $\gamma_{\pm}$ that lie outside of the disks $U_{\delta_t,b_t}$ and $U_{\delta_t,-b_t}$, as shown in Figure \ref{figure: contours R}. We choose $z=s\pm i\epsilon|b_t-s|$ for $\epsilon$ sufficiently small. Using Lemma \ref{lem:hb}, we have
\[
\mbox{Re} \, \varphi_t(z)
\le-c|z\mp b_t|^{\frac{4k+3}{2}}, \quad z \in C_{0}^{(u)} \cup C_{0}^{(l)},
\]
\[\mbox{Re} \, \varphi_{t,-}(x)
\ge c|x\mp b_t|^{\frac{4k+3}{2}},\quad x\in(-\infty,-b_t-\delta_t]\cup[b_t+\delta_t,\infty),\]
which implies
\begin{equation}\label{uniform5}
|J_{R}(z) - I | \leq Ce^{-c_{0}nt^{\frac{4k+3}{4k+2}}}, \quad \mbox{for all}\quad    z\in\Sigma_R\setminus(\partial
U_{\delta_t,-b_t}\cup\partial U_{\delta_t,b_t}).
\end{equation}
For $t \ge n^{-a}$ with the parameter $a$ given by \eqref{eq:overlap-exponent-range}, $J_R(z)$ approaches the identity matrix exponentially fast as $n \to \infty$.

We next consider the jump matrices on the boundary circles $\partial U_{\delta_t, b_t} \cup \partial U_{\delta_t, -b_t}$. Combining \eqref{Pinf Region 1}, \eqref{asy D_b}, \eqref{eq:asy D_-b}, and \eqref{eq:uniform ft}, we have
\begin{align}\label{asymptotics for the jumps of JR}
 J_{R}(z) = I+ J_{R}^{(1)}(z;t)n^{-1} + 
\bigO\left(
n^{-2}
t^{-\frac{8k+7}{4k+2}}
\right), \quad |J_R^{(1)}(z;t)|
\le Ct^{-\frac{2k+1}{2k+2}},
\end{align}
uniformly for $z \in \partial U_{\delta_t, b_t} \cup \partial U_{\delta_t, -b_t}$ as $n\to\infty$. This gives us
\begin{equation}\label{eq:est JR}
    \|J_R-I\|_{L^\infty(\partial U_{\delta_t,b_t}\cup\partial U_{\delta_t,-b_t})}
=
\bigO\!\left(
n^{-1}t^{-\frac{2k+2}{2k+1}}
\right),\quad
\|J_R-I\|_{L^1(\partial U_{\delta_t,b_t}\cup\partial U_{\delta_t,-b_t})}
=
\bigO\!\left(
n^{-1}t^{-1}
\right).
\end{equation}
These estimates imply that a standard small-norm RH problem for $R$ can only be established for $0<a<\frac{2k+1}{2k+2}$, which falls short of the desired range in \eqref{eq:overlap-exponent-range}.

To overcome this restriction, we adopt the constant conjugation technique from \cite[Section 4]{DIK} and introduce two global conjugations of $R$. That is, let $\Gamma_r=\partial U_{\delta_t,b_t},
\Gamma_l=\partial U_{\delta_t,-b_t},$
and define 
\begin{equation}\label{eq:relation hatR and R}
    \widehat R_{j}(z)=D_j^{-1}MR(z)M^{-1}D_j,\quad  j\in\{r,l\},\quad z\in\mathbb{C}\setminus\Sigma_R,
\end{equation}
with $M=\frac{I+i\sigma_1}{\sqrt{2}}$, $D_r=t^{-\frac{\sigma_3}{8k+4}}$ and $D_l=D_r^{-1}$. Both \(\widehat R_r(z)\) and \(\widehat R_l(z)\) possess jumps on \(\Sigma_R\), denoted by $\widehat J_{R,j}(z)$, which satisfy
\[
\widehat J_{R,j}(z)
=
D_j^{-1}MJ_R(z)M^{-1}D_j, \qquad \textrm{for \(j\in\{r,l\}\)}.
\]
In the RH problem for $\widehat R_j$, the jump matrices $\widehat J_{R,j}$ on $\Sigma_R\setminus(\Gamma_r\cup\Gamma_l)$ remain exponentially decaying. Meanwhile, on the corresponding circle $\Gamma_j$, conjugating \eqref{asymptotics for the jumps of JR} yields
\begin{equation}
   \widehat J_{R,j}(z)
=
I+\frac{\widehat J_{R,j}^{(1)}(z;t)}{n}
+
\bigO\left(
n^{-2}t^{-\frac{4k+3}{2k+1}}
\right), 
\qquad
z\in\Gamma_j.
\end{equation}
Moreover, from \eqref{eq:ht and ht'} and \eqref{eq:uniform ft}, we obtain
\begin{equation}\label{eq:hat JR}
    |\widehat J_{R,j}^{(1)}(z;t)|
\le Ct^{-\frac{4k+3}{4k+2}},\quad z\in\Gamma_j.
\end{equation}
Since
\(
|\Gamma_j|=\bigO\left(t^{\frac1{2k+1}}\right),
\)
it follows that
\begin{equation}\label{eq:L2 for hatJ}
    \left\|
\widehat J_{R,j}-I
\right\|_{L^2(\Gamma_j)}
=
\bigO(n^{-1}t^{-1}).
\end{equation}
Consequently, with $0<a<\frac{4k+2}{4k+3}$, for all $t\ge n^{-a}$, we have $\left\|
\widehat J_{R,j}-I
\right\|_{L^\infty(\Gamma_j)}
=
\bigO(n^{-1}t^{-\frac{4k+3}{4k+2}})=\bigO(n^{\frac{4k+3}{4k+2}a-1})=o(1)$.

Now, we consider $\widehat{R}_j(z)$ for \(j\in\{r,l\}\). First, we have (cf. \cite{DKMVZ1,DKMVZ2})
\begin{equation}\label{eq:int equ for hatR}
\widehat{R}_j(z)-I
=
\frac{1}{2\pi i}
\int_{\Sigma_R}
\frac{\widehat{R}_{j,-}(s)\bigl(\widehat{J}_{R,j}(s)-I\bigr)}{s-z}\,ds,
\qquad
z\in\mathbb{C}\setminus\Sigma_R.
\end{equation}
We next estimate the boundary values of $\widehat{R}_r$ on $\Gamma_r$ and those of $\widehat{R}_l$ on $\Gamma_l$. This analysis is essential because each conjugation is specifically tailored to optimize the jump estimates exclusively on its corresponding circle $\Gamma_j$. Combining these two separate localized $L^2$-estimates on $\Gamma_r$ and $\Gamma_l$ enables us to deduce the desired asymptotics for the original matrix $R(z)$. In what follows, we illustrate this procedure using $\widehat{R}_r$ as an example. Decomposing the contour $\Sigma_R$ in \eqref{eq:int equ for hatR} yields the representation
\begin{multline} \label{eq:coupled}
 \widehat R_r(z)=I
 +\frac1{2\pi i}
\int_{\Gamma_r}\frac{\widehat {R}_{r,-}(s)(\widehat J_{R,r}(s)-I)}{s-z}\,ds+D_r^{-2}\left[\frac1{2\pi i}
\int_{\Gamma_l}\frac{\widehat {R}_{l,-}(s)(\widehat J_{R,l}(s)-I)}{s-z}\,ds\right]D_r^2\\
+D_r^{-1}M\left[\frac1{2\pi i}
\int_{\Sigma_R\setminus(\Gamma_l\cup\Gamma_r)}\frac{{R}_{-}(s)( J_{R}(s)-I)}{s-z}\,ds\right]M^{-1}D_r.
\end{multline}
Since the Cauchy projection on each endpoint circle is uniformly bounded on
$L^2$, and the distance between $\Gamma_r$ and $\Gamma_l$ is bounded below by a
positive constant, we have
\begin{multline}  \label{eq:cross-cauchy}
\left\|\frac1{2\pi i}\int_{\Gamma_\ell}
                \frac{ (\widehat R_{l,-}-I)(\widehat J_{R,l}-I)}{s-z}\,d s\right\|_{L^2(\Gamma_r)}
 \le C|\Gamma_r|^{1/2}|\Gamma_\ell|^{1/2}
          \|{(\widehat R_{l,-}-I)(\widehat J_{R,l}-I)}\|_{L^2(\Gamma_\ell)}\\
 \le Ct^{\frac{1}{2k+1}} \|\widehat R_{l,-}-I\|_{L^2(\Gamma_\ell)}\|\widehat J_{R,l}-I\|_{L^\infty(\Gamma_\ell)}.
\end{multline}
Combining \eqref{uniform5}, \eqref{eq:hat JR}, \eqref{eq:L2 for hatJ}, \eqref{eq:coupled}, \eqref{eq:cross-cauchy} and the fact $\|D_r\|=\|D_l\|=\bigO\left(t^{-\frac1{8k+4}}\right)$, we obtain
\begin{equation}\label{eq:L2 for Rr-}
    \|\widehat R_{r,-}-I\|_{L^2(\Gamma_r)}\le Cn^{-1}t^{-1}+Cn^{-1}t^{-\frac{4k+3}{4k+2}} \left(\|\widehat R_{r,-}-I\|_{L^2(\Gamma_r)}+\|\widehat R_{l,-}-I\|_{L^2(\Gamma_l)}\right).
\end{equation}
Similarly, we have
\begin{equation}\label{eq:L2 for Rl-}
    \|\widehat R_{l,-}-I\|_{L^2(\Gamma_l)}\le Cn^{-1}t^{-1}+Cn^{-1}t^{-\frac{4k+3}{4k+2}}  \left(\|\widehat R_{r,-}-I\|_{L^2(\Gamma_r)}+\|\widehat R_{l,-}-I\|_{L^2(\Gamma_l)}\right).
\end{equation}
Here, the exponentially small terms are omitted. The coefficient of the second term on the right of \eqref{eq:L2 for Rr-} and \eqref{eq:L2 for Rl-}
tends to zero uniformly, since $ n^{-1}t^{-\frac{4k+3}{4k+2}}
 \le n^{-1+a\frac{4k+3}{4k+2}}\to0$ when $0<a<\frac{4k+2}{4k+3}$.
We can therefore obtain
\begin{equation}\label{eq:L2 for hatR}
    \|\widehat R_{j,-}-I\|_{L^2(\Gamma_j)}\le Cn^{-1}t^{-1},\quad z\in\Gamma_j.
\end{equation}
In addition, $\left\|
R_--I
\right\|_{L^2(\Sigma_R\setminus(\Gamma_r\cup\Gamma_l))}
=
\mathcal O\!\left(
n^{-1}t^{-\frac{4k+3}{4k+2}}
\right).$

We now estimate the remainder term $ R(z)-I-\frac{ R^{(1)}(z)}n$ for $R$. Using \eqref{eq:int equ for hatR}, we have
\begin{align}
    \label{eq:int equ for R}
 R(z)-I-\frac{R^{(1)}(z)}n
=
\frac1{2\pi i}
\int_{\Sigma_R}\frac{ {\mathcal{E}}(s)}{s-z}\,ds,\quad  z\in\mathbb{C}\setminus\Sigma_R,
\end{align}
where \({\mathcal{E}}=
\left( J_{R}-I-\frac{ J_{R}^{(1)}}{n}\right)
+( R_{-}-I)( J_{R}-I).\) 
On $\Gamma_j$, $j\in\{r,l\}$, the conjugation
$D_j^{-1}M(\cdot)M^{-1}D_j$ introduced in \eqref{eq:relation hatR and R} gives us
\begin{align}
\widehat {\mathcal{E}}_j=D_j^{-1}M\mathcal E M^{-1}D_j
=
\left(
\widehat J_{R,j}-I
-\frac{\widehat J_{R,j}^{(1)}}{n}
\right)+
(\widehat R_{j,-}-I)
(\widehat J_{R,j}-I).
\label{eq:adapted-remainder-direct}
\end{align}
Applying Hölder's inequality alongside \eqref{eq:hat JR}, \eqref{eq:L2 for hatJ}, and \eqref{eq:L2 for hatR}, we find
\begin{align*}
\|\widehat {\mathcal{E}}_j\|_{L^1(\Sigma_j)} \le \left\|\widehat J_{R,j}-I-\frac{\widehat J_{R,j}^{(1)}}{n}\right\|_{L^1(\Sigma_j)}+\|\widehat R_{j,-}-I\|_{L^2(\Gamma_j)}\left\|
\widehat J_{R,j}-I
\right\|_{L^2(\Gamma_j)}
\le Cn^{-2}t^{-2}.
\end{align*}
Therefore, $\left\|
D_j^{-1}M\mathcal E M^{-1}D_j
\right\|_{L^1(\Gamma_j)}
=
\mathcal O(n^{-2}t^{-2}).$
Since $\|D_j\|,\,\|D_j^{-1}\|
=
t^{-\frac1{4k+2}},$
we conclude that
\begin{equation}\label{eq:original-remainder-circle-simple}
\|\mathcal E\|_{L^1(\Gamma_j)}
=
\mathcal O\!\left(
n^{-2}t^{-\frac{8k+5}{4k+2}}
\right),
\qquad
j\in\{r,l\}.
\end{equation}
On the remaining contour $\Sigma_R\setminus(\Gamma_r\cup\Gamma_l)$, a similar computation yields  $\|\mathcal E\|_{L^1(\Sigma_R\setminus(\Gamma_r\cup\Gamma_l))}
=
o\!\left(
n^{-2}t^{-\frac{8k+5}{4k+2}}
\right)$, which finally implies $\|\mathcal{E}\|_{L^1(\Sigma_R)}=\bigO(n^{-2}t^{-\frac{8k+5}{4k+2}})$. 

For any $K \subset \mathbb{C} \setminus \Sigma_R$ satisfying
$
\operatorname{dist}(K,\Sigma_R)\ge\eta>0
$, using representation \eqref{eq:int equ for R}, we have
\begin{equation} 
\left|R(z)
-
I-\frac{R^{(1)}(z;t)}n\right|\le \|\mathcal{E}\|_{L^1(\Sigma_R)}\le Cn^{-2}t^{-\frac{8k+5}{4k+2}},
\end{equation}
uniformly for all $z\in K$ and $t\ge n^{-a}$. Finally, for any choice of $a$ in the range \eqref{eq:overlap-exponent-range}, we have $\bigO\left(\frac{n^{-2}t^{-\frac{8k+5}{4k+2}}}{t^{-1}n^{-1}}\right)=o(1)$, as $n\to\infty$. This completes the proof the proposition.
\end{proof}

\section{Proof of Theorem \ref{Main theorem}}\label{sec:proof 1.3}


We first introduce the differential identity obtained in \cite{Charlier}
\begin{align}\label{eq:differential-identity}
\partial_{t} \log D_{n}(V_{t}) =  \frac{1}{2\pi i}\int_{\mathbb{R}}[Y^{-1}(x)Y^{\prime}(x)]_{21}\partial_{t}w_{t}(x)dx,
\end{align}
where in our case, we have
\begin{equation}
w_{t}(x) = e^{-nV_{t}(x)} \qquad \mbox{and} \qquad V_t(x)=(1-t)V(x)+tx^2.
\end{equation} 
In the following, we integrate \eqref{eq:differential-identity} from $t$ to $1$, separating the regular regime $t\ge n^{-a}$ from the critical regime $t\le n^{-a}$ as in \cite{BWW,BI,Charlier,CFLW}. We first integrate on the regular side.

\subsection{Integration in $V$ when $t\in[n^{-a},1]$}\label{Section: integrating V}

Using the jump relations \eqref{jump relations of Y} of $Y$, the differential identity \eqref{eq:differential-identity} becomes
\begin{align}\label{lol 1}
\partial_{t} \log D_{n}(V_{t})={}&
\frac{1}{2\pi i}\int_{\mathbb{R}\setminus[-b_t-\epsilon,b_t+\epsilon]}
 [Y^{-1}(x)Y^{\prime}(x)]_{21}\,\partial_t w_t(x)\,dx \notag\\
&-\frac{1}{2\pi i}\int_{\mathcal C}
 [Y^{-1}(z)Y^{\prime}(z)]_{11}\,\partial_t\log w_t(z)\,dz.
\end{align}
where $\epsilon >0$ is fixed and $\mathcal{C}$ is a closed curve surrounding $[-b_t,b_t]$ and the lenses $\gamma_{+}\cup\gamma_{-}$, is oriented clockwise and passes through $-b_t-\epsilon$ and $b_t+\epsilon$.

For $z$ outside the lenses, we have
\begin{equation}\label{Y near inf 2}
Y(z) = e^{\frac{n\ell_{t}}{2}\sigma_{3}}R(z)P^{(\infty)}(z)e^{ng_t(z)\sigma_{3}}e^{-\frac{n\ell_{t}}{2}\sigma_{3}},
\end{equation}
and thus, by \eqref{variationalcondition:must-inequality1} and \eqref{eq:g++g-}, one has
\begin{equation}
[Y^{-1}(x)Y^{\prime}(x)]_{21} \partial_{t} w_{t}(x) = \bigO(e^{-cn}), \qquad \mbox{ as } n \to \infty,
\end{equation}
uniformly for $x \in \mathbb{R}\setminus [-b_t-\epsilon,b_t+\epsilon]$ and $t \in (0,1]$ , where $c >0$ is a fixed constant. We can explicitly compute $[Y^{-1}(z)Y^{\prime}(z)]_{11}$ using \eqref{Y near inf 2}. Thus, as $n\to\infty$, equation \eqref{lol 1} becomes
\begin{equation}\label{eq:DI-regular}
\begin{array}{r c l}
\displaystyle \partial_{t} \log D_{n}(V_{t})  & = & \displaystyle I_{1,t} + I_{2,t} + I_{3,t} + \bigO(e^{-cn}), \\[0.35cm]
\displaystyle I_{1,t} & = & \displaystyle \frac{-n}{2\pi i} \int_{\mathcal{C}} g_t^{\prime}(z) \partial_{t} \log w_{t}(z) dz, \\[0.35cm]
\displaystyle I_{2,t} & = & \displaystyle \frac{-1}{2\pi i} \int_{\mathcal{C}} [P^{(\infty)}(z)^{-1}P^{(\infty)}(z)^{\prime}]_{11} \partial_{t} \log w_{t}(z) dz, \\[0.35cm]
\displaystyle I_{3,t} & = & \displaystyle \frac{-1}{2\pi i} \int_{\mathcal{C}} [P^{(\infty)}(z)^{-1}R^{-1}(z)R^{\prime}(z)P^{(\infty)}(z)]_{11} \partial_{t} \log w_{t}(z) dz.
\end{array}
\end{equation}
 By \eqref{property g: 2}, a direct calculation shows that
\begin{equation}\label{eq:I1-regular}
I_{1,t} = n^{2} \int_{-b_t}^{b_t} \big(V(x)-x^{2}\big)\rho_{t}(x)dx,
\end{equation}
where $\rho_t(x)$ is defined in \eqref{eq:regular-equilibrium-density}.
Then it follows that
\begin{equation}
    I_{1,t}=n^2[(1-t)P(b_t)+tS(b_t)]=n^2[P(b_t)+t(S(b_t)-P(b_t))],
\end{equation}
with $P(b)$ and $S(b)$ are defined in \eqref{eq:P-definition} and \eqref{eq:S-U-definition}, 
respectively.
As $t\to0$, combining \eqref{eq:bt-expansion}, \eqref{eq:P-expansion}, and \eqref{eq:S-expansion}, we have
\begin{equation}
    \frac{I_{1,t}}{n^2}=p_0+C_{1}t+C_2t^\frac{2k+2}{2k+1}+\bigO(t^\frac{2k+3}{2k+1}),~t\to0,
\end{equation}
with
\begin{equation}
    p_0=\int_{-A}^A(V(x)-x^2)\rho(x)\,dx.
\end{equation}
In addition,
\begin{align}
  C_{1}
  &=\frac1\pi\int_{-A}^{A}(V(x)-x^2)\sqrt{A^2-x^2}\dd x
  -\int_{-A}^{A}(V(x)-x^2)\psi(x)(A^2-x^2)^{2k+\frac12}\dd x
  \notag\\
  &\quad
  +\frac{2-A^2}{2\pi}\int_{-A}^{A}\frac{V(x)-x^2}{\sqrt{A^2-x^2}}\dd x,
  \label{eq:def-C1}
\end{align}
and
\begin{equation}\label{eq:def C2}
  C_{2}=-\frac{(2k+1)}{(2k+2)}\alpha_k^{-\frac{1}{2k+1}}f_1(A)|f_1(A)|^{\frac{2k+2}{2k+1}},
\end{equation}
where $\alpha_k$ and $f_1(A)$ are defined in \eqref{eq:def alpha k} and \eqref{eq:def c1}, respectively.
Then
\begin{equation}\label{eq:asy I1,t}
    \int_{n^{-a}}^1I_{1,t}dt=\left(C_0-p_0n^{-a}-\frac{1}{2}C_1n^{-2a}-\frac{2k+1}{4k+3}C_2n^{-\frac{4k+3}{2k+1}a}+\bigO(n^{-\frac{4k+4}{2k+1}a})\right)n^2,
\end{equation}
with
\begin{equation}\label{eq:def C0}
    C_0=\int_0^1\int_{-b_s}^{b_s}(V(u)-u^2)\rho_s(u)\,du\,ds.
\end{equation}
The explicit global parametrix in \eqref{Pinf Region 1} gives
\[
 (P^{(\infty)}(z))^{-1}(P^{(\infty)}(z))'
 =-\frac{a_t'(z)}{a_t(z)}M^{-1}\sigma_3M.
\]
Since $[M^{-1}\sigma_3M]_{11}=0$, the second contribution vanishes, i.e., $ I_{2,t}=0.$
Proposition \ref{prop:uniform-R} gives
\begin{equation}\label{eq:I3-regular-R1}
I_{3,t} = \frac{1}{2\pi i} \int_{\mathcal{C}} [P^{(\infty)}(z)^{-1}R^{(1)}(z)^{\prime}P^{(\infty)}(z)]_{11}\,\partial_{t}V_{t}(z)\,dz
 + \bigO\!\left(n^{-1}t^{-\frac{8k+5}{4k+2}}\right),
\end{equation}
uniformly for $n^{-a}\le t\le1$.  
Using \eqref{Pinf Region 1}, it becomes, as $n \to \infty$
\begin{multline}\label{lol 12}
I_{3,t} = -\frac{1}{2\pi i} \int_{\mathcal{C}} \left( \frac{a_t(z)^{2}+a_t(z)^{-2}}{4}[R_{11}^{(1)}(z)^{\prime} - R_{22}^{(1)}(z)^{\prime}] + \frac{1}{2}[R_{11}^{(1)}(z)^{\prime}+R_{22}^{(1)}(z)^{\prime}] \right. \\ \left. - i \frac{a_t(z)^{-2}-a_t(z)^{2}}{4} [R_{12}^{(1)}(z)^{\prime} + R_{21}^{(1)}(z)^{\prime}] \right) \big(V(z)-z^{2}\big)dz
+ \bigO\!\left(n^{-1}t^{-\frac{8k+5}{4k+2}}\right).
\end{multline}
Substituting \eqref{eq:A-2} and \eqref{eq:A-1} into \eqref{R^{(1)}} gives
\begin{align*}
& R_{11}^{(1)\prime}(z)-R_{22}^{(1)\prime}(z)
 =\frac{5}{12h_t(b_t)}\left[ \frac{1}{(z-b_t)^{3}}+\frac{1}{(z+b_t)^3}\right] \notag\\
&\hspace{1.4cm}+\frac{1}{8h_t(b_t)}\left(\frac{h_t'(b_t)}{h_t(b_t)}+\frac{1}{b_t}\right)
 \left[\frac{1}{(z+b_t)^{2}}-\frac{1}{(z-b_t)^2}\right],\nonumber \\
& R_{11}^{(1)\prime}(z) + R_{22}^{(1)\prime}(z) = 0, \\
& R_{12}^{(1)\prime}(z)+R_{21}^{(1)\prime}(z)
 =\frac{5i}{12h_t(b_t)}\left[ \frac{1}{(z+b_t)^{3}}- \frac{1}{(z-b_t)^3}\right] \notag\\
&\hspace{1.4cm}+\frac{i}{h_t(b_t)}\left(\frac{h_t'(b_t)}{8h_t(b_t)}-\frac{1}{6b_t}\right)
 \left[\frac{1}{(z+b_t)^{2}}+ \frac{1}{(z-b_t)^2}\right].\nonumber 
\end{align*}
Put $\mathfrak r_t(z):=(z^2-b_t^2)^{\frac{1}{2}}$ with the branch fixed by $\mathfrak r_t(z)\sim z$ as $z\to\infty$. Then
\begin{equation}
    [P^{(\infty)}(z)^{-1}R^{(1)}(z)^{\prime}P^{(\infty)}(z)]_{11}
    =\frac{b_t^2}{4h_t(b_t)\mathfrak r_t(z)^5}
     -\frac{b_th_t'(b_t)}{8h_t(b_t)^2\mathfrak r_t(z)^3}.
\end{equation}
With $Q(z)=V(z)-z^2$, this gives
\begin{equation}
\begin{array}{r c l}
  \displaystyle I_{3,t}&=&\displaystyle\frac{1}{2\pi i}\left(-\frac{b_t^2}{4h_t(b_t)}\int_{\mathcal{C}} \frac{Q(z)}{\mathfrak r_t(z)^5}\,dz+\frac{b_th_t'(b_t)}{8h_t(b_t)^2}\int_{\mathcal{C}} \frac{Q(z)}{\mathfrak r_t(z)^3}\,dz\right)\\
  &:=&\displaystyle-\frac{b_t^2}{4h_t(b_t)}K_1+\frac{b_th_t'(b_t)}{8h_t(b_t)^2}K_2.
\end{array}
\end{equation}
Since $\mathfrak r_t'(z)=\frac{z}{\mathfrak r_t(z)}$, we have
\begin{equation}
    \left(\frac{zQ(z)}{\mathfrak r_t(z)}\right)'=\frac{zQ'(z)}{\mathfrak r_t(z)}-\frac{b_t^2Q(z)}{\mathfrak r_t(z)^3},\qquad
    \left(\frac{zQ(z)}{\mathfrak r_t(z)^3}\right)'=\frac{zQ'(z)}{\mathfrak r_t(z)^3}-\frac{2Q(z)}{\mathfrak r_t(z)^3}-\frac{3b_t^2Q(z)}{\mathfrak r_t(z)^5}.
\end{equation}
Thus,
\begin{equation}
K_2=\frac{1}{b_t^2}\frac{1}{2\pi i}\int_{\mathcal{C}}\frac{zQ'(z)}{\mathfrak r_t(z)}\,dz,\qquad
    3b_t^2K_1=\frac{1}{2\pi i}\int_{\mathcal{C}}\frac{zQ'(z)}{\mathfrak r_t(z)^3}\,dz-2K_2
\end{equation}
Integrating by parts, we have
\begin{equation}
    K_1=\frac{1}{3b_t^4}\frac{1}{2\pi i}\int_{\mathcal{C}}\frac{z^2Q''(z)-zQ'(z)}{\mathfrak r_t(z)}\,dz.
\end{equation}
 Using
$\mathfrak r_{t,\pm}(x)=\pm i\sqrt{b_t^2-x^2}$, we have
\begin{equation}\label{eq:K12-real-integrals}
    K_2=-\frac{1}{\pi b_t^2}\int_{-b_t}^{b_t}\frac{xQ'(x)}{\sqrt{b_t^2-x^2}}\,\dd x,
    \qquad
    K_1=-\frac{1}{3\pi b_t^4}\int_{-b_t}^{b_t}\frac{x^2Q''(x)-xQ'(x)}{\sqrt{b_t^2-x^2}}\,\dd x.
\end{equation}
The endpoint identity
\eqref{eq:M-prime-identity} reads
\[
 M'(b_t)=\frac{\pi b_t}{1-t}\bigl(h_t(b_t)-2t\bigr).
\]
Together with \eqref{eq:endpoint-equation}, this yields
\begin{equation}\label{eq:K12-endpoint}
    K_2=-\frac{2-b_t^2}{b_t^2(1-t)},
    \qquad
    K_1=-\frac{b_t^2h_t(b_t)-4}{3b_t^4(1-t)}.
\end{equation}
Then we obtain the following expansion:
\begin{equation}\label{lol 3}
I_{3,t} = \frac{1}{1-t}\left(\frac{b_t^2h_t(b_t)-4}{12b_t^2h_t(b_t)}-\frac{(2-b_t^2)h_t'(b_t)}{8b_th_t(b_t)^2}\right)
+ \bigO\!\left(n^{-1}t^{-\frac{8k+5}{4k+2}}\right),
\qquad n\to\infty.
\end{equation}
Define the $n$-independent regular contribution
\begin{equation}\label{eq:I3-regular-main}
 \mathcal I_3(t):=\frac{1}{1-t}\left(
 \frac{b_t^2h_t(b_t)-4}{12b_t^2h_t(b_t)}
 -\frac{(2-b_t^2)h_t'(b_t)}{8b_th_t(b_t)^2}\right).
\end{equation}
Then \eqref{lol 3} and the endpoint expansions imply
\begin{equation}\label{eq:uniform I3}
 I_{3,t}=-\frac{k}{6(2k+1)t}
 +\mathcal O\!\left(t^{-\frac{2k}{2k+1}}\right)
 +\mathcal O\!\left(n^{-1}t^{-\frac{8k+5}{4k+2}}\right),
 \qquad n^{-a}\le t\le1.
\end{equation}
The finite part is consequently defined without any $n$-dependent error by
\begin{equation}\label{eq:Creg}
 C_{\mathrm{reg}}:=\lim_{\epsilon\downarrow0}
 \left[\int_\epsilon^1\mathcal I_3(s)\,\dd s
 -\frac{k}{6(2k+1)}\log\epsilon\right].
\end{equation}
Integrating \eqref{eq:uniform I3} gives
\begin{equation}\label{eq:integration of I_3,s,1}
\int_{n^{-a}}^{1} I_{3,s}\,\dd s
= C_{\mathrm{reg}}-\frac{ak}{6(2k+1)}\log n
 +\mathcal O\!\left(n^{-\frac{a}{2k+1}}+n^{-1+a\frac{4k+3}{4k+2}}\right).
\end{equation}
We state the following proposition, which gives an explicit expression for $C_{\rm reg}$.
\begin{proposition}
    Assume $A^2\ne2$, we have
        \begin{equation}
C_{\mathrm{reg}}
=
\frac{1}{24}
\log\left\{
\frac{A^{4}}{16}
\left[
2\pi\psi(A)\frac{(4k+1)!!}{(2k)!}
\right]^{\frac{2}{2k+1}}
\left[
(2k+1)\lvert 2-A^{2}\rvert
\right]^{\frac{4k}{2k+1}}
\right\}.
\label{eq:Creg-explicit}
\end{equation}
\end{proposition}
\begin{proof}
    Differentiating \eqref{eq:endpoint-equation} and using \eqref{eq:M-prime-identity}, we have
   \begin{equation}
\label{eq:bt-derivative-Creg}
    \frac{d}{dt}b_t
    =
    \frac{2-b_t^2}
    {(1-t)b_t h_t(b_t)}.
\end{equation}
Moreover, differentiation of the Cauchy representation of $h_t(b_t)$ defined in \eqref{eq:hb-PV}, together with the evenness of the potential, yields
\begin{equation}
\label{eq:ht-endpoint-derivative}
    \frac{d}{dt}h_t(b_t)
    =
    \frac{2-h_t(b_t)}{1-t}
    +
    \frac{3}{2}b_t'h_t'(b_t).
\end{equation}
Here $h_t'(b_t)
:=
\left.\frac{\partial h_t(x)}{\partial x}\right|_{x=b_t}.$ Thus, using \eqref{eq:bt-derivative-Creg} and \eqref{eq:ht-endpoint-derivative}, it can be checked that
\begin{equation}
\label{eq:I3-Kt-identity}
    I_3(t)
    =
    -\frac{1}{24}\frac{d}{dt}
    \log\left(
        \frac{b_t^4h_t^2(b_t)}{16}
    \right).
\end{equation}

Combining the definition of $C_{\rm reg}$ in \eqref{eq:Creg}, \eqref{eq:H-expansion}, \eqref{eq:H0}, \eqref{eq:bt-expansion} and \eqref{eq:b1}, we obtain \eqref{eq:Creg-explicit}.
\end{proof}

\subsection{Integration in $V$ when $t\in[0,n^{-a}]$}\label{Section: integrating V to 0}
The case when $t\in[0,n^{-a}]$ is similar to Section \ref{Section: integrating V}. We first show the following lemma that we will use later.

\begin{lemma}
\label{lem:critical-overlap-error}
Let $a$ satisfy \eqref{eq:overlap-exponent-range} and put
$s_n(t):=n^{\frac{4k+2}{4k+3}} f_1(A)t$. Then $R(z)$ in \eqref{R-S-P} and $R'(z)$ satisfy
\begin{equation}\label{eq:critical-overlap-R}
 R(z)=I+n^{-\frac{1}{4k+3}}R^{(1)}(z;s_n(t),0)
 +\mathcal O\!\left(n^{-\frac{2}{4k+3}}(1+|s_n(t)|)^{\frac{4k+4}{2k+1}}\right),
\end{equation}
\begin{equation}\label{eq:critical-overlap-R'}
 R'(z)=n^{-\frac{1}{4k+3}}(R^{(1)})'(z;s_n(t),0)
 +\mathcal O\!\left(n^{-\frac{2}{4k+3}}(1+|s_n(t)|)^{\frac{4k+4}{2k+1}}\right),
\end{equation}
uniformly for $0\leq t\leq n^{-a}$ and $z\in\mathbb{C}\setminus(\partial
U_{\delta,-A}\cup\partial U_{\delta,A})$.
\end{lemma}

\begin{proof} 
Adopting Lemma~\ref{lem:uniform-P2k-Psi} with
$|s|=|s_n(t)|$, it provides the residual jump estimate
\[
 J_R-I-n^{-\frac{1}{4k+3}}\Delta_1
 =\mathcal O\!\left(n^{-\frac{2}{4k+3}}(1+|s_n(t)|)^{\frac{4k+4}{2k+1}}\right).
\]
The remaining jumps are exponentially small, as stated by the strict variational inequality. The small-norm singular-integral equation and Cauchy's formula
then give \eqref{eq:critical-overlap-R} and \eqref{eq:critical-overlap-R'}.
\end{proof}
As $n \to \infty$, we also have
\begin{equation}\label{eq:DI-critical}
\begin{array}{r c l}
\displaystyle \partial_{t} \log D_{n}(V_{t})  & = & \displaystyle \hat{I}_{1,t} + \hat{I}_{2,t} + \hat{I}_{3,t} + \bigO(e^{-cn}), \\[0.3cm]
\displaystyle \hat{I}_{1,t} & = & \displaystyle \frac{-n}{2\pi i} \int_{\Omega} g^{\prime}(z) \partial_{t} \log w_{t}(z) dz, \\[0.3cm]
\displaystyle \hat{I}_{2,t} & = & \displaystyle \frac{-1}{2\pi i} \int_{\Omega} [\hat{P}^{(\infty)}(z)^{-1}\hat{P}^{(\infty)}(z)^{\prime}]_{11} \partial_{t} \log w_{t}(z) dz, \\[0.3cm]
\displaystyle \hat{I}_{3,t} & = & \displaystyle \frac{-1}{2\pi i} \int_{\Omega} [\hat{P}^{(\infty)}(z)^{-1}R^{-1}(z)R^{\prime}(z)\hat{P}^{(\infty)}(z)]_{11} \partial_{t} \log w_{t}(z) dz,
\end{array}
\end{equation}
where $\Omega$ is a closed curve surrounding $[-A,A]$ and the lenses, is oriented clockwise, and passes through $-A-\epsilon$ and $A+\epsilon$.
 By \eqref{eq:definition: psist} and \eqref{property g: 2}, a direct calculation shows that
\begin{equation}\label{eq:I1-critical}
\begin{array}{r c l}
\displaystyle \widehat I_{1,t}=n^2\int_{-A}^{A}(V(x)-x^2)\rho(x)\,dx
+t n^2\int_{-A}^{A}(V(x)-x^2)\rho_1(x)\,dx,
\end{array}
\end{equation}
where $\rho$ is defined in \eqref{definition: h0},  \(\rho_1\) is defined in \eqref{definition: psij} and is the fixed-support measure corresponding to
\(V_1=x^2-V\). By integrating it in $s$ from $t$ to $n^{-a}$, it gives
\begin{align}\label{eq:asy hat I1,t}
\int_t^{n^{-a}}\hat I_{1,s}\,ds={}&(n^{-a}-t)n^2
 \int_{-A}^{A}\bigl(V(x)-x^2\bigr)\rho(x)\,dx \notag\\
&+\frac{n^2}{2}(n^{-2a}-t^2)
 \int_{-A}^{A}\bigl(V(x)-x^2\bigr)\rho_1(x)\,dx.
\end{align}
Similarly, $\widehat I_{2,t}=0$ because the explicit outer parametrix satisfies
$[\hat P^{(\infty)}{}^{-1}(\hat P^{(\infty)})']_{11}=0$. By Lemma~\ref{lem:critical-overlap-error},
\begin{equation}\label{eq:critical-I3-leading}
\hat{I}_{3,t} = \frac{1}{2\pi i} \int_{\Omega}
 [\hat{P}^{(\infty)}(z)^{-1}R^{(1)}(z)^{\prime}\hat{P}^{(\infty)}(z)]_{11}
 \partial_{t}V_{t}(z)\,\dd z\, n^{\frac{4k+2}{4k+3}}
 +\mathcal O\!\left(n^{\frac{4k+1}{4k+3}}
 (1+|s_n(t)|)^{\frac{4k+4}{2k+1}}\right).
\end{equation}
Using \eqref{definition: R(1)}, \eqref{eq:A1} and the global parametrix $\hat{P}^{(\infty)}(z)$ defined in \eqref{definition: Pinfty} gives
\begin{equation}
    [\hat{P}^{(\infty)}(z)^{-1}R^{(1)}(z)^{\prime}\hat{P}^{(\infty)}(z)]_{11}=\frac{(2A)^\frac{3}{2}}{2\sqrt{C_A}}h(n^\frac{4k+2}{4k+3}f_1(A)t,0)(z^2-A^2)^{-\frac{3}{2}}.
\end{equation}
From \eqref{eq:J-prime}, one obtains
\[
 \frac{1}{2\pi i}\int_{\Omega}\frac{V(z)-z^2}{(z^2-A^2)^{3/2}}\,\dd z
 =\frac{J'(A)}{A}=\frac{2-A^2}{A^2}.
\]
Then \eqref{eq:critical-I3-leading} becomes
\begin{equation}\label{eq:def KA}
    \hat{I}_{3,t}=K_Ah(s_n(t),0)n^\frac{4k+2}{4k+3}
    +\mathcal O\!\left(n^\frac{4k+1}{4k+3}(1+|s_n(t)|)^{\frac{4k+4}{2k+1}}\right),
    \quad K_A=\frac{\sqrt{2A}(2-A^2)}{A\sqrt{C_A}},
\end{equation}
where $C_A$ is defined in \eqref{eq:def CA}. From \eqref{eq:rho1}, we have $K_A=-2f_1(A)$. We then make the change of variables in the integration as follows:
\begin{equation}
    t=n^{-\frac{4k+2}{4k+3}}x,~~s=n^{-\frac{4k+2}{4k+3}}y.
\end{equation}
Then Lemma~\ref{lem:critical-overlap-error} and the change of variables above give
\begin{equation}\label{eq:asy hat I3}
\begin{aligned}
    \int_t^{n^{-a}}\hat{I}_{3,s}\,\dd s
    ={}&\int_x^{n^{\frac{4k+2}{4k+3}-a}}K_Ah(f_1(A)y,0)\,\dd y +\mathcal O\!\left(
    n^{-\frac1{4k+3}}
    \bigl(1+n^{\frac{4k+2}{4k+3}-a}\bigr)^{\frac{6k+5}{2k+1}}
    \right).
\end{aligned}
\end{equation}
The estimate is uniform for $0\le t\le n^{-a}$ for every $ \frac{3(2k+1)}{6k+5}<a<\frac{4k+2}{4k+3}$.  The lower bound in
\eqref{eq:overlap-exponent-range} is exactly equivalent to the vanishing of
the error in \eqref{eq:asy hat I3}; it also places the matching circle in the
uniform sector of Lemma~\ref{lem:uniform-P2k-Psi}.
  Then
\begin{equation}\label{eq:int h}
\begin{aligned}
\int_x^{n^{\frac{4k+2}{4k+3}-a}} K_A h\bigl(f_1(A)y,0\bigr)\,\dd y
={}&\int_x^{n^{\frac{4k+2}{4k+3}-a}}
      K_A\widehat h\bigl(f_1(A)y,0\bigr)\,\dd y \\
&\hspace{-2.5cm} +\frac{K_A}{2f_1(A)}\frac{k}{12(2k+1)}
  \log\!\left(
    \frac{1+f_1(A)^2n^{\frac{4(2k+1)}{4k+3}-2a}}
         {1+f_1(A)^2x^2}
  \right) \\
&\hspace{-2.5cm} +\frac{(2k+1)^2}{(4k+4)(4k+3)}
  K_A\alpha_k^{-\frac{1}{2k+1}}
  |f_1(A)|^\frac{2k+2}{2k+1}
  \left[
    n^{2-\frac{4k+3}{2k+1}a}
    -x^\frac{4k+3}{2k+1}
  \right],
\end{aligned}
\end{equation}
where $f_1(A)$ is defined in \eqref{eq:def c1}. Combining with \eqref{eq:def alpha k}, \eqref{eq:def C2}, \eqref{eq:def KA} and \eqref{eq:b1}, it can be checked that
\begin{equation}
    \frac{(2k+1)^2}{(4k+4)(4k+3)}K_A\alpha_k^{-\frac{1}{2k+1}}|f_1(A)|^\frac{2k+2}{2k+1}=\frac{2k+1}{4k+3}C_2.
\end{equation}
Finally, we obtain
\begin{eqnarray}
   && \log D_{n}(x^2) - \log D_{n}(V_{t}(x)) = \int_t^1(I_{1,s}+I_{3,s})ds \nonumber \\
 && \hspace{2cm} =  \int_t^{n^{-a}}(\hat{I}_{1,s}+\hat{I}_{3,s})ds+\int_{n^{-a}}^1(I_{1,s}+I_{3,s})ds \nonumber \\
 &&\hspace{2cm} =\int_t^{n^{-a}}\widehat I_{3,s}\,ds+\int_t^{n^{-a}}\widehat I_{1,s}\,ds+\int_{n^{-a}}^1(I_{1,s}+I_{3,s})\,ds.
\end{eqnarray}
 Combining \eqref{eq:asy I1,t}, \eqref{eq:integration of I_3,s,1},
\eqref{eq:asy hat I1,t}, \eqref{eq:asy hat I3}, and \eqref{eq:int h}, the
terms proportional to $p_0n^{2-a}$ and $C_1n^{2-2a}$ cancel by
\eqref{eq:iden-rho1}; the power term at the matching point cancels by the
preceding identity involving $C_2$.  The logarithmic contributions combine as
\[
 -\frac{1}{12}\log n+\frac{k}{3(4k+3)}\log n
 =-\frac{1}{4(4k+3)}\log n
\]
after the Gaussian normalization is inserted.  The remaining error terms are
$o(1)$ because $a$ satisfies \eqref{eq:overlap-exponent-range}. This completes the proof of Theorem~\ref{Main theorem}.

\section{Proof of Theorem \ref{theorem: universality}}
    \label{section: universality}

The Christoffel--Darboux representation for the orthogonal-polynomial RH
problem, together with \eqref{TinY} and \eqref{property phi: 2}, gives the
following formula for the two-point kernel $K_n^{(\mathbf t)}$; cf.\
\cite{CK,CKV,ClaeysVan}:
\begin{equation}\label{KinY}
    K_n^{(\mathbf t)}(x,y) = e^{-n\phi_{\mathbf{t},+}(x)} e^{-n\phi_{\mathbf{t},+}(y)} \frac{1}{2\pi i(x-y)}
        \begin{pmatrix}
            0 & 1
        \end{pmatrix}
        T_+^{-1}(y)T_+(x)
        \begin{pmatrix}
            1 \\
            0
        \end{pmatrix}
\end{equation}
for $x,y\in\mathbb{R}$.
Using (\ref{SinT}) and (\ref{R-S-P}), we have
\begin{equation}\label{KinPhat}
    K_n^{(\mathbf t)}(x,y)
        = e^{-n\phi_{\mathbf t,+}(x)} e^{-n\phi_{\mathbf{t},+}(y)} \frac{1}{2\pi i(x-y)}
        \begin{pmatrix}
            0 & 1
        \end{pmatrix}
        \widehat P^{-1}(y)R^{-1}(y)R(x)\widehat P(x)
        \begin{pmatrix}
            1 \\
            0
        \end{pmatrix},
\end{equation}
for $x\in(A-\delta,A+\delta)$, where
\begin{equation}\label{definition: hatP}
    \widehat P(x)=
        \begin{cases}
            P^{(A)}_+(x),&\mbox{on $(A,A+\delta)$,}
            \\[1ex]
            P^{(A)}_+(x)
            \begin{pmatrix}
                1 & 0 \\
                e^{2n\phi_{\mathbf{t},+}(x)} & 1
            \end{pmatrix}, &\mbox{on $(A-\delta,A)$.}
        \end{cases}
\end{equation}
Further, we define
\begin{equation}\label{definition: hatPsi}
    \widehat\Psi(x)=
        \begin{cases}
            \Psi_+(x)&\mbox{on $\mathbb R_+$,}
            \\[1ex]
            \Psi_+(x)
            \begin{pmatrix}
                1 & 0 \\
                1 & 1
            \end{pmatrix},&\mbox{on $\mathbb R_-$,}
        \end{cases}
\end{equation}
where $\Psi$ is the solution of the RH problem for $\Psi$, see
Section \ref{subsection: PI2 equation}. Using (\ref{definition: Psingular}),
(\ref{definition: hatP}) and (\ref{definition: hatPsi}), a
straightforward calculation yields,
\[
    \widehat P(x)=
        E^{(A)}(x) \widehat\Psi\left(n^{\frac{2}{4k+3}}f_A(x);n^{\frac{4k+2}{4k+3}}t_1f_1(x),...,n^{\frac{4}{4k+3}}t_{2k}f_{2k}(x)\right)
        e^{n\phi_{\mathbf{t},+}(x)\sigma_3}
\]
for $x\in(A-\delta,A+\delta)$.
Inserting this into (\ref{KinPhat}) we then obtain
{\small
\begin{multline}\label{KnN: PsiER}
    K_n^{(\mathbf t)}(x,y)=\frac{1}{2\pi i(x-y)}
        \begin{pmatrix}
            0 & 1
        \end{pmatrix}
         \widehat\Psi^{-1}\left(n^{\frac{2}{4k+3}}f_A(y);n^{\frac{4k+2}{4k+3}}t_1f_1(y),...,n^{\frac{4}{4k+3}}t_{2k}f_{2k}(y)\right)
         \\[1ex]
         \times (E^{(A)})^{-1}(y)R^{-1}(y) R(x)E^{(A)}(x)\widehat
        \Psi\left(n^{\frac{2}{4k+3}}f_A(x);n^{\frac{4k+2}{4k+3}}t_1f_1(x),...,n^{\frac{4}{4k+3}}t_{2k}f_{2k}(x)\right)
        \begin{pmatrix}
            1 \\
            0
        \end{pmatrix},
\end{multline}}
for $x\in(A-\delta,A+\delta)$.

Let
\begin{equation}\label{eq:scaled-uv}
    u_n= A+\frac{u}{c n^{\frac{2}{4k+3}}},\quad\mbox{and}\quad v_n=A+\frac{v}{cn^{\frac{2}{4k+3}}}, \quad \mbox{with
    $c=C_A$.}
\end{equation}
We then have,
\begin{equation}\label{uv1}
    \lim_{n\to\infty}n^{\frac{2}{4k+3}}f_A(u_n)=u, \qquad\mbox{and}\qquad
    \lim_{n\to\infty}n^{\frac{2}{4k+3}}f_A(v_n)=v.
\end{equation}
Furthermore, since $f_j(A)=c_j$, $j=1,...,2k$ (see (\eqref{eq:def c1}) and (\eqref{eq:def cj})) we have 
\begin{equation}
   \lim_{\substack{n\to\infty\\ t_j\to0}} n^{\frac{4k-2j+4}{4k+3}}t_jf_j(u_n)=\hat{t}_j, \quad  \lim_{\substack{n\to\infty\\ t_j\to0}} n^{\frac{4k-2j+4}{4k+3}}t_jf_j(v_n)=\hat{t}_j,\quad j=1,...,2k. \label{eq:scaled-parameters}
\end{equation}
Now, a similar argument as in \cite{KV2} shows that
\begin{equation}\label{uv4}  \lim_{\substack{n\to\infty\\ t_j\to0}}
    (E^{(A)})^{-1}(v_n)R(v_n)^{-1}R(u_n) (E^{(A)})(u_n)=I.
\end{equation}
Inserting \eqref{eq:scaled-parameters} and \eqref{uv4} into (\ref{KnN: PsiER}), we obtain
\begin{align}
    \nonumber
    &  \lim_{\substack{n\to\infty\\ t_j\to0}} \frac{1}{C_A n^{\frac{2}{4k+3}}}K_n^{(\mathbf t)}(u_n,v_n) =\frac{1}{-2\pi
i(u-v)}\Big(\Psi_1^{(2k)}(u;\hat{t}_1,\ldots,\hat{t}_{2k})\Psi_2^{(2k)}(v;\hat{t}_1,\ldots,\hat{t}_{2k}) \\
    \nonumber
    & \hspace{6cm} -\Psi_1^{(2k)}(v;\hat{t}_1,\ldots,\hat{t}_{2k})\Psi_2^{(2k)}(u;\hat{t}_1,\ldots,\hat{t}_{2k})\Big).
\end{align}
This completes the proof of Theorem
\ref{theorem: universality}.

\appendix


\section{Cauchy operators and endpoint expansions}
Throughout this appendix, $V$ is an even real-analytic external field.  Fix a positively oriented simple closed contour $\Gamma$ in a common complex neighborhood to which $V$ and the analytic continuation of $\psi$ extend. We also take $\Gamma$ so that $[-b,b]\subset\operatorname{Int}(\Gamma)$ for all
$b$ sufficiently close to $A$. For $b>0$, set
\begin{equation}\label{eq:Rb-definition}
  R_b(z)=\sqrt{z^2-b^2},\qquad R_b(z)\sim z,\qquad z\to\infty,
\end{equation}
with branch cut $[-b,b]$. Thus
\begin{equation}\label{eq:Rb-boundary-values}
  R_{b,+}(x)=i\sqrt{b^2-x^2},\qquad
  R_{b,-}(x)=-i\sqrt{b^2-x^2},
  \qquad x\in(-b,b).
\end{equation}
Define (see \cite{Muskhelishvili1992})
\begin{equation}\label{eq:hb-definition}
  h_b(z)=h_V^{(b)}(z):=\frac{1}{2\pi i}\oint_\Gamma
  \frac{V'(\xi)}{R_b(\xi)}\frac{d\xi}{\xi-z},
  \qquad z\in\operatorname{Int}(\Gamma),
  \qquad H_b(z):=h_b(z)R_b(z).
\end{equation}
 The function $h_b$ is analytic in
$\operatorname{Int}(\Gamma)$ and satisfies
\begin{equation}\label{eq:rho-b-definition}
  \rho_b(x)=\frac{1}{2\pi}h_b(x)\sqrt{b^2-x^2},
  \qquad x\in(-b,b),
\end{equation}
and
\begin{equation}\label{eq:hb-PV}
  h_V^{(b)}(x)=\frac1\pi\dashint_{-b}^{b}
  \frac{V'(s)}{(s-x)\sqrt{b^2-s^2}}\,ds,
  \qquad x\in(-b,b).
\end{equation}

At the critical point,
\begin{equation}\label{eq:critical-density}
  \rho(x)=\psi(x)(A^2-x^2)^{\frac{4k+1}{2}},
  \qquad \psi>0\quad\hbox{on }[-A,A].
\end{equation}
Since $\psi$ is even and analytic, in a neighborhood of $A$ it has the
exact decomposition
\begin{equation}\label{eq:psi-decomp}
  \psi(z)=\psi(A)+\frac{\psi'(A)}{2A}(z^2-A^2)
  +(z^2-A^2)^2\chi(z),
\end{equation}
where $\chi$ is analytic.  Finally, let $Q(x):=V(x)-x^2$
and define
\begin{equation}\label{eq:P-definition}
  P(b)=\frac{1}{2\pi}\int_{-b}^{b}Q(x)h_b(x)\sqrt{b^2-x^2}\,\dd x,
\end{equation}
\begin{equation}\label{eq:S-U-definition}
  S(b)=\frac{1}{\pi}\int_{-b}^{b}Q(x)\sqrt{b^2-x^2}\,\dd x,
  \qquad U(b)=S(b)-P(b).
\end{equation}

\begin{lemma}\label{lem:hb}
Let $\delta=b-A$.  As $\delta\to0$,
\begin{equation}\label{eq:H-expansion}
  h_b(b)=H_0\delta^{2k}+H_1\delta^{2k+1}+\mathcal O(\delta^{2k+2}),
  \qquad
  \partial_z h_b(b)=\frac{4k}{3}H_0\delta^{2k-1}
  +\mathcal O(\delta^{2k}),
\end{equation}
where
\begin{equation}\label{eq:H0}
  H_0=2\pi\psi(A)\frac{(4k+1)!!}{(2k)!}A^{2k},
\end{equation}
and
\begin{equation}\label{eq:H1}
  H_1=H_0\left[\frac{k}{A}
  +\frac{4k+3}{2(2k+1)}\frac{\psi'(A)}{\psi(A)}\right].
\end{equation}
\end{lemma}

\begin{proof}
Let
\[
  g_A'(z):=\int_{-A}^{A}\frac{\rho(u)}{z-u}\,\dd u,
  \qquad z\in\mathbb C\setminus[-A,A].
\]
The differentiated Euler–Lagrange relation gives, initially in
$\mathcal V\setminus[-A,A]$,
\[
  F_A(z):=V'(z)-h_A(z)R_A(z)=2g_A'(z)=\frac{2}{z}+\mathcal O(z^{-2}),
  \qquad z\to\infty,
\]
where
\begin{equation}\label{eq:hA-critical}
  h_A(z)=2\pi\psi(z)(A^2-z^2)^{2k}.
\end{equation}
Thus $F_A$ has an analytic continuation to $\mathbb C\setminus[-A,A]$;
in particular, $\frac{F_A}{R_b}=\mathcal O(z^{-2})$ at infinity. 
Consequently,
\begin{equation}\label{eq:hb-reduced}
  h_b(z)=\frac{1}{2\pi i}\oint_\Gamma
  \frac{h_A(\xi)R_A(\xi)}{R_b(\xi)}\frac{d\xi}{\xi-z},
  \qquad z\in\operatorname{Int}(\Gamma),
\end{equation}
and differentiation with respect to the spectral variable gives
\begin{equation}\label{eq:hbprime-reduced}
  \partial_z h_b(z)=\frac{1}{2\pi i}\oint_\Gamma
  \frac{h_A(\xi)R_A(\xi)}{R_b(\xi)}\frac{d\xi}{(\xi-z)^2},
  \qquad z\in\operatorname{Int}(\Gamma).
\end{equation}

For $m\in\mathbb N_0$, set
\begin{equation}\label{eq:Phi-definition}
  \Phi_m(z;b):=\frac{1}{2\pi i}\oint_\Gamma
  \frac{(\xi^2-A^2)^{m+\frac12}}{R_b(\xi)}\frac{d\xi}{\xi-z},
  \qquad z\in\operatorname{Int}(\Gamma).
\end{equation}
As $z\to\infty$,
\begin{align}
  \frac{(z^2-A^2)^{m+\frac12}}{R_b(z)}
  &=R_b(z)^{2m}\left(1+\frac{b^2-A^2}{R_b(z)^2}\right)^{m+\frac12}\notag\\
  &=\sum_{j=0}^{m}\binom{m+\frac12}{j}(b^2-A^2)^j
  R_b(z)^{2m-2j}+E_m(z;b),
  \label{eq:model-expansion}
\end{align}
where $E_m(\,\cdot\,;b)$ is analytic in the exterior of $\Gamma$ and is
$\mathcal O(z^{-2})$ at infinity. Its Cauchy projection into
$\operatorname{Int}(\Gamma)$ is therefore zero, and hence
\begin{equation}\label{eq:Phi-polynomial}
  \Phi_m(z;b)=\sum_{j=0}^{m}\binom{m+\frac12}{j}(b^2-A^2)^j
  (z^2-b^2)^{m-j}.
\end{equation}
For $m\ge1$, evaluating \eqref{eq:Phi-polynomial} and its derivative at
$z=b$ yields
\begin{equation}\label{eq:model-value}
  \Phi_m(b;b)=\binom{m+\frac12}{m}(b^2-A^2)^m,
  \qquad
  \partial_z\Phi_m(b;b)=2b\binom{m+\frac12}{m-1}(b^2-A^2)^{m-1}.
\end{equation}
This, together with \eqref{eq:psi-decomp} and \eqref{eq:model-value}, we obtain
\begin{align}
  h_b(b)
  &=2\pi\psi(A)\binom{2k+\frac12}{2k}(b^2-A^2)^{2k}\notag\\
  &\quad+\frac{\pi\psi'(A)}{A}\binom{2k+\frac32}{2k+1}
  (b^2-A^2)^{2k+1}
  +\mathcal O\bigl((b^2-A^2)^{2k+2}\bigr),\label{eq:hb-value-expansion}\\
  \partial_z h_b(b)
  &=4\pi b\psi(A)\binom{2k+\frac12}{2k-1}(b^2-A^2)^{2k-1}
  +\mathcal O\bigl((b^2-A^2)^{2k}\bigr).
  \label{eq:hb-derivative-expansion}
\end{align}
Using $b^2-A^2=\delta(2A+\delta)$, we obtain 
\eqref{eq:H-expansion}–\eqref{eq:H1}.
\end{proof}

\begin{lemma}\label{lem:m3m4}
Recall
\begin{equation}\label{eq:M-def-appendix}
  M(b)=\int_{-b}^{b}\frac{xV'(x)}{\sqrt{b^2-x^2}}\,\dd x.
\end{equation}
As $\delta\to0$,
\begin{equation}\label{eq:M-expansion}
  M(A+\delta)=2\pi+\mu_0\delta^{2k+1}
  +\mu_1\delta^{2k+2}+\mathcal O(\delta^{2k+3}),
\end{equation}
where
\begin{equation}\label{eq:mu0}
  \mu_0=\frac{\pi A H_0}{2k+1}
  =\frac{2\pi^2\psi(A)A^{2k+1}}{2k+1}
  \frac{(4k+1)!!}{(2k)!},
\end{equation}
and
\begin{equation}\label{eq:mu1}
  \mu_1=\frac{\pi(AH_1+H_0)}{2k+2}
  =\frac{\pi^2A^{2k}}{k+1}\frac{(4k+1)!!}{(2k)!}
  \left[(k+1)\psi(A)+\frac{A(4k+3)}{2(2k+1)}\psi'(A)\right].
\end{equation}
In particular, $M^{(j)}(A)=0$ for $1\le j\le2k$.
\end{lemma}

\begin{proof}
Differentiation of the Cauchy representation of $h_b$ (equivalently, the
substitution $x=b\cos\theta$ followed by an integration by parts) gives the
standard endpoint identity
\begin{equation}\label{eq:M-prime-identity}
  M'(b)=\pi b h_b(b).
\end{equation}
Combining \eqref{eq:M-prime-identity} with \eqref{eq:H-expansion} gives
\begin{align*}
  M'(A+\delta)
  &=\pi(A+\delta)\left(H_0\delta^{2k}+H_1\delta^{2k+1}
  +\mathcal O(\delta^{2k+2})\right)\\
  &=\pi A H_0\delta^{2k}+\pi(AH_1+H_0)\delta^{2k+1}
  +\mathcal O(\delta^{2k+2}).
\end{align*}
Integration from $A$ to $A+\delta$, together with $M(A)=2\pi$, proves the
claim.
\end{proof}

\begin{lemma}\label{lem:bt-expansion}
Assume $A^2\neq2$; as $t\to0$, we have
\begin{equation}\label{eq:bt-expansion}
  b_t=A+b_1t^{\frac{1}{2k+1}}+b_2t^{\frac{2}{2k+1}}
  +\mathcal O\!\left(t^{\frac{3}{2k+1}}\right),
\end{equation}
where
\begin{equation}\label{eq:b1}
  b_1^{2k+1}=\frac{\pi(2-A^2)}{\mu_0} \quad \mbox{and}\quad  b_2=-\frac{b_1^2}{2k+1}
  \left(\frac{\mu_1}{\mu_0}+\frac{2A}{2-A^2}\right).
\end{equation}
If $A^2=2$, we have $b_t=A=\sqrt{2}$.
\end{lemma}

\begin{proof}
Put $\delta_t=b_t-A$.  Substitution of \eqref{eq:M-expansion} into the
endpoint equation
\[
  (1-t)M(b_t)+\pi t b_t^2=2\pi
\]
gives
\begin{equation}\label{eq:delta-equation}
  \mu_0\delta_t^{2k+1}+\mu_1\delta_t^{2k+2}
  +\pi(A^2-2)t+2\pi At\delta_t
  +\mathcal O(\delta_t^{2k+3})+\mathcal O(t\delta_t^2)=0.
\end{equation}
If $A^2=2$, from \eqref{eq:delta-equation}, we have $\delta_t=0$. If $A^2\neq2$, the first and third terms in
\eqref{eq:delta-equation} have nonzero coefficients and determine the dominant
balance.  Hence $\delta_t=\mathcal O(t^{\frac{1}{2k+1}})$.  Substitution of \eqref{eq:bt-expansion} into \eqref{eq:delta-equation} and comparison
of the coefficients of $t$ and $t^{\frac{2k+2}{2k+1}}$ gives respectively
\[
  \mu_0b_1^{2k+1}=\pi(2-A^2),
\]
and
\[
  (2k+1)\mu_0b_1^{2k}b_2+\mu_1b_1^{2k+2}+2\pi Ab_1=0.
\]
These two equations are equivalent to \eqref{eq:b1}.  Since
the exponent $2k+1$ is odd, \eqref{eq:b1} determines a unique real $b_1$;
this is precisely the branch selected by the one-cut endpoint $b_t$.
\end{proof}

\begin{lemma}\label{lem:P-expansion}
As $t\to0$, 
\begin{equation}\label{eq:P-expansion}
  P(b_t)=P(A)+P_1(b_t-A)^{2k+1}+P_2(b_t-A)^{2k+2}
  +\mathcal O\bigl((b_t-A)^{2k+3}\bigr),
\end{equation}
where
\begin{equation}\label{eq:p3-formula}
  P_1=\frac{\mu_0}{2\pi}J_A,\quad P_2=\frac{(2k+1)\mu_0(2-A^2)}{2\pi A(2k+2)}
  +\frac{\mu_1}{2\pi}J_A,
\end{equation}
with $ J_A:=\frac1\pi\int_{-A}^{A}\frac{Q(x)}{\sqrt{A^2-x^2}}\,\dd x.$ Particularly, if $A^2=2$, we have $P(b_t)=P(A)$.
\end{lemma}

\begin{proof}
Because $\Gamma$ is positively oriented and $H_{b,+}-H_{b,-}
=2i h_b\sqrt{b^2-x^2}$ on $(-b,b)$, contour deformation gives
\begin{equation}\label{eq:P-contour}
  P(b)=-\frac{1}{4\pi i}\oint_\Gamma Q(z)H_b(z)\,\dd z.
\end{equation}

To justify the $b$-derivative used below, observe that $\partial_bH_b$
is analytic on $\mathbb C\setminus[-b,b]$, is $\mathcal O(z^{-1})$ at
infinity, and has at most inverse-square-root endpoint singularities.  Since
$\partial_bR_b(z)=-\frac{b}{R_b(z)}$, its singular part at $z=b$ is
$-b \frac{h_b(b)}{R_b(z)}$.  The scalar one-cut relation and Liouville's theorem
therefore give
\begin{equation}\label{eq:dHb}
  \partial_bH_b(z)=-\frac{b h_b(b)}{R_b(z)},
  \qquad z\in\mathbb C\setminus[-b,b].
\end{equation}
Combining \eqref{eq:P-contour}, \eqref{eq:dHb}, and the boundary values in
\eqref{eq:Rb-boundary-values}, we obtain
\begin{equation}\label{eq:Pprime-contour}
  P'(b)=\frac{b h_b(b)}{2}\,J(b),
  \qquad
  J(b):=\frac1\pi\int_{-b}^{b}\frac{Q(x)}{\sqrt{b^2-x^2}}\,\dd x.
\end{equation}
After $x=b\cos\theta$,
\begin{equation}\label{eq:J-prime}
  J'(b)=\frac1{\pi b}\int_{-b}^{b}
  \frac{x\bigl(V'(x)-2x\bigr)}{\sqrt{b^2-x^2}}\,\dd x.
\end{equation}
Since $M(A)=2\pi$, \eqref{eq:J-prime} gives
\begin{equation}\label{eq:J-expansion}
  J(b)=J_A+\frac{2-A^2}{A}(b-A)+\mathcal O\bigl((b-A)^2\bigr).
\end{equation}
Substitution of \eqref{eq:H-expansion} and \eqref{eq:J-expansion} into
\eqref{eq:Pprime-contour} yields
\begin{align*}
  P'(b)
  &=\frac{A H_0}{2}J_A\delta^{2k}
  +\frac12\left[H_0(2-A^2)+(AH_1+H_0)J_A\right]\delta^{2k+1}
  +\mathcal O(\delta^{2k+2}).
\end{align*}
Integrating with respect to $\delta$ and using \eqref{eq:mu0}--\eqref{eq:mu1}
proves \eqref{eq:P-expansion}--\eqref{eq:p3-formula}.
\end{proof}

\begin{lemma}\label{lem:S-expansion}
As $t\to0$,
\begin{equation}\label{eq:S-expansion}
  S(b_t)=S(A)+S_1(b_t-A)+S_2(b_t-A)^2
  +\mathcal O\bigl((b_t-A)^3\bigr),
\end{equation}
where
\begin{equation}\label{eq:S1}
  S_1=A J_A \quad\mbox{and}\quad S_2=\frac12\bigl[J_A+(2-A^2)\bigr].
\end{equation}
Particularly, if $A^2=2$, we have $S(b_t)=S(A)$.
\end{lemma}

\begin{proof}
Differentiating \eqref{eq:S-U-definition} gives
\[
  S'(b)=bJ(b).
\]
Therefore $S'(A)=AJ_A$ and, by \eqref{eq:J-expansion}, we have $ S''(A)=J_A+(2-A^2),$ which completes the proof of the lemma.
\end{proof}

\begin{lemma}\label{lem:rho1-integral}
Let $\rho_1$ be defined in \eqref{definition: psij} and be the fixed-support measure corresponding to
\(V_1=x^2-V\).
Then $\rho_1$ is characterized by
\begin{equation}\label{eq:rho1-eq}
  2\dashint_{-A}^{A}\frac{\rho_1(u)}{x-u}\,\dd u=2x-V'(x),
  \qquad x\in(-A,A),
\end{equation}
and
\begin{equation}\label{eq:rho1}
  \rho_1(x)=\frac1\pi\sqrt{A^2-x^2}-\rho(x)
  +\frac{2-A^2}{2\pi\sqrt{A^2-x^2}},
  \qquad x\in(-A,A).
\end{equation}
Consequently
\begin{equation}\label{eq:iden-rho1}
  \int_{-A}^{A}\rho_1(x)Q(x)\,\dd x=C_1,
\end{equation}
where $C_1$ is defined in \eqref{eq:def-C1}.
\end{lemma}

\begin{proof}
We use
\begin{equation}\label{eq:semicircle-hilbert}
  \dashint_{-A}^{A}\frac{1}{\pi}
  \frac{\sqrt{A^2-u^2}}{x-u}\,\dd u=x,
  \qquad x\in(-A,A),
\end{equation}
and
\begin{equation}\label{eq:arcsine-hilbert}
  \dashint_{-A}^{A}\frac{1}{\pi\sqrt{A^2-u^2}}
  \frac{\dd u}{x-u}=0,
  \qquad x\in(-A,A).
\end{equation}
Since $2\dashint_{-A}^{A}\rho(u)(x-u)^{-1}\,\dd u=V'(x)$, the right-hand
side of \eqref{eq:rho1} satisfies \eqref{eq:rho1-eq}  and its mass is zero.

For uniqueness, the difference of two admissible solutions lies in the kernel
of the finite Hilbert transform \cite{Gakhov,Muskhelishvili1992}.  In the class of integrable functions with at
most inverse-square-root endpoint singularities, this kernel is spanned by
$(A^2-x^2)^{-\frac{1}{2}}$.  The zero-mass condition forces its coefficient to vanish.
Thus \eqref{eq:rho1} holds.  Multiplication by $Q$ and the definition of $C_1$
in \eqref{eq:def-C1} give \eqref{eq:iden-rho1}.
\end{proof}

\section{Uniform asymptotics of the $P_{\rm I}^{2k}$ parametrix}
\begin{lemma}[Uniform large-$s$ expansion of the $P_{\rm I}^{2k}$ model]
\label{lem:uniform-P2k-Psi}
For $\boldsymbol\tau=(\tau_1,\ldots,\tau_{2k-1})$ in a compact subset of
$\mathbb R^{2k-1}$, there exist constants $L,C>0$ such that
\begin{equation}\label{eq:uniform-third-term}
  \Psi(\zeta;s,\boldsymbol\tau)
  =\zeta^{-\frac{\sigma_3}{4}}N\left[I-h(s,\boldsymbol\tau)\sigma_3\zeta^{-\frac{1}{2}}
  +\mathcal O\!\left((1+|s|)^{\frac{4k+4}{2k+1}}\zeta^{-1}\right)\right]
  e^{-\theta(\zeta;s,\boldsymbol\tau)\sigma_3},
\end{equation}
uniformly for all real $s$, for $\zeta\notin\Gamma$, and whenever
$|\zeta|\ge L(1+|s|)^{\frac{4k+4}{2k+1}}$.  The implicit constant is bounded by
$C$ uniformly for $\boldsymbol\tau$ in that compact set.
\end{lemma}

\begin{proof}
    We follow the steepest descent analysis in \cite{Claeys}. Put
    \begin{equation}\label{eq:scaling}
        \eta=\frac{\zeta}{r},
       \qquad r=|s|^\frac{1}{2k+1}.
\end{equation}
    Following the transformations in \cite{Claeys}
    \[
        \Psi \longmapsto Y \longmapsto S \longmapsto R,
\]
where the global parametrix is
\begin{equation}\label{eq:Pinf}
        P^{(\infty)}(\eta)
        =|s|^{-\frac{1}{4(2k+1)}\sigma_3}
        (\eta-z_0)^{-\frac14\sigma_3}N,
\end{equation}
and a local parametrix $P(\zeta)$ is used inside \(U\).  The final error matrix is
\begin{equation}\label{eq:Rdef}
R(\eta)=
\begin{cases}
S(\eta)P(\eta)^{-1},& \eta\in U,\\[2pt]
S(\eta)P^{(\infty)}(\eta)^{-1},& \eta\in \mathbb{C} \setminus U.
\end{cases}
\end{equation}
The jumps of \(R\) satisfy
\begin{equation}\label{eq:JR}
        J_R(\eta)=I+O(|s|^{-1})\quad\hbox{on }\partial U,
        \qquad
        J_R(\eta)=I+O(e^{-c|s|})\quad\hbox{elsewhere},
\end{equation}
with \(c>0\), uniformly for fixed $\tau_j$, $j=1,...,2k-1$.  
We now have, for $|\zeta|>Lr$
\begin{equation}\label{eq:identity for Psi}
    TH\Psi(\zeta)e^{G(\zeta)\sigma_3}N^{-1}\zeta^{\frac{1}{4}\sigma_3}=R\left(\frac{\zeta}{r}\right)P^{(\infty)}\left(\frac{\zeta}{r}\right)N^{-1}\zeta^{\frac{1}{4}\sigma_3},
\end{equation}
where
\begin{equation}
    H=\begin{pmatrix}
        1 & 0\\
        -h & 1
    \end{pmatrix},\qquad T=\begin{pmatrix}
        1 & 0\\
        d_1|s|^\frac{1}{4k+2} & 1
    \end{pmatrix}.
\end{equation}
And $P^{(\infty)}$ and $R$ satisfy
\begin{equation}\label{eq:asy Pinfty}
   P^{(\infty)}\left(\frac{\zeta}{r}\right)N^{-1}\zeta^{\frac{1}{4}\sigma_3}=I+\frac{rz_0}{4\zeta}\sigma_3+\bigO\left(\frac{(rz_0)^2}{\zeta^2}\right),\quad z_0=\bigO(1), \quad \zeta\to\infty.
\end{equation}
\begin{equation}\label{eq:asy R}
    R\left(\frac{\zeta}{r}\right)=I+\bigO\left(|s|^{-\frac{2k}{2k+1}}\zeta^{-1}\right),\qquad |\zeta|>Lr
\end{equation}
At the same time, we have for any fixed $\tau_j$, $j=1,...,2k-1$,
\begin{equation}\label{eq:asy h and q}
    h(s,\tau_1,...,\tau_{2k-1})=\bigO(|s|^\frac{2k+2}{2k+1}),\qquad q(s,\tau_1,...,\tau_{2k-1})=\bigO(r),\quad |s|\to\infty,
\end{equation}
and as $|s|\to\infty$,
\begin{equation}\label{eq:asy G}
    G(\zeta;s,\tau_1,...,\tau
    _{2k-1})-\theta(\zeta;s,\tau_1,...,\tau
    _{2k-1})=\gamma_1\zeta^{-\frac{1}{2}}+\bigO\left(|s|^\frac{2k+3}{2k+1}\zeta^{-\frac{3}{2}}\right),\qquad \gamma_1=\bigO(|s|^\frac{2k+2}{2k+1}),
\end{equation}
uniformly for all $|\zeta|>L|s|^\frac{4k+4}{2k+1}$.

Putting \eqref{eq:asy Pinfty}, \eqref{eq:asy R}, \eqref{eq:asy h and q} and \eqref{eq:asy G} into \eqref{eq:identity for Psi}, we obtain \eqref{eq:uniform-third-term}. This completes the proof of the lemma.
\end{proof}

\section{Airy model RH problem}
\begin{itemize}
\item[(a)] $\Phi_{\mathrm{Ai}} : \mathbb{C} \setminus \Sigma_{A} \rightarrow \mathbb{C}^{2 \times 2}$ is analytic, and $\Sigma_{A}$ is shown in Figure \ref{figAiry}.
\item[(b)] $\Phi_{\mathrm{Ai}}$ has the jump relations
\begin{equation}\label{jumps P3}
\Phi_{\mathrm{Ai},+}(z) = \Phi_{\mathrm{Ai},-}(z) \times \left\{ \begin{array}{l l}\begin{pmatrix}
0 & 1 \\ -1 & 0
\end{pmatrix}, & \mbox{ on } \mathbb{R}^{-}, \\
\begin{pmatrix}
 1 & 1 \\
 0 & 1
\end{pmatrix}, & \mbox{ on } \mathbb{R}^{+}, \\
\begin{pmatrix}
 1 & 0  \\ 1 & 1
\end{pmatrix}, & \mbox{ on } e^{ \frac{2\pi i}{3} }  \mathbb{R}^{+} , \\
\begin{pmatrix}
 1 & 0  \\ 1 & 1
\end{pmatrix}, & \mbox{ on }e^{ -\frac{2\pi i}{3} }\mathbb{R}^{+} . \\
\end{array} \right.
\end{equation}
\item[(c)] As $z\to \infty$, $z \notin \Sigma_{A}$, we have
\begin{equation}\label{Asymptotics Airy}
\Phi_{\mathrm{Ai}}(z) = z^{-\frac{\sigma_{3}}{4}}M \left( I + \sum_{k=1}^{\infty} \frac{\Phi_{\mathrm{Ai,k}}}{z^{3k/2}} \right) e^{-\frac{2}{3}z^{\frac{3}{2}}\sigma_{3}},
\end{equation}
where $M = \frac{1}{\sqrt{2}}\begin{pmatrix}
1 & i \\ i & 1
\end{pmatrix}$ and $\Phi_{\mathrm{Ai,1}} = \frac{1}{8}\begin{pmatrix}
\frac{1}{6} & i \\ i & -\frac{1}{6}
\end{pmatrix}$.

As $z \to 0$, we have
\begin{equation}
\Phi_{\mathrm{Ai}}(z) = \bigO(1).
\end{equation} 
\end{itemize}
The Airy model RH problem was introduced and solved in \cite{DKMVZ1} (see also \cite[equation (7.30)]{DKMVZ1}, where explicit forms for the constant matrices $\Phi_{\mathrm{Ai,k}}$ can be found). We have
\begin{figure}[t]
    \begin{center}
    \setlength{\unitlength}{1truemm}
    \begin{picture}(100,55)(-5,10)
        \put(50,40){\line(1,0){30}}
        \put(50,40){\line(-1,0){30}}
        \put(50,39.8){\thicklines\circle*{1.2}}
        \put(50,40){\line(-0.5,0.866){15}}
        \put(50,40){\line(-0.5,-0.866){15}}
        \qbezier(53,40)(52,43)(48.5,42.598)
        \put(53,43){$\frac{2\pi}{3}$}
        \put(50.3,36.8){$0$}
        \put(65,39.9){\thicklines\vector(1,0){.0001}}
        \put(35,39.9){\thicklines\vector(1,0){.0001}}
        \put(41,55.588){\thicklines\vector(0.5,-0.866){.0001}}
        \put(41,24.412){\thicklines\vector(0.5,0.866){.0001}}
    \end{picture}
    \caption{\label{figAiry}The jump contour $\Sigma_{A}$ for $\Phi_{\mathrm{Ai}}$.}
\end{center}
\end{figure}
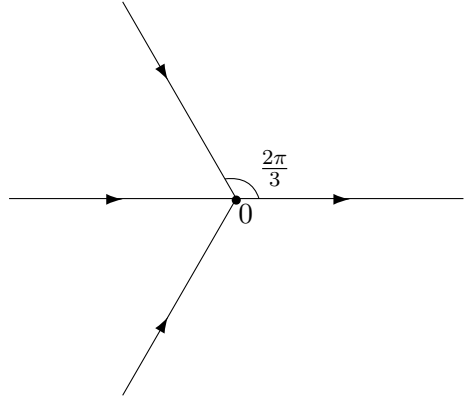
\begin{equation}
\Phi_{\mathrm{Ai}}(z) := M_{A} \times \left\{ \begin{array}{l l}
\begin{pmatrix}
\mbox{Ai}(z) & \mbox{Ai}(\omega^{2}z) \\
\mbox{Ai}^{\prime}(z) & \omega^{2}\mbox{Ai}^{\prime}(\omega^{2}z)
\end{pmatrix}e^{-\frac{\pi i}{6}\sigma_{3}}, & \mbox{for } 0 < \arg z < \frac{2\pi}{3}, \\
\begin{pmatrix}
\mbox{Ai}(z) & \mbox{Ai}(\omega^{2}z) \\
\mbox{Ai}^{\prime}(z) & \omega^{2}\mbox{Ai}^{\prime}(\omega^{2}z)
\end{pmatrix}e^{-\frac{\pi i}{6}\sigma_{3}}\begin{pmatrix}
1 & 0 \\ -1 & 1
\end{pmatrix}, & \mbox{for } \frac{2\pi}{3} < \arg z < \pi, \\
\begin{pmatrix}
\mbox{Ai}(z) & - \omega^{2}\mbox{Ai}(\omega z) \\
\mbox{Ai}^{\prime}(z) & -\mbox{Ai}^{\prime}(\omega z)
\end{pmatrix}e^{-\frac{\pi i}{6}\sigma_{3}}\begin{pmatrix}
1 & 0 \\ 1 & 1
\end{pmatrix}, & \mbox{for } -\pi < \arg z < -\frac{2\pi}{3}, \\
\begin{pmatrix}
\mbox{Ai}(z) & - \omega^{2}\mbox{Ai}(\omega z) \\
\mbox{Ai}^{\prime}(z) & -\mbox{Ai}^{\prime}(\omega z)
\end{pmatrix}e^{-\frac{\pi i}{6}\sigma_{3}}, & \mbox{for } -\frac{2\pi}{3} < \arg z < 0, \\
\end{array} \right.
\end{equation}
with $\omega = e^{\frac{2\pi i}{3}}$, Ai the Airy function and
\begin{equation}
M_{A} := \sqrt{2 \pi} e^{\frac{\pi i}{6}} \begin{pmatrix}
1 & 0 \\ 0 & -i
\end{pmatrix}.
\end{equation}

\section*{Acknowledgements}
Dan Dai was partially supported by grants from the Research Grants Council of the Hong Kong
Special Administrative Region, China [Project No. CityU 11311622, CityU 11306723 and CityU
11301924].


\begin{thebibliography}{99}



\bibitem{ACC2026}
Y. Ameur, C. Charlier, and J. Cronvall,
{Free energy and fluctuations in the random normal matrix
model with spectral gaps},
\textit{Constr. Approx.} \textbf{63} (2026), 279--335.





\bibitem{BWW} N. Berestycki, C. Webb, and M.D. Wong, Random Hermitian Matrices and Gaussian Multiplicative Chaos, \textit{Probab. Theory Related Fields} \textbf{172} (2018), no. 1-2, 103--189.



\bibitem{BI} P. Bleher and A. Its, Asymptotics of the partition function of a random matrix model, \textit{Ann. Inst. Fourier} \textbf{55} (2005), 1943--2000.






\bibitem{BG1} G. Borot and A. Guionnet, Asymptotic expansion of $\beta$-matrix models in the one-cut regime, \textit{Comm. Math. Phys.} \textbf{317} (2013), no. 2, 447--483.



\bibitem{BB} M.J. Bowick and E. Bre\'zin, Universal scaling of the tail of the density of eigenvalues in random matrix models, \textit{Phys. Lett. B} \textbf{268} (1991), 21--28.

\bibitem{BMP} E. Bre\'zin, E. Marinari and G. Parisi, A nonperturbative ambiguity free solution of a string model, \textit{Phys. Lett. B} \textbf{242} (1990), 35--38.

\bibitem{ByunKangSeoYang2025}
S.-S. Byun, N.-G. Kang, S.-M. Seo, and M. Yang,
{Free energy of spherical Coulomb gases with point charges},
\textit{J. Lond. Math. Soc.}  \textbf{112} (2025), no.~3, Paper No.~e70294.

\bibitem{ByunSeoYang2025}
S.-S. Byun, S.-M. Seo, and M. Yang,
{Free energy expansions of a conditional GinUE and large
deviations of the smallest eigenvalue of the LUE},
\textit{Comm. Pure Appl. Math.} \textbf{78} (2025), 2247--2304.

\bibitem{ByunYangYoo2026}
S.-S. Byun, M. Yang, and E. Yoo,
{Free energy expansion of determinantal Coulomb gases in the
quadratic fields with a point charge},
arXiv:2605.29594 (2026).







\bibitem{Charlier} C. Charlier, {Asymptotics of Hankel determinants with a one-cut regular potential and Fisher--Hartwig singularities}, \textit{Int. Math. Res. Not.} \textbf{2019} (2019), 7515--7576.

\bibitem{CFWW} C. Charlier, B. Fahs, C. Webb and M.D. Wong, {Asymptotics of Hankel determinants with a multi-cut regular potential and Fisher-Hartwig singularities}, \textit{Mem. Amer. Math. Soc.} \textbf{310} (2025), 1567.

\bibitem{CG}  C. Charlier and R. Gharakhloo, {Asymptotics of Hankel determinants with a Laguerre-type and Jacobi-type potential and Fisher--Hartwig singularities}, \textit{Adv. Math.} \textbf{383} (2021), no. 107672, 69 pp.






\bibitem{Claeys} T. Claeys, {Pole-free solutions of the first Painlev\'e hierarchy and non-generic critical behavior for the KdV equation}, \textit{Physica D} \textbf{241} (2012), 2226--2236.


\bibitem{CFLW} T. Claeys, B. Fahs, G. Lambert, and C. Webb, How much can the eigenvalues of a random Hermitian matrix fluctuate? \textit{Duke Math. J.} {\bf 170} (2021), no. 9, 2085--2235.

\bibitem{CG09} T. Claeys and T. Grava, Universality of the break-up profile for the KdV equation in the small dispersion limit using the Riemann–Hilbert approach, \textit{Comm. Math. Phys.} \textbf{286} (2009), 979--1009.

\bibitem{CG12} T. Claeys and T. Grava, {The KdV hierarchy: universality and a Painlev\'e transcendent}, \textit{Int. Math. Res. Not.} \textbf{2012} (2012), 5063--5099.

\bibitem{CGML} T. Claeys, T. Grava, and K.T.-R. McLaughlin, Asymptotics for the partition function in two-cut random matrix models, \textit{Comm. Math. Phys.} \textbf{339} (2015), no. 2, 513--587.

\bibitem{CK} T. Claeys and A.B.J. Kuijlaars, Universality of the double scaling limit in random matrix models, \textit{Comm. Pure Appl. Math.} \textbf{59} (2006), 1573--1603.

\bibitem{CKV} T. Claeys, A.B.J. Kuijlaars and M. Vanlessen, Multi-critical unitary random matrix ensembles and the general Painlev\'e II equation, \textit{Ann. of Math.} \textbf{167} (2008), 601--641.






\bibitem{ClaeysItsK} {T. Claeys, A. Its and I. Krasovsky, Higher-order analogues of the Tracy-Widom distribution and the Painlev\'e II hierarchy, \textit{Comm. Pure Appl. Math.} \textbf{63} (2010), 362--412.}

\bibitem{ClaeysVan} T. Claeys and M. Vanlessen, {Universality of a double scaling limit near singular edge points in random matrix models}, \textit{Comm. Math. Phys.} \textbf{273} (2007), 499--532.

\bibitem{ClaeysVan2} T. Claeys and M. Vanlessen, {The existence of a real pole-free solution of the fourth-order analogue of the Painlev\'e I equation}, \textit{Nonlinearity} \textbf{20} (2007), 1163--1184.



\bibitem{DLXYZ} D. Dai, W.-G. Long, S.-X. Xu, L.-M. Yao and L. Zhang, {The multiplicative constant in asymptotics of higher-order analogues of the Tracy–Widom distribution}, arXiv:2501.12679 (2025).

\bibitem{DMMS}
A. Dea{\~n}o, K.T.-R. McLaughlin, L. Molag and N. Simm,
{Asymptotics for a class of planar orthogonal polynomials and
truncated unitary matrices},
arXiv:2505.12633 (2025).

\bibitem{Deift} P. Deift, \emph{Orthogonal Polynomials and Random Matrices: A Riemann--Hilbert Approach}, Courant Lecture Notes in Mathematics, vol.~3, American Mathematical Society, 2000.





\bibitem{DeiKriMcL} P. Deift, T. Kriecherbauer and K.T.-R. McLaughlin, New results on the equilibrium measure for logarithmic potentials in the presence of an external field, \textit{J. Approx. Theory} \textbf{95} (1998), 388--475.

\bibitem{DIK} {P. Deift, A. Its and I. Krasovsky, {Asymptotics of Toeplitz, Hankel, and Toeplitz+Hankel determinants with Fisher--Hartwig singularities}, \textit{Ann. of Math.} \textbf{174} (2011), 1243--1299.}

\bibitem{DKMVZ1} P. Deift, T. Kriecherbauer, K.T.-R. McLaughlin, S. Venakides and X. Zhou, Strong asymptotics of orthogonal polynomials with respect to exponential weights, \textit{Comm. Pure Appl. Math.} {\bf 52} (1999), 1491--1552.

\bibitem{DKMVZ2} P. Deift, T. Kriecherbauer, K.T.-R. McLaughlin, S. Venakides and X. Zhou, Uniform asymptotics for polynomials orthogonal with respect to varying exponential weights and applications to universality questions in random matrix theory, \textit{Comm. Pure Appl. Math.} {\bf 52} (1999), 1335--1425.

\bibitem{DeiftZhou} P. Deift and X. Zhou, A steepest descent method for oscillatory Riemann--Hilbert problems: asymptotics for the mKdV equation, \textit{Ann. of Math.} \textbf{137} (1993), 295--368.

\bibitem{DeiftZhou1992} P. Deift and X. Zhou, {A steepest descent method for oscillatory Riemann--Hilbert problems}, \textit{Bull. Amer. Math. Soc.} \textbf{26} (1992), 119--123.

\bibitem{Dub06} B. Dubrovin, On Hamiltonian perturbations of hyperbolic systems of conservation laws, II: Universality of critical behavior, \textit{Comm. Math. Phys.} \textbf{267} (2006), 117--139.

\bibitem{Dub08} B. Dubrovin, On Universality of Critical Behavior in Hamiltonian PDEs, Geometry, Topology, and Mathematical Physics, American Mathematical Society Translation Series 2, vol. 224, American Mathematical Society, Providence, (2008), 59–109.

\bibitem{Dub09} B. Dubrovin, T. Grava and C. Klein, {On universality of critical behavior in the focusing nonlinear Schr\"odinger equation, elliptic umbilic catastrophe and the tritronqu\'ee solution to the Painlev\'e-I equation}, \textit{J. Nonlinear Sci.} \textbf{19} (2009), 57--94.



\bibitem{EML} N.M. Ercolani and K.T.-R. McLaughlin, Asymptotics of the partition function for random matrices via Riemann-Hilbert techniques and applications to graphical enumeration, \textit{Int. Math. Res. Not.} \textbf{2003} (2003), no. 14, 755--820.





\bibitem{FokasItsKitaev} A.S. Fokas, A.R. Its and A.V. Kitaev, The isomonodromy approach to matrix models in 2D quantum gravity, \textit{Comm. Math. Phys.} \textbf{147} (1992), 395--430.





\bibitem{Gakhov}{F. Gakhov, {Boundary Value Problems}, Pergamon Press, Oxford (1966). Reprinted by Dover Publications, New York (1990).}

\bibitem{Gordoa} R. Gordoa and A. Pickering, Nonisospectral scattering problems: a key to integrable hierarchies, \textit{J. Math. Phys.} \textbf{40} (1999), 5749--5786.



\bibitem{Grava-Kapaev-Klein-2015} T. Grava, A. Kapaev and C. Klein, On the tritronqu\'ee solution of $P_I^2$, \textit{Constr. Approx.} \textbf{41} (2015), 425--466

\bibitem{HastingsMcLeod} S.P. Hastings and J.B. McLeod, A boundary value problem associated with the second Painlev\'e transcendent and the Korteweg-de Vries equation, \textit {Arch. Rational Mech. Anal.} {\bf 73} (1980), 31-51.








\bibitem{Kud97} N.A. Kudryashov, The first and second Painlev\'{e} equations of higher order and some relations between them, \textit{Phys. Lett. A} \textbf{224} (1997), 353--360.

\bibitem{Kuijlaars_survey} A.B.J. Kuijlaars, \textit{Universality}, in {\it The Oxford Handbook of Random Matrix Theory}, 103--134, Oxford Univ. Press, Oxford, 2011.

\bibitem{KM} A.B.J. Kuijlaars and K. T.-R. McLaughlin, Generic behavior of the density of states in random matrix theory and equilibrium problems in the presence of real analytic external fields, \textit{Comm. Pure Appl. Math.} \textbf{53} (2000), 736--785.



\bibitem{KV2} A.B.J. Kuijlaars and M. Vanlessen, {Universality for eigenvalue correlations at the origin of the spectrum}, \textit{Comm. Math. Phys.} \textbf{243} (2003), 163--191.



\bibitem{Mehta} M.L. Mehta, {\it Random matrices}, Third Edition, \textit{Pure and Applied Mathematics Series} \textbf{142}, Elsevier Academic Press, 2004.

\bibitem{Mugan} U. Mugan and F. Jrad, Painlev\'{e} test and the first Painlev\'{e} hierarchy, \textit{J. Phys. A: Math. Gen.} 32 (1999), 7933--7952.

\bibitem{Muskhelishvili1992} N.I. Muskhelishvili, {\it Singular Integral Equations}, Dover Publications, New York, 1992.



\bibitem{SaTo} E.B. Saff and V. Totik, {\it Logarithmic Potentials with External Fields}, Springer-Verlag, 1997.

\bibitem{Shim04} R. Shimomura, A certain expression for the first Painlev\'{e} hierarchy, \textit{Proc. Japan Acad. Ser. A} \textbf{80} (2004), 105--109.



\bibitem{Szego} {G. Szeg\H{o}, \textit{Orthogonal Polynomials}, AMS Colloquium Publ. \textbf{23} (1959), New York: AMS.}

\bibitem{TW94} C. Tracy and H. Widom, Level spacing distributions and the Airy kernel, \textit{Comm. Math. Phys.} \textbf{159} (1994), 151--174.

\bibitem{XDZ} S.-X. Xu, D. Dai and Y.-Q. Zhao, Painlev\'e III asymptotics of Hankel determinants for a singularly perturbed Laguerre weight, \textit{J. Approx. Theory} \textbf{192} (2015), 1--18.

\bibitem{ZXZ} Z.-Y. Zeng, S.-X. Xu and Y.-Q. Zhao, Painlev\'e III asymptotics of Hankel determinants for a perturbed Jacobi weight, \textit{Stud. Appl. Math.} \textbf{135} (2015), 347--376.

\bibitem{ZhaoCD} Y. Zhao, L.H. Cao and D. Dai, Asymptotics of the partition function of a Laguerre-type random matrix model, \emph{J. Approx. Theory} \textbf{178} (2014), 64--90.



\end{thebibliography}
\end{document}